\documentclass[aps,prx,notitlepage,superscriptaddress,longbibliography,twocolumn,nofootinbib,floatfix]{revtex4-2} 

\usepackage[utf8]{inputenc}
\usepackage{graphicx}
\usepackage{hyperref}
\usepackage[table]{xcolor}
\usepackage{physics}
\usepackage{float}
\usepackage{svg}
\usepackage{hhline}
\usepackage{multirow}
\usepackage{interval}
\usepackage{array}
\usepackage{tabularx}
\usepackage{comment}
\usepackage[T1]{fontenc}

\intervalconfig {
soft open fences
}

\usepackage{tikz,tikz-3dplot}
\usetikzlibrary{calc, patterns, math}

\usepackage{amsmath,amsthm,amssymb,mathtools}
\usepackage{multirow}
\usepackage{braket}
\usepackage{bbm,bm}
\usepackage[bb=boondox]{mathalfa}
\usepackage{makecell}

\newcommand{\ie}{{\it i.e.,}\ }
\newcommand{\eg}{{\it e.g.,}\ }

\definecolor{Mathematica1}{rgb}{0.368, 0.507, 0.71}
\definecolor{Mathematica2}{rgb}{0.881, 0.611, 0.142}
\definecolor{Mathematica3}{rgb}{0.560, 0.692, 0.195}
\definecolor{Mathematica4}{rgb}{0.922, 0.386, 0.209}
\definecolor{Mathematica5}{rgb}{0.528, 0.471, 0.701}
\definecolor{yellow(ncs)}{rgb}{1.0, 0.94, 0.0}

\newcolumntype{M}[1]{>{\centering\arraybackslash}m{#1}}
\newcolumntype{P}[1]{>{\centering\arraybackslash}p{#1}}

\newtheorem{proposition}{Proposition}

\newtheorem{lemma}{Lemma}

\DeclareMathOperator{\erfc}{erfc}

\makeatletter 
    
\renewcommand\onecolumngrid{
\do@columngrid{one}{\@ne}%
\def\set@footnotewidth{\onecolumngrid}
\def\footnoterule{\kern-6pt\hrule width 1.5in\kern6pt}%
}

\renewcommand\twocolumngrid{
        \def\footnoterule{
        \dimen@\skip\footins\divide\dimen@\thr@@
        \kern-\dimen@\hrule width.5in\kern\dimen@}
        \do@columngrid{mlt}{\tw@}
}%

\makeatother    

\makeatletter
\def\amsbb{\use@mathgroup \M@U \symAMSb}
\makeatother

\begin{document}
\title{Typical entanglement entropy in systems with particle-number conservation}

\author{Yale Yauk}
\affiliation{School of Mathematics and Statistics, The University of Melbourne, Parkville, Victoria 3010, Australia}
\affiliation{School of Physics, The University of Melbourne, Parkville, Victoria 3010, Australia}
\affiliation{Perimeter Institute for Theoretical Physics, Waterloo, Ontario N2L 2Y5, Canada}
\affiliation{Department of Physics and Astronomy, University of Waterloo, Waterloo, Ontario N2L 3G1, Canada}

\author{Rohit Patil}
\affiliation{Department of Physics, The Pennsylvania State University, University Park, Pennsylvania 16802, USA}

\author{Yicheng Zhang}
\affiliation{Homer L. Dodge Department of Physics and Astronomy,\\
The University of Oklahoma, Norman, Oklahoma 73019, USA}
\affiliation{Center for Quantum Research and Technology, The University of Oklahoma, Norman, Oklahoma 73019, USA}

\author{Marcos Rigol}
\affiliation{Department of Physics, The Pennsylvania State University, University Park, Pennsylvania 16802, USA}

\author{Lucas Hackl}
\affiliation{School of Mathematics and Statistics, The University of Melbourne, Parkville, Victoria 3010, Australia}
\affiliation{School of Physics, The University of Melbourne, Parkville, Victoria 3010, Australia}

\begin{abstract}
We calculate the typical bipartite entanglement entropy $\braket{S_A}_N$ in systems containing indistinguishable particles of any kind as a function of the total particle number $N$, the volume $V$, and the subsystem fraction $f=V_A/V$, where $V_A$ is the volume of the subsystem. We expand our result as a power series $\braket{S_A}_N=a f V+b\sqrt{V}+c+o(1)$, and find that $c$ is universal (\ie independent of the system type), while $a$ and $b$ can be obtained from a generating function characterizing the local Hilbert space dimension. We illustrate the generality of our findings by studying a wide range of different systems, \eg bosons, fermions, spins, and mixtures thereof. We provide evidence that our analytical results describe the entanglement entropy of highly excited eigenstates of quantum-chaotic spin and boson systems, which is distinct from that of integrable counterparts. 
\end{abstract}

\maketitle

\clearpage
\section{Introduction}

Entanglement is widely regarded as one of the most important features of quantum theory. It describes quantum correlations that cannot be explained classically, and has become an important probe for physical properties in many areas of quantum physics. There exist a large number of different entanglement measures and witnesses~\cite{eisert2007quantitative}, of which the bipartite entanglement entropy is the most prominent one with broad applications ranging from quantum information processing~\cite{zheng2000efficient} and characterizing phases of matter~\cite{pollmann2010entanglement} to studying the black hole information paradox~\cite{page1993information} and holography~\cite{ryu2006aspects}. The behavior of the bipartite entanglement entropy of highly excited energy eigenstates has become a widely used probe in quantum many-body systems, including quantum-chaotic interacting models~\cite{mejia_05, santos_12, deutsch_li_13, beugeling_andreanov_15, yang_chamon_15, vidmar2017entanglement, dymarsky2018subsystem, garrison2018does, nakagawa_watanabe_18, liu2018quantum, lu_grover_19, murthy_19, huang2019universal, LeBlond_19, kaneko_iyoda_20, huang_21, Haque_Khaymovich_2022, kliczkowski2023average, nieva2023, patil2023}, integrable interacting models~\cite{alba09, moelter_barthel_14, LeBlond_19, patil2023}, quadratic models~\cite{storms_singh_14, lai_yang_15, nandy_sen_16, VidmarHackl_2017, vidmar2018volume, zhang_vidmar_18, liu2018quantum, HacklVidmar_2019, jafarizadeh_rajabpour_19, lydzba2020eigenstate, lydzba2021entanglement}, and systems with Hilbert space fragmentation~\cite{frey2023probing}. In all of them, the average eigenstate entanglement entropy satisfies a volume law, contrasting the typical area-law~\cite{eisert_colloquium_2010} found in ground states and low-excited states of locally interacting systems.

While over the years the average eigenstate entanglement entropy of physical Hamiltonians has been largely studied numerically, recently, there has been tremendous progress in understanding its behavior for large systems using different classes of (Haar-)random states. This includes general pure states~\cite{Page_1993, vidmar2017entanglement, garrison2018does, Bianchi_2022} and fermionic Gaussian states~\cite{lydzba2020eigenstate, lydzba2021entanglement, bianchi2021page, Bianchi_2022}, both with and without total particle-number conservation. Random matrix theory enabled these analytical calculations. For general pure states, they reproduce the correct leading volume-law term in the average eigenstate entanglement entropy of highly excited eigenstates of quantum-chaotic interacting models (the differences have been found to occur in the $O(1)$ term~\cite{huang_21, Haque_Khaymovich_2022, kliczkowski2023average, nieva2023}). For fermionic Gaussian states, however, they qualitatively reproduce the behavior of the leading volume-law term observed in translationally invariant integrable interacting and quadratic models~\cite{LeBlond_19, VidmarHackl_2017, HacklVidmar_2019}. There is mounting evidence (including evidence provided in this work) that the average eigenstate entanglement entropy enables one to discriminate between quantum-chaotic and integrable interacting systems~\cite{LeBlond_19, Bianchi_2022, patil2023}.

The notion of particles plays an important role in quantum theory. It can be related to an underlying $U(1)$ symmetry of the system. Models with $U(1)$ symmetry, which can also be spin models, can be described using particle-number preserving Hamiltonians. The effect of particle-number conservation on the average entanglement entropy of pure states has been explored before, but the focus has been on systems whose local Hilbert spaces are two-dimensional~\cite{vidmar2017entanglement, garrison2018does, liu2018quantum, bianchi2021page, Bianchi_2022}, which naturally describe spinless fermions, hard-core bosons, and spin-$\frac{1}{2}$ degrees of freedom. From the perspective of a Haar-random state with fixed particle number (or fixed total magnetization) all these systems are equivalent, such that they are all described by the same formulas for the average entanglement entropy and its variance. Consequently, prior to this work, it was not possible to identify which properties of the average entanglement entropy are universal and which depend on the specifics of the local Hilbert space dimension and how.

Our goal in this work is to address those questions in full generality. Therefore, we consider the most general case of a system with total particle-number conservation, for which we compute the typical entanglement entropy as a function of the subsystem size and the particle density. Such a system is purely characterized by the tensor product structure over local sites $\mathcal{H}=\otimes_i\mathcal{H}_i$, and the total particle-number operator $\hat{N}=\sum_i\hat{N}_i$ written as a sum over local number operators. No assumption is made about the Hamiltonian describing the system at that stage. In the second part of our work, we then show evidence that the analytically computed entropy reproduces the leading terms [greater than $O(1)$] in the typical eigenstate entanglement entropy of quantum-chaotic interacting Hamiltonians that commute with $\hat{N}$.

The general setup for the calculations of the entanglement entropy of pure states is as follows. Given a pure state $\ket{\psi}$ in a Hilbert space $\mathcal{H}=\mathcal{H}_A\otimes\mathcal{H}_B$ with subsystems $A$ and $B$, the bipartite entanglement entropy is defined as $S_A(\ket{\psi})=-\Tr (\hat \rho_A\ln \hat \rho_A)$, where $\hat \rho_A=\Tr_B(\ket{\psi}\bra{\psi})$ is the reduced state in subsystem $A$, which is obtained after tracing over $B$. Given a total particle-number operator $\hat{N}$, one can restrict to an eigenspace $\mathcal{H}^{(N)}\subset\mathcal{H}$ of the total particle-number operator $\hat{N}$ and compute the average $\braket{S_A}_N=\int S_A(\ket{\psi})d\mu_{\ket{\psi}}$ with respect to the Haar measure on $\mathcal{H}^{(N)}$.

For a system with local two-dimensional Hilbert space, \ie each site can either be empty or occupied by a single particle (or have an up or down spin-$\frac{1}{2}$), the average entanglement entropy was computed as~\cite{Bianchi_2022,vidmar2017entanglement}
\begin{align}\label{eq:n_lead}
\begin{split}
	\braket{S_A}_{N}&=-[n\ln n+(1-n)\ln(1-n)]\, f V\\
	&\phantom{=}-\sqrt{\frac{n(1-n)}{2\pi}}\left|\ln\left(\frac{1-n}{n}\right)\right|\delta_{f,\frac{1}{2}}\sqrt{V}\\
    &\phantom{=}+\frac{1}{2}\left[f+\ln(1-f)-\delta_{f,\frac{1}{2}}\delta_{n,\frac{1}{2}}\right]+o(1)\,,
\end{split}
\end{align}
where $f=\frac{V_A}{V}$ is the subsystem fraction. The leading order had been found in Ref.~\cite{garrison2018does}. This expression has a number of interesting features, such as the existence of a $\sqrt{V}$ correction at $f=1/2$, the independence of the $O(1)$ term from $n=N/V$ (except at half-filling), and the existence of Kronecker $\delta$s. The latter indicate points of nonuniform convergence, requiring further resolution through double scaling limits. By double scaling we mean that, at those points, otherwise $V$ independent quantities such as $f$ or $n$, or both, are treated as functions of $V$. The specific points of interest are then approached by taking the limit $V\to\infty$. Our goal is to understand which terms in this expression, if any, are universal, and how the nonuniversal terms are modified in systems with larger local Hilbert spaces.

The presentation is organized as follows: In Sec.~\ref{sec:analytical-results}, we derive the main analytical results, \ie the average and variance of the pure-state entanglement entropy. In Sec.~\ref{sec:applications}, we illustrate the generality of these results using simple examples and explain how the methods readily apply to boson, fermion, spin systems, and their mixtures. In Sec.~\ref{sec:physical}, we connect our analytical findings to concrete physical Hamiltonians with local Hilbert space dimensions greater than (the typically studied dimension of) two, namely, the spin-$1$ $XXZ$ model and the Bose-Hubbard model. We provide evidence that our analytical results describe the leading order terms of the typical eigenstate entanglement entropy for those local Hamiltonians when they are quantum chaotic, but not when they are integrable. We conclude in Sec.~\ref{sec:summary} with a summary and discussion of our results.

\section{Analytical results: Average and variance}\label{sec:analytical-results}

In this section we derive our main analytical results, namely, the average entanglement entropy [Eq.~\eqref{eq:entropy_N_asymptotic}], its equivalent with the resolved Kronecker $\delta$ functions that become continuous functions in double scaling limits [Eq.~\eqref{eq:EE-resolved}], and the leading order variance [Eq.~\eqref{eq:variance_asymptotic}]. The variance vanishes exponentially fast as $V$ increases, which means that the formulas obtained for the averages also describe the \emph{typical} entanglement entropy of Haar-random states with fixed particle-number density.

\subsection{Setup: System with fixed particle number}

We consider the general setting of a system with a set $\mathcal{S}=\left\lbrace 1, 2, \dots, V\right\rbrace$ of sites, which could be the sites of a $D$-dimensional hypercubic lattice with linear dimension $L$, in which case $V=L^D$. However, our findings apply to any graph with $V$ edges, so they can be used in the context of irregular (and even fractal) lattices. Each site is described by a local Hilbert space $\mathcal{H}_\text{loc}$ that is isomorphic throughout all sites. It decomposes into a direct sum over the number of $k$ indistinguishable particles that it can hold, so\footnote{For example, a system of spinless fermions would have $\mathcal{H}_{\text{loc}}=\mathcal{H}_{\text{loc}}^{(0)}\oplus \mathcal{H}_{\text{loc}}^{(1)}$, where $\mathcal{H}_{\text{loc}}^{(0)}=\mathrm{span}(\ket{0})$ and $\mathcal{H}_{\text{loc}}^{(1)}=\mathrm{span}(\ket{1})$.}
\begin{equation}
    \mathcal{H}_{\text{loc}}=\bigoplus_{k} \mathcal{H}_{\text{loc}}^{(k)}\,.
\end{equation}
The dimension of the Hilbert space, $a_k=\dim \mathcal{H}_{\text{loc}}^{(k)}$ is a nonnegative integer equal to the number of ways to place $k$ particles at the site (see Fig.~\ref{fig:lattice_sites_diagram}). We will make the rather mild assumption that for large $k$ this dimension scales at most exponentially, \ie
\begin{equation}\label{eq:scaling-assumption}
    a_k=\dim \mathcal{H}_\text{loc}^{(k)}=O(R^{-k})\quad \text{as}\quad k\to\infty\,,
\end{equation}
for some positive constant $0<R\leq 1$. Note that this assumption guarantees that the series $\sum^{\infty}_{k=0}a_kz^k$ has a radius of convergence equal to $R$. In fact, for most physical systems, this sequence of dimensions is finite, such that $a_k=0$ for $k>n_{\max}$, or is bounded from above by a fixed integer. In addition, one typically has $a_0=\dim \mathcal{H}_\text{loc}^{(0)}=1$ corresponding to a unique vacuum (zero particles at a site), but our method can also be used for degenerate vacua where $a_0>1$. In particular, we do not require that the sequence $(a_k)^{\infty}_{k=0}$ converges or it is bounded.

To each site labeled $1\leq i\leq V$, we assign an independent Hilbert space $\mathcal{H}_i$, which is a copy of the model Hilbert space $\mathcal{H}_\text{loc}$. We define the total particle-number operator of the system to be $\hat{N}=\sum_{i=1}^V \hat{N}_i$, where $\hat{N}_i$ has eigenvalue $N_i$ on $\mathcal{H}_i^{(N_i)}$.

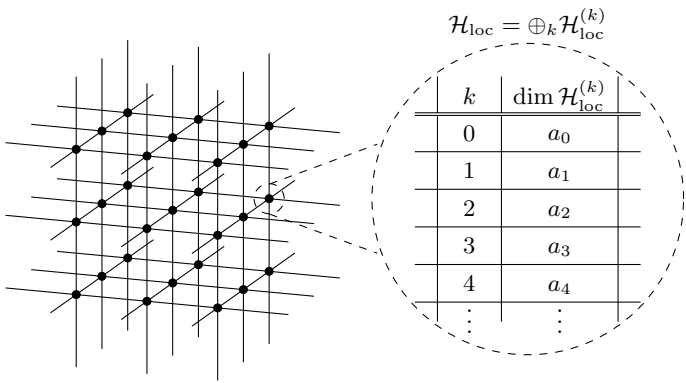
\begin{figure}
    \begin{center}
    \tdplotsetmaincoords{75}{20}
    \begin{tikzpicture}
        [cube/.style={very thick,black}]
    \begin{scope}[tdplot_main_coords]
    \foreach \l in {0,1,2} 
        \foreach \d in {0,1,2} {
        \draw (\l,-1,\d) -- (\l,3,\d);
        \draw (-1,\l,\d) -- (3,\l,\d);
        \draw (\l,\d,-1) -- (\l,\d,3);
    }
    \foreach \x in {0,1,2}
    \foreach \y in {0,1,2}
    \foreach \z in {0,1,2} {
    \filldraw[black,tdplot_main_coords] (\x,\y,\z) circle (1.5pt) node {};
    }
    \coordinate (X) at (2,2,1);
    \end{scope}
    \draw[dashed] (X) circle (0.2);
    \draw[dashed] (2.57,1.47) -- (4,2);
    \draw[dashed] (2.57,1.07) -- (4,0.55);
    \draw[dashed] (6,1.27) circle (2.07);
    \draw (5.65, 2.9) -- (5.65, -0.36);
    \draw (4.8, 2.9) -- (4.8, -0.36);
    \draw (7.2, 2.9) -- (7.2, -0.36);
    \draw[double] (4.5, 2.4) -- (7.5, 2.4); 
    \foreach \i in {0,1,2,3,4} {
    \draw (4.5, 1.9-0.5*\i) -- (7.5, 1.9-0.5*\i);
    }
    \node at (5.225, 2.65) {$k$};
    \node at (6.425, 2.65) {$\dim{\mathcal{H}_\text{loc}^{(k)}}$};
    \node at (6, 3.6) {$\mathcal{H}_\text{loc}=\oplus_{k}\mathcal{H}^{(k)}_{\text{loc}}$}; 
    \foreach \i in {0,...,4} {
    \node at (5.225, 2.15-0.5*\i) {$\i$};
    \node at (6.425, 2.1-0.5*\i) {$a_\i$};
    }
    \node at (5.225, 2.15-2.4) {$\vdots$};
    \node at (6.425, 2.15-2.4) {$\vdots$};
    \end{tikzpicture}
    \end{center}
    \vspace{-0.3cm}
    \caption{{\it Total vs local Hilbert spaces.} Illustration of how the total Hilbert space is constructed from a tensor product of local Hilbert spaces $\mathcal{H}_{\text{loc}}$, which themselves are direct sums over Hilbert spaces $\mathcal{H}^{(k)}_{\text{loc}}$ holding exactly $k$ particles each. The system is fully parameterized by the sequence of dimensions $a_k=\dim\mathcal{H}^{(k)}_{\text{loc}}$.}
    \label{fig:lattice_sites_diagram}
\end{figure}

We can decompose the total Hilbert space as
\begin{equation}
\mathcal{H}=\bigotimes_{i=1}^V\mathcal{H}_i=\bigoplus_{N}\mathcal{H}^{(N)}\,,
\end{equation}
\ie either as a tensor product over all individual sites or as direct sum over Hilbert spaces $\mathcal{H}^{(N)}$, in which we fix the total number of particles $N$. $\mathcal{H}^{(N)}$ is related to each local Hilbert space with fixed particle number by
\begin{equation}
    \mathcal{H}^{(N)}=\bigoplus_{N=\sum_{i=1}^V\! N_i}\,\bigotimes_{i=1}^V \mathcal{H}_i^{(N_i)}\,,
\end{equation}
where the direct sum is over the number of ways to distribute $N$ indistinguishable particles over $V$ distinguishable sites. Its dimension $d_N$ can be calculated as
\begin{equation}
    \label{eq:dN_exact_sum}
    d_N=\dim\mathcal{H}^{(N)}=\sum_{N=\sum_{i=1}^V\! N_i}\prod_{i=1}^V\, \underbrace{\dim\mathcal{H}_i^{(N_i)}}_{=a_{N_i}}\,.
\end{equation}
The exact expression for $d_N$ thus only depends on $V$, $N$, and the sequence $\lbrace a_k\rbrace^\infty_{k=0}$ of local Hilbert space dimensions. If the series truncates at $n_{\max}$, we must have $N\leq V n_{\max}$, as we can place at most $n_{\max}$ particles on each of the $V$ sites.

\subsection{Hilbert space dimension for large $V$}
\label{sec:combinatorics_of_dN}

We begin by introducing the generating function
\begin{equation}\label{eq:generating_function}
\zeta(z)=\sum_{k=0}^\infty a_k z^k\,,
\end{equation}
which is fully determined by the sequence $\lbrace a_k\rbrace^\infty_{k=0}$ of local Hilbert space dimensions. Using $\zeta(z)$, we can evaluate Eq.~\eqref{eq:dN_exact_sum} using the combinatorial identity
\begin{align}
    \left[ \zeta(z) \right]^V=\sum_{N}d_Nz^N\,,
\end{align}
\ie when expanding the l.h.s.~in powers of $z$, the coefficient in front of $z^N$ is exactly the quantity $d_N$ that we are looking for. Using the relation
\begin{align}
    \frac{1}{2\pi i}\oint_\Gamma z^n dz=\delta_{n,-1}\,,
\end{align}
we can extract this coefficient $d_N$ as
\begin{align}
    \label{eq:dN_raw_definition}
    d_N&=\frac{1}{2\pi i} \oint_\Gamma\underbrace{\frac{\left[ \zeta(z) \right]^V}{z^{N}}}_{=e^{V\psi(z)}}\frac{\dd{z}}{z}\,,
\end{align}
where $\Gamma$ is a simple closed contour around the origin. To evaluate $d_N$ asymptotically for large $V$ and fixed $n=N/V$, we use the saddle point approximation by rewriting the integrand as $e^{V\psi(z)}$, where\footnote{Note that $\psi$ has a branch cut along the negative real line (due to the logarithm having a jump of $2\pi i$ there), but since $\psi$ only appears in $e^{V\psi(z)}$ and $V$ is an integer, discontinuous jumps of $2\pi i$ do not affect the integral, because $e^{2\pi i V}=1$.}
\begin{equation}
    \label{eq:psi_zeta_index_definition}
    \psi(z)=\ln[\zeta(z)] -n\ln(z).
\end{equation}
Saddle-points of $\psi$ occur when $\psi'(z)=0$. This is equivalent to solving the equation
\begin{equation}
    \label{eq:psi_derivative_equals_zero}
    z\frac{\zeta'(z)}{\zeta(z)}=n\,.
\end{equation}
In Appendix~\ref{app:saddle_analysis} we show that, among all solutions of this equation, there is a unique positive real number $z_0(n)>0$ with $\psi'[z_0(n)]=0$ if $n>0$. We further show that for $V\to\infty$ this saddle point dominates the contour integral if we deform $\Gamma$ to a contour $\Gamma'$ that passes the real axis perpendicularly right at $z_0(n)$. The saddle point method then yields the asymptotic expansion
\begin{equation}\label{eq:saddle-asymp-sol}
    d_N=\frac{1}{z_0(n)\sqrt{2\pi \psi''[z_0(n)]V}}e^{V \psi[z_0(n)]}+o(1)\,.
\end{equation}
This result shows that no matter the structure of the local Hilbert space, the dimension of the total Hilbert space scales as 
\begin{equation}
    \label{eq:general_saddle_dN}
    d_N=\frac{\alpha(n)}{\sqrt{V}}e^{\beta(n) V}+o(1)\quad\text{for}\quad V\to\infty\,.
\end{equation}
We show in Appendix~\ref{app:saddle-point-analysis} that
\begin{align}\label{eq:solution-saddle}
    \beta(n)=\psi[z_0(n)]\quad\text{and}\quad \alpha(n)=\sqrt{\frac{-\beta''(n)}{2\pi}}\,,
\end{align}
satisfy $\beta(n)>0$ and $\beta''(n)<0$. Therefore, computing $\beta(n)$ fully determines the asymptotics of $d_N$.

In summary, for a system fully described by the sequence $\lbrace a_k\rbrace^\infty_{k=0}$, we can determine the asymptotics of $d_N$ [Eq.~\eqref{eq:general_saddle_dN}] by first finding the unique real solution $z_0(n)>0$ of Eq.~\eqref{eq:psi_derivative_equals_zero} and then computing $\beta(n)=\psi[z_0(n)]$. Even if there is no simple analytical solution, $z_0(n)$ and therefore $\beta(n)$ can be efficiently evaluated numerically. For finite $n_{\max}$, Eq.~\eqref{eq:psi_derivative_equals_zero} is equivalent to finding the unique positive root of a polynomial of degree $n_{\max}$.

\subsection{Average entanglement entropy}

Consider now a bipartition of the system into subsystems $A$ and $B$. Subsystem $A$ has a subset $\mathcal{S}_A\subset \mathcal{S}$ of sites. We denote the number of sites in $A$ as $\abs{\mathcal{S}_A}=V_A$. Subsystem $B$ then has $V-V_A$ sites given by $\mathcal{S}_B=\mathcal{S}\setminus \mathcal{S}_A$. The Hilbert space then decomposes as
\begin{equation}
    \label{eq:Hilbert_N_relation_Hilbert_A_Hilbert_B}
    \mathcal{H}^{(N)}=\bigoplus_{N_A} \mathcal{H}_A^{(N_A)}\otimes \mathcal{H}_B^{(N-N_A)}\,,
\end{equation}
where $N_A$ is the particle number in subsystem $A$. We also define $n_A=N_A/V$ as the particle density in subsystem $A$. It is bounded by
\begin{equation}
    \max[0, n-n_{\max}(1-f)]\leq n_A\leq\min[n, n_{\max}f]\,.
\end{equation}

When we focus on a subsystem of the entire system, only the scale changes. That is, the structures of $\mathcal{H}_A^{(N_A)}$ and $\mathcal{H}_B^{(N-N_A)}$ are similar to that of $\mathcal{H}^{(N)}$, except that one needs to replace the variables $(N, V)\to(N_A, V_A)$ for $A$ and $(N, V)\to(N-N_A, V-V_A)$ for $B$, respectively. In particular, the Hilbert space dimensions of the subsystems $A$ and $B$ can be found using Eq.~\eqref{eq:general_saddle_dN}, along with the changes $(n, V)\to (n_A/f,f V)$ and $(n, V)\to ([n-n_A]/[1-f], [1-f]V)$, respectively, such that
\begin{eqnarray}
    \label{eq:general_saddle_dA_dB}
        d_A&=&\dim\mathcal{H}_A^{(N_A)}=\sqrt{\frac{-\beta''\left(\frac{n_A}{f}\right)}{2\pi fV}}e^{\beta\left(\frac{n_A}{f}\right) fV}+o(1)\,, \nonumber \\
        d_B&=&\dim\mathcal{H}_B^{(N-N_A)} \\ &=& \sqrt{\frac{-\beta''\left(\frac{n-n_A}{1-f}\right)}{2\pi (1-f)V}}e^{\beta\left(\frac{n-n_A}{1-f}\right) (1-f)V}+o(1)\,.\nonumber
\end{eqnarray}

As shown in Ref.~\cite{Bianchi_2019}, the average entanglement entropy for fixed total particle number $N$ is then given by
\begin{align}
\label{eq:entropy_fixed_N_exact_sum}
\begin{split}
\braket{S_A}_N&=\sum_{N_A}\varrho_{N_A}\varphi_{N_A}\quad\text{with}\quad\varrho_{N_A}=\frac{d_Ad_B}{d_N}\,,\\
\varphi_{N_A}&=\Psi(d_N\!+\!1)-\Psi(\max[d_A,d_B]\!+\!1)\\
&\phantom{=}-\min\left[\frac{d_A-1}{2d_B},\frac{d_B-1}{2d_A}\right]\,,
\end{split}
\end{align}
where $\Psi(x)=\frac{\Gamma'(x)}{\Gamma(x)}=\dv{x} \ln\left[\Gamma(x)\right]$ is the digamma function.

While one can efficiently compute the sum in Eq.~\eqref{eq:entropy_fixed_N_exact_sum} numerically, we are more interested in its asymptotic behavior in the thermodynamic limit. To find it, it is useful to notice that the prefactor $\varrho_{N_A}=d_Ad_B/d_N$ is a probability distribution ($\sum_{N_A}\varrho_{N_A}=1$) that can be well approximated by a continuous Gaussian distribution $\varrho(n_A)$ in the rescaled variable $n_A=N_A/V$, such that $\varrho_{N_A}=V\varrho(n_A)+o(1)$. Similarly, we can define the continuous function $\varphi(n_A)=\varphi_{n_AV}$ with a ``kink'' at $n_A=n_{\mathrm{crit}}$, where $n_{\mathrm{crit}}$ is defined as the point $n_A=N_A/V$ with $d_A(N_A)=d_B(N-N_A)$. Together, this allows us to approximate the sum by a continuous integral
\begin{equation}
\label{eq:entropy_N_integral}
    \hspace{-2mm}\braket{S_A}_N=\sum_{N_A}\varrho_{N_A}\varphi_{N_A}=\int \varrho(n_A)\varphi(n_A)dn_A+o(1)\,.
\end{equation}
As shown in Ref.~\cite{supp},
\begin{eqnarray}
\varrho(n_A)&=&\sqrt{\frac{V|\beta''(n)|}{2\pi f(1-f)}}\exp\left[-\frac{V}{2}\frac{|\beta''(n)|(n_A-fn)^2}{f(1-f)}\right]\,,\nonumber \\ \label{eq:varrho_final}\\ 
    \varphi(n_A)&=&V\left[\beta(n)-(1-f)\beta\left(\frac{n-n_A}{1-f}\right)\right]\nonumber \\
    &&-\frac{1}{2}\delta_{f,\frac{1}{2}}\delta_{n,n^*}\exp\left[-\frac{4V|\beta^{(3)}(n^*)|}{3}\left|n_A-\frac{n^*}{2}\right|^3\right] \nonumber \\
    &&+\frac{1}{2}\ln\left[\frac{(1-f)\beta''(n)}{\beta''\left(\frac{n-n_A}{1-f}\right)}\right]+o(1)\,.\label{eq:varphi-expand}
\end{eqnarray}
for $f\leq 1/2$. We see that $\varrho(n_A)$ describes a Gaussian distribution with mean $\bar{n}_A=f n$ and variance $\sigma^2=f(1-f)/(\abs{\beta''(n)}V)$. The expression for $\varphi(n_A)$ is only valid for $n_A\leq n_{\mathrm{crit}}$. For $n_A\geq n_{\mathrm{crit}}$, we need to replace $n_A\to n-n_A$ and $f\to 1-f$. A subtlety arises from the term $\min[\tfrac{d_A-1}{2d_B},\tfrac{d_B-1}{2d_A}]$ in Eq.~\eqref{eq:entropy_fixed_N_exact_sum}, which gives rise to the Kronecker $\delta$ term in Eq.~\eqref{eq:varphi-expand}. This term is nonzero only if $n=n^*$, where $n^*$ is the point with $\beta'(n^*)=0$.

In the limit of large $V$, the Gaussian in Eq.~\eqref{eq:varrho_final} narrows since the standard deviation scales as $\sigma\sim1/\sqrt{V}$. Therefore, to calculate the integral to $O(1)$ in $V$, it suffices to Taylor expand $\varphi(n_A)$ up to quadratic order around the mean $\bar{n}_A=fn$. As discussed in Appendix~\ref{app:resolution_kronecker_delta}, the case $f=1/2$ is special, since the kink of $\varphi(n_A)$ lies exactly at the mean $\bar{n}_A$, \ie $n_{\mathrm{crit}}=\bar{n}_A$. Therefore, we need to integrate the regions $n_A\leq n_{\mathrm{crit}}$ and $n_A\geq n_{\mathrm{crit}}$ separately against the respective expressions of $\varphi(n_A)$, which produces a term proportional to $\sqrt{V}$. The integral involving the Kronecker $\delta$ term in Eq.~\eqref{eq:varphi-expand} is trivial, as the term is effectively constant around the peak of the Gaussian $\varrho(n_A)$, so the integral yields $-\frac{1}{2}\delta_{f,\frac{1}{2}}\delta_{n,n^*}$.

We now have all the ingredients to evaluate the integral in Eq.~\eqref{eq:entropy_N_integral}. Using the known moments of the Gaussian distribution to simplify the integral, we obtain~\cite{supp}
\begin{align}
    \label{eq:entropy_N_asymptotic}
    \begin{split}
        \braket{S_A}_N&=\beta(n) fV-\frac{\abs{\beta'(n)}}{\sqrt{2\pi\abs{\beta''(n)}}}\sqrt{V}\delta_{f,\frac{1}{2}}\\
        &\phantom{=}+\frac{1}{2}\left[f+\ln(1-f)-\delta_{f,\frac{1}{2}}\delta_{n,n^*}\right]+o(1)\,,
    \end{split}
\end{align}
valid for $f\leq1/2$ and with a unique $n^*>0$ computed from $\beta'(n^*)=0$.

It is remarkable that while the leading volume-law term depends on the exponential scaling of the dimension $d_N$, away from $f=1/2$ the correction is a universal $O(1)$ function of $f$, $[f+\ln(1-f)]/2$. For two-dimensional local Hilbert spaces, this $O(1)$ term was obtained in Ref.~\cite{vidmar2017entanglement} within a ``mean-field'' calculation. At $f=1/2$, there is an extra $O(1)$ term that has a universal $1/2$ prefactor, and depends on the specifics of the system being considered only through the Kronecker $\delta$ at $n=n^*$. We also find that, at $f=1/2$, the $\sqrt{V}$ correction (identified in Ref.~\cite{vidmar2017entanglement} in the context of two-dimensional local Hilbert spaces) appears generically and vanishes only at $n=n^*$ due to $\beta'(n^*)=0$. Hence, independent of the details of the system, we establish that the two terms in $\braket{S_A}_N$ that contain Kronecker $\delta$s are mutually exclusive. Another key finding of our work is that knowledge of the leading order term as a function of $n$ allows one, in principle, to calculate all terms up to $O(1)$ in $V$. Those are the nonvanishing terms in the thermodynamic limit, and are fully determined by $\beta(n)$.

\subsection{Resolving Kronecker $\delta$s}

The average entanglement entropy $\braket{S_A}_N$ as computed in Eq.~\eqref{eq:entropy_N_asymptotic} has the interesting feature that it contains Kronecker $\delta$s with respect to the continuous (in the thermodynamic limit) variables $f$ and $n$. This means that the respective expansion coefficients $b$ and $c$ (introduced in the abstract) containing these Kronecker $\delta$s do not converge uniformly to a real-valued function. Something interesting happens in the neighborhood of $f=1/2$ (for $b$) and at the point with $f=1/2$ and $n=n^*$ (for $c$). We can \emph{resolve} these Kronecker $\delta$s by considering a double scaling limit, in which $f\leq 1/2$ and $n$ are not assumed to be a fixed real value, but rather have their own scaling in $V$ (around $f=1/2$ and $n=n^*$, if $n^*$ is finite),
\begin{widetext}
\begin{align}\label{eq:EE-resolved}
    \braket{S_A}_N&=\beta(n) f V+\frac{1}{2}\left[f+\ln(1-f)\right]\\
    &\quad +V|f-\tfrac{1}{2}| \beta(n) \erfc\!\left[\sqrt{2V\abs{\beta''(n)}} \frac{|f-\frac{1}{2}|\beta(n)}{\abs{\beta'(n)}}\right]\! - \abs{\beta'(n)}\sqrt{\frac{V}{2\pi\abs{\beta''(n)}}} \exp\!\left[-2V\abs{\beta''(n)} \frac{\left(f-\frac{1}{2}\right)^2\beta(n)^2}{\abs{\beta'(n)}^2}\right]\nonumber\\
    &\quad -\frac{1}{4}\exp\left[\frac{(n-n^*)^2}{2}V\abs{\beta''(n^*)}\right]\left(e^{2\left(f-\frac{1}{2}\right)V\beta(n^*)}\erfc\left[\sqrt{\frac{\abs{\beta''(n^*)}V}{2}}\frac{(n-n^*)^2\abs{\beta''(n^*)}+2\left(f-\frac{1}{2}\right)\beta(n^*)}{\abs{(n-n^*)\beta''(n^*)}}\right] \right.\nonumber \\
    &\quad\left.\hspace{4.5cm}+e^{-2\left(f-\frac{1}{2}\right)V\beta(n^*)}\erfc\left[\sqrt{\frac{\abs{\beta''(n^*)}V}{2}}\frac{(n-n^*)^2\abs{\beta''(n^*)}-2\left(f-\frac{1}{2}\right)\beta(n^*)}{\abs{(n-n^*)\beta''(n^*)}}\right] \right)\,, \nonumber
\end{align}
\end{widetext}
where the second line resolves the Kronecker $\delta_{f,\frac{1}{2}}$ of order $\sqrt{V}$, while the third and fourth lines resolve the Kronecker $\delta_{f,\frac{1}{2}}\delta_{n,n^*}$ of $O(1)$.

This formula can also be used when approximating the Page curve as a whole, \ie when plotting $\braket{S_A}_N$ as a function of $V_A=1,\dots,V$ for large $V$ and finite $n=N/V$. The reason is that, in such a plot, we will always have $V_A$ near $V/2$, which corresponds to the double scaling limit $f-1/2=O(1/V)$.

In Appendix~\ref{app:resolution_kronecker_delta}, we analyze in detail the different double scaling limits $f=1/2+\Lambda_f/V^s$ and $n=n^*+\Lambda_n/V^t$, and find closed resolved expressions involving $\Lambda_f$ and $\Lambda_n$ around these critical points.

\subsection{Variance}

The variance $\left(\Delta S_A\right)^2_N\equiv\braket{S^{\,2}_A}_N-\braket{S_A}_N^2$ of the entanglement entropy for fixed $N$ is~\cite{Bianchi_2019}
    \begin{equation}
    \label{eq:variance_sum_definition}
    \left(\Delta S_A\right)^2_N=\frac{1}{d_N+1}\left[\sum_{N_A}\varrho_{N_A}\left(\varphi_{N_A}^2+\chi_{N_A}\right)-\braket{S_A}_N^2\right]\,,
    \end{equation}
with $\braket{S_A}_N$, $\varrho_{N_A}$, and $\varphi_{N_A}$ defined in Eq.~\eqref{eq:entropy_fixed_N_exact_sum}, and 
\begin{equation}
    \label{eq:chi_definition}
    \chi_{N_A}=\begin{cases}
        (d_A+d_B)\Psi'(d_B+1) & \\
        \,\,-(d_N+1)\Psi'(d_N+1)& d_A\leq d_B\\
        \,\,-\frac{(d_A-1)(d_A+2d_B-1)}{4d_B^2} &   \\[2mm]
        (d_A+d_B)\Psi'(d_A+1) & \\
        \,\,-(d_N+1)\Psi'(d_N+1)& d_A> d_B\\
        \,\,-\frac{(d_B-1)(d_B+2d_A-1)}{4d_A^2} & 
    \end{cases},
\end{equation}
where the function $\Psi'(x)=\dv[2]{x} \ln[\Gamma(x)]$ is the derivative of the digamma function. To shorten our equations, in what follows we drop the $n_A$ dependence of $\varrho$ and $\varphi$ and the differential $\dd{n_A}$, \eg $\int \varrho$ is understood to be $\int \varrho(n_A)\dd{n_A}$.

The last term $\braket{S_A}_N^2$ in Eq.~\eqref{eq:variance_sum_definition} is just the squared average, which we have calculated in Eq.~\eqref{eq:entropy_N_asymptotic}. The sum $\sum_{N_A=0}^N \varrho_{N_A}\varphi_{N_A}^2$ is evaluated with the methods of the previous section and yields
\begin{equation}
    \begin{split}
        \int \varrho\,\varphi^2&=\left(f[f+\ln(1-f)]\beta(n)+\frac{f(1-f)\beta'(n)^2}{\abs{\beta''(n)}}\right)V\\ 
        &\phantom{={}} +\beta(n)^2f^2V^2-\frac{\beta(n)\abs{\beta'(n)}}{\sqrt{2\pi\abs{\beta''(n)}}}V^{\frac{3}{2}}\delta_{f,\frac{1}{2}}+o(V)\,.
    \end{split}
\end{equation}
Then it follows that 
\begin{equation}
    \int \!\! \varrho\, \varphi^2 - \left(\int \!\! \varrho\,\varphi\right)\!^2\!= \! \left[f(1-f)- \! \frac{1}{2\pi}\delta_{f,\frac{1}{2}}\right] \! \frac{\beta'(n)^2}{\abs{\beta''(n)}}V+o(V).
\end{equation}
The behavior of $\chi_{N_A}$ as $V$ tends to infinity is obvious once we use the expansion $\Psi'(V)=1/V+O(1/V^2)$ in Eq.~\eqref{eq:chi_definition}. One can show that
\begin{equation}
    \chi_{N_A}=\begin{cases}
        \frac{d_A}{2d_B}+O(\frac{1}{d_B^2})\,, & d_A<d_B \\
        \frac{1}{4}+o(1)\,, & d_A=d_B \\ 
        \frac{d_B}{2d_A}+O(\frac{1}{d_A^2})\,, & d_A>d_B
    \end{cases}.
\end{equation}
Therefore, the term $\sum_{N_A=0}^N\varrho_{N_A}\chi_{N_A}$ vanishes unless $d_A=d_B$, which occurs for all $n_A$ only at $f=1/2$ and $n=n^*$, as discussed in Ref.~\cite{supp}. Then
\begin{equation}
    \sum_{N_A=0}^N \varrho_{N_A}\chi_{N_A}=\frac{1}{4}\delta_{f,\frac{1}{2}}\delta_{n,n^*}+o(1)\,.
\end{equation}
In fact, from Eq. \eqref{eq:variance_sum_definition}, we see that 
\begin{equation}
    \left(\Delta S_A\right)_N^2=\frac{1}{d_N\!+\!1}\left[f(1\!-\!f)\!-\!\frac{1}{2\pi}\delta_{f,\frac{1}{2}}\right]\frac{\beta'(n)^2}{\abs{\beta''(n)}}V+o(V).
\end{equation}
Plugging in the asymptotic form of $d_N$ from Eq.~\eqref{eq:general_saddle_dN}, we find that~\cite{supp}
\begin{align}
    \label{eq:variance_asymptotic}
    \left(\Delta S_A\right)_N^2&=\frac{\sqrt{2\pi}\beta'(n)^2}{\abs{\beta''(n)}^{\frac{3}{2}}}\left[f(1-f)-\frac{1}{2\pi}\delta_{f,\frac{1}{2}}\right]V^{\frac{3}{2}}e^{-\beta(n)V} \nonumber \\
    &\phantom{={}} +o(V^{\frac{3}{2}}e^{-\beta(n)V})\,.
\end{align}
We see that the variance vanishes in the thermodynamic limit, thus the average entanglement entropy is typical. That is, we expect the overwhelming majority of the quantum states with $N$ particles to have the entanglement entropy predicted by Eq.~\eqref{eq:entropy_N_asymptotic}.

\section{Applications}\label{sec:applications}

We illustrate the generality and elegance of our results by discussing various applications. Specifically, we consider five examples that showcase the full range of possible behaviors of the integer sequence $\{a_k\}^\infty_{k=0}$. We also explain how spin-$j$ systems with fixed $S_z$ can naturally be described within our framework, and how one can straightforwardly describe composite systems by combining different particle species.

\subsection{Five examples}

To illustrate how $\braket{S_A}_{N}$ depends on the particle density $n=N/V$ for different systems, we must specify the coefficients $a_k$ that fully characterize each system. This is usually straightforward. The following task is to determine $\beta(n)$, from which the entanglement entropy can be calculated using Eq.~\eqref{eq:entropy_N_asymptotic}. The outcomes of these steps are shown in Fig.~\ref{fig:bucket_numbers_common_systems} for five illustrative examples, which for simplicity involve spinless fermions and bosons.

\textbf{(a) Spinless fermions.} In this case no more than one particle may be placed at each site due to Pauli's exclusion principle. That is, $n_{\max}=1$ and either the site is filled or not. Thus
\begin{equation}
    a_k=\begin{cases}
        1 & k=0, 1 \\
        0 & \text{otherwise}
    \end{cases}.
\end{equation}
This case can be used to describe all local two-level quantum systems, including hard-core bosons and spin-$1/2$ systems. The vast majority of previous works on the typical entanglement entropy of pure states focused on this case (see Ref.~\cite{Bianchi_2022} for a review).

\textbf{(b) Two-species hard-core bosons.} Single-species hard-core bosons behave like spinless fermions in position space in that no two hard-core bosons can be placed at the same site. One can extend the hard-core constraint to two-species hard-core bosons, which is the example we have in mind here. For two species of hard-core bosons, though particles are still indistinguishable within each species, there are now two ways to place one hard-core boson on a site (one per species), and no more than a single hard-core boson can be placed on a site, so
\begin{equation}
    a_k=\begin{cases}
        1, & k=0 \\
        2, & k=1 \\ 
        0, & \text{otherwise}
    \end{cases}\,.
\end{equation}
In contrast to this case, for two species of fermions (\eg spin-$1/2$ fermions) one can place two fermions of different species in a site. Of course, if the two species of fermions have an infinite repulsion between them, which produces an effective hard-core constraint, then the results would be identical to those for two-species hard-core bosons.

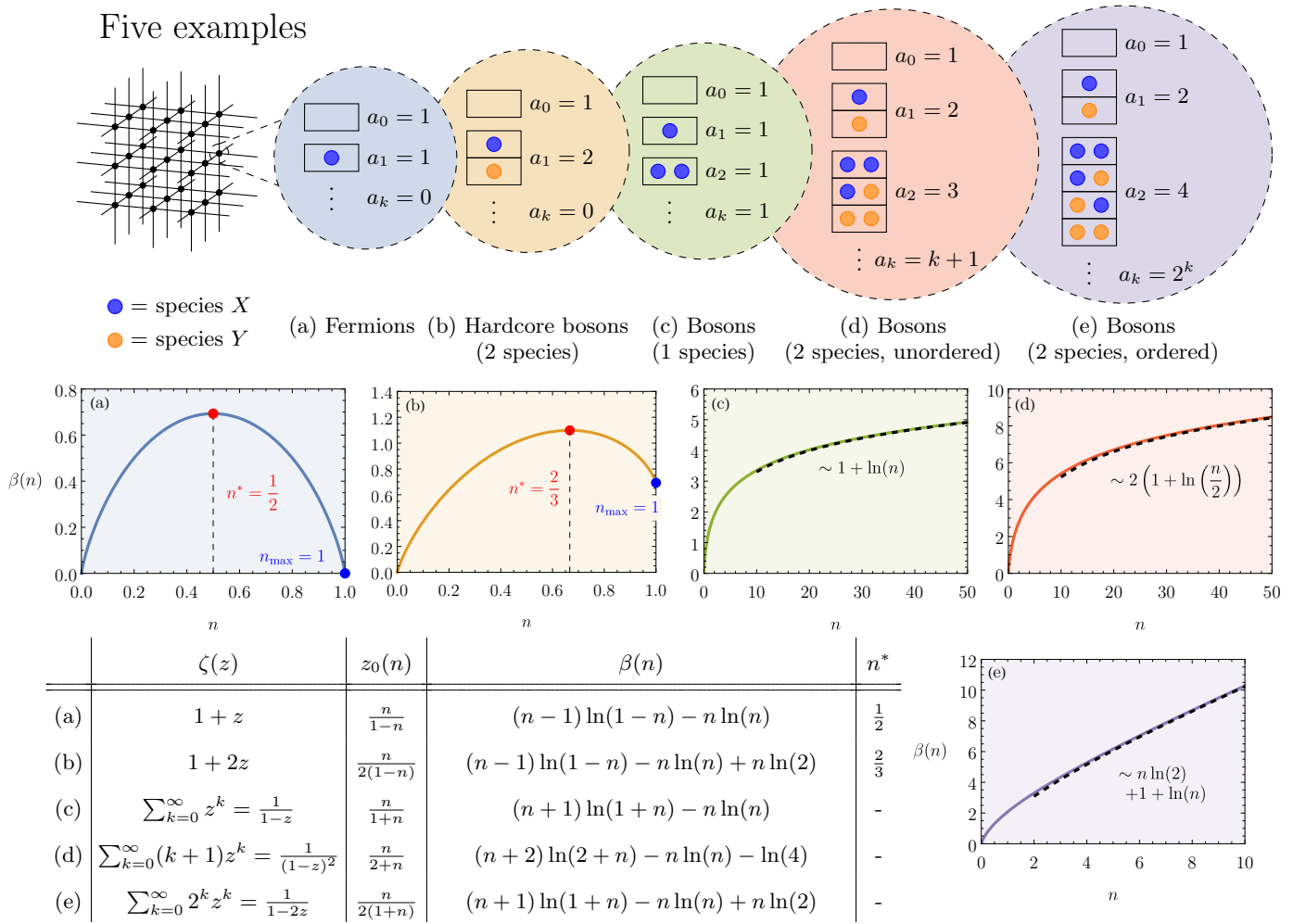
\begin{figure*}
    \makebox[\textwidth][c]{
        \tdplotsetmaincoords{75}{20}
        \tikzmath{\scale = 0.6;}
        \hspace{-10mm}
        \begin{tikzpicture}[scale=.97]
            \coordinate (f) at (2.8,1.1);
            \coordinate (q) at (5.2,1.3);
            \coordinate (b) at (7.8,1.5);
            \coordinate (tbu) at (10.6,2);
            \coordinate (tbo) at (14,2.2);
            \begin{scope}[tdplot_main_coords]
            \foreach \l in {0, \scale, 2*\scale} 
            \foreach \d in {0, \scale , 2*\scale} {
            \draw (\l,-\scale,\d) -- (\l,3*\scale,\d);
            \draw (-\scale,\l,\d) -- (3*\scale,\l,\d);
            \draw (\l,\d,-\scale) -- (\l,\d,3*\scale);
            }
            \foreach \x in {0, \scale, 2*\scale} 
            \foreach \y in {0, \scale, 2*\scale} 
            \foreach \z in {0, \scale, 2*\scale}  {
            \filldraw[black,tdplot_main_coords] (\x,\y,\z) circle (1.2pt) node {};
            }
            \coordinate (X) at (2*\scale,2*\scale,\scale);
            \end{scope}
            \draw[dashed] (X) circle (0.14);
            \draw[dashed] ($(X)+(0,0.14)$)--($(X)+(1.2,0.6)$);
            \draw[dashed] ($(X)+(0,-0.14)$)--($(X)+(1.2,-0.6)$);
            \draw[dashed,fill=Mathematica5!30] ($(tbo)+(0.9,-1.45)$) circle (2.2);
            \draw[dashed,fill=Mathematica4!30] ($(tbu)+(0.9,-1.25)$) circle (2.15);
            \draw[dashed,fill=Mathematica3!30] ($(b)+(0.9,-0.7)$) circle (1.6);
            \draw[dashed,fill=Mathematica2!30] ($(q)+(0.9,-0.5)$) circle (1.5);
            \draw[dashed,fill=Mathematica1!30] ($(f)+(0.9,-0.4)$) circle (1.35);
            \draw (f) rectangle ($(f)+(0.8, 0.4)$);
            \node at ($(f)+(1.4, 0.2)$) {$a_0=1$};
            \draw ($(f)+(0,-0.6)$) rectangle ($(f)+(0.8, -0.2)$);
            \draw[blue, fill=blue!70] ($(f)+(0.4, -0.4)$) circle (3pt);
            \node at ($(f)+(1.4, -0.4)$) {$a_1=1$};
            \node at ($(f)+(0.4,-0.9)$) {$\vdots$};
            \node at ($(f)+(1.4,-1)$) {$a_k=0$};
            \draw (b) rectangle ($(b)+(0.8, 0.4)$);
            \node at ($(b)+(1.4, 0.2)$) {$a_0=1$};
            \draw ($(b)+(0,-0.6)$) rectangle ($(b)+(0.8, -0.2)$);
            \draw[blue, fill=blue!70] ($(b)+(0.4, -0.4)$) circle (3pt);
            \node at ($(b)+(1.4, -0.4)$) {$a_1=1$};
            \draw ($(b)+(0,-1.2)$) rectangle ($(b)+(0.8, -0.8)$);
            \draw[blue, fill=blue!70] ($(b)+(0.225, -1)$) circle (3pt);
            \draw[blue, fill=blue!70] ($(b)+(0.575, -1)$) circle (3pt);
            \node at ($(b)+(1.4, -1)$) {$a_2=1$};
            \node at ($(b)+(0.4,-1.5)$) {$\vdots$};
            \node at ($(b)+(1.4,-1.6)$) {$a_k=1$};  
            \draw (tbu) rectangle ($(tbu)+(0.8, 0.4)$);
            \node at ($(tbu)+(1.4, 0.2)$) {$a_0=1$};
            \draw ($(tbu)+(0,-0.6)$) rectangle ($(tbu)+(0.8, -0.2)$);
            \draw[blue, fill=blue!70] ($(tbu)+(0.4, -0.4)$) circle (3pt);
            \draw ($(tbu)+(0,-1)$) rectangle ($(tbu)+(0.8, -0.6)$);
            \draw[orange, fill=orange!70] ($(tbu)+(0.4, -0.8)$) circle (3pt);
            \node at ($(tbu)+(1.4, -0.6)$) {$a_1=2$};
            \draw ($(tbu)+(0,-1.6)$) rectangle ($(tbu)+(0.8, -1.2)$);
            \draw[blue, fill=blue!70] ($(tbu)+(0.225, -1.4)$) circle (3pt);
            \draw[blue, fill=blue!70] ($(tbu)+(0.575, -1.4)$) circle (3pt);
            \draw ($(tbu)+(0,-2)$) rectangle ($(tbu)+(0.8, -1.6)$);
            \draw[blue, fill=blue!70] ($(tbu)+(0.225, -1.8)$) circle (3pt);
            \draw[orange, fill=orange!70] ($(tbu)+(0.575, -1.8)$) circle (3pt);
            \draw ($(tbu)+(0,-2.4)$) rectangle ($(tbu)+(0.8, -2)$);
            \draw[orange, fill=orange!70] ($(tbu)+(0.225, -2.2)$) circle (3pt);
            \draw[orange, fill=orange!70] ($(tbu)+(0.575, -2.2)$) circle (3pt);
            \node at ($(tbu)+(1.4, -1.8)$) {$a_2=3$};
            \node at ($(tbu)+(0.4,-2.7)$) {$\vdots$};
            \node at ($(tbu)+(1.4,-2.8)$) {$a_k=k+1$};
            \draw (tbo) rectangle ($(tbo)+(0.8, 0.4)$);
            \node at ($(tbo)+(1.4, 0.2)$) {$a_0=1$};
            \draw ($(tbo)+(0,-0.6)$) rectangle ($(tbo)+(0.8, -0.2)$);
            \draw[blue, fill=blue!70] ($(tbo)+(0.4, -0.4)$) circle (3pt);
            \draw ($(tbo)+(0,-1)$) rectangle ($(tbo)+(0.8, -0.6)$);
            \draw[orange, fill=orange!70] ($(tbo)+(0.4, -0.8)$) circle (3pt);
            \node at ($(tbo)+(1.4, -0.6)$) {$a_1=2$};
            \draw ($(tbo)+(0,-1.6)$) rectangle ($(tbo)+(0.8, -1.2)$);
            \draw[blue, fill=blue!70] ($(tbo)+(0.225, -1.4)$) circle (3pt);
            \draw[blue, fill=blue!70] ($(tbo)+(0.575, -1.4)$) circle (3pt);
            \draw ($(tbo)+(0,-2)$) rectangle ($(tbo)+(0.8, -1.6)$);
            \draw[blue, fill=blue!70] ($(tbo)+(0.225, -1.8)$) circle (3pt);
            \draw[orange, fill=orange!70] ($(tbo)+(0.575, -1.8)$) circle (3pt);
            \draw ($(tbo)+(0,-2.4)$) rectangle ($(tbo)+(0.8, -2)$);
            \draw[orange, fill=orange!70] ($(tbo)+(0.225, -2.2)$) circle (3pt);
            \draw[blue, fill=blue!70] ($(tbo)+(0.575, -2.2)$) circle (3pt);
            \draw ($(tbo)+(0,-2.8)$) rectangle ($(tbo)+(0.8,-2.4)$);
            \draw[orange, fill=orange!70] ($(tbo)+(0.225, -2.6)$) circle (3pt);
            \draw[orange, fill=orange!70] ($(tbo)+(0.575, -2.6)$) circle (3pt);
            \node at ($(tbo)+(1.4, -2)$) {$a_2=4$};
            \node at ($(tbo)+(0.4,-3.1)$) {$\vdots$};
            \node at ($(tbo)+(1.4,-3.2)$) {$a_k=2^k$};
            \draw (q) rectangle ($(q)+(0.8, 0.4)$);
            \node at ($(q)+(1.4, 0.2)$) {$a_0=1$};
            \draw ($(q)+(0,-0.6)$) rectangle ($(q)+(0.8, -0.2)$);
            \draw[blue, fill=blue!70] ($(q)+(0.4, -0.4)$) circle (3pt);
            \draw ($(q)+(0,-1)$) rectangle ($(q)+(0.8, -0.6)$);
            \draw[orange, fill=orange!70] ($(q)+(0.4, -0.8)$) circle (3pt);
            \node at ($(q)+(1.4, -0.6)$) {$a_1=2$};
            \node at ($(q)+(0.4,-1.3)$) {$\vdots$};
            \node at ($(q)+(1.4,-1.4)$) {$a_k=0$};
            \node at ($(f)+(0.7,-2.91)$) {(a) Fermions};
            \node[align=center] at ($(q)+(0.9,-3.3)$) {(b) Hardcore bosons\\ (2 species)};
            \node[align=center] at ($(b)+(0.9,-3.5)$) {(c) Bosons\\ (1 species)};
            \node[align=center] at ($(tbu)+(0.9,-4)$) {(d) Bosons \\ (2 species, unordered)};
            \node[align=center] at ($(tbo)+(0.9,-4.2)$) {(e) Bosons \\ (2 species, ordered)};
            \coordinate (label) at (0,-1.5);
            \draw[blue, fill=blue!70] (label) circle (3pt);
            \node[right] at (label) {\ $=$ species $X$};
            \draw[orange, fill=orange!70] ($(label)+(0,-0.5)$) circle (3pt);
            \node[right] at ($(label)+(0,-0.5)$) {\ $=$ species $Y$};
            \node at (1.3,2.6) {\Large{Five examples}};
        \coordinate (graph) at (1, -4.5);
        \begingroup 
        \def\arraystretch{1.8}
        \node at ($(graph)+(4.3,-4)$) {
            \begin{tabular}{M{.5cm}|M{3.5cm}|M{1cm}|M{6cm}|M{0.5cm}}
         & $\zeta(z)$ & $z_0(n)$ & $\beta(n)$ & $n^*$ \\
        \hhline{=|=|=|=|=}
        (a) & $1+z$ & $\frac{n}{1-n}$ & $(n-1)\ln(1-n)-n\ln(n)$ & $\frac{1}{2}$\\
        (b) & $1+2z$ & $\frac{n}{2(1-n)}$ & $(n-1)\ln(1-n)-n\ln(n)+n\ln(2)$ & $\frac{2}{3}$ \\
        (c) & $\sum_{k=0}^\infty z^k=\frac{1}{1-z}$ & $\frac{n}{1+n}$ & $(n+1)\ln(1+n)-n\ln(n)$ & - \\
        (d) & $\sum_{k=0}^\infty(k+1)z^k=\frac{1}{(1-z)^2}$ & $\frac{n}{2+n}$ & $(n+2)\ln(2+n)-n\ln(n)-\ln(4)$ & - \\
        (e) & $\sum_{k=0}^\infty 2^kz^k=\frac{1}{1-2z}$ & $\frac{n}{2(1+n)}$ & $(n+1)\ln(1+n)-n\ln(n)+n\ln(2)$ & -  \\
        \end{tabular}
        };
        \endgroup
        \node at ($(graph)+(13.3,-4)$) {\includegraphics[width=.28\textwidth]{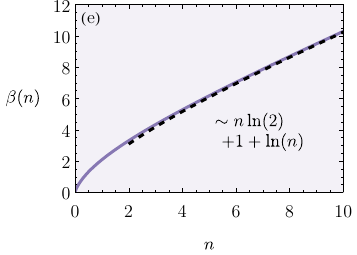}};
        \node at (graph) {\includegraphics[width=.28\textwidth]{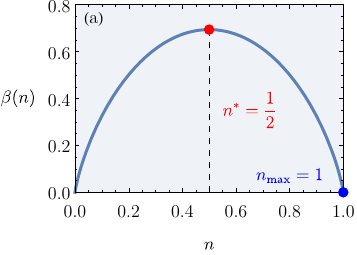}};
        \node at ($(graph)+(4.7,0)$) {\includegraphics[width=.28\textwidth]{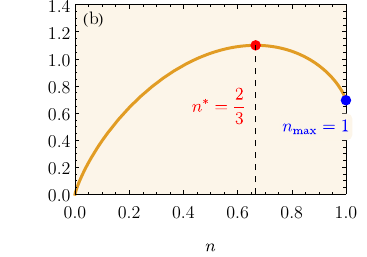}};
        \node at ($(graph)+(9.2,0)$) {\includegraphics[width=.28\textwidth]{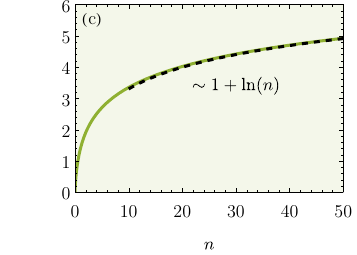}};
        \node at ($(graph)+(13.7,0)$){\includegraphics[width=.28\textwidth]{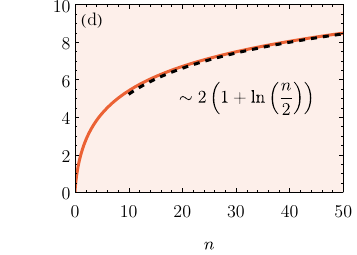}};
        \end{tikzpicture}
    }
    \caption{{\it Five illustrative examples.} Sequence $\left\{a_k\right\}_{k=0}^\infty$ for five examples involving spinless particles. Each connected block lists the possible arrangements of $k$ particles at a site (shown up to $k=2$). In the table, we list the generating function $\zeta(z)$, saddle point $z_0(n)$, and the exponential scaling coefficient $\beta(n)$. The latter is the volume-law coefficient $a$ [see Eq.~\eqref{eq:entropy_N_asymptotic}]. For each example, we plot the function $\beta(n)$ and label the points $n^*$ and $n_{\max}$, when they exist.}
    \label{fig:bucket_numbers_common_systems}
\end{figure*}

\textbf{(c) Bosons.} An arbitrary number of bosons can be placed on a site, and for single-species bosons we have
\begin{equation}
    a_k=1,\quad k\in\amsbb{N}\,.
\end{equation}
This is the first case considered in this work in which the local Hilbert space $\mathcal{H}_{\mathrm{loc}}$ is infinite dimensional, \ie in which $n_{\max}=\infty$. Consequently, the particle density $n$ can grow without bound and, as a result, the coefficient of the volume law can diverge (logarithmically) with $n$.

\textbf{(d) Two-species bosons.} Say we have two species $X$ and $Y$ of indistinguishable bosons. If a site contains $k$ bosons, it could have $0$ through $k$ bosons of type $X$. The number of $Y$ bosons would then be $k$ minus the number of $X$ bosons, which implies that
\begin{equation}
    a_k=k+1, \quad k\in\amsbb{N}\,.
\end{equation}
While the behavior of the integer sequence $\{a_k\}^\infty_{k=0}$ is different from case (c), we note that in both cases the coefficient $\beta(n)$ of the volume in the entanglement entropy is proportional to $ \ln n$ for large $n$.

\textbf{(e) Two-species bosons with ordering.} The difference with (d) is that in this case we care about the ordering of the bosons. For a site with $k$ bosons, there are a total of $2^k$ different ways in which we can place the $X$ and $Y$ bosons. Hence,
\begin{equation}
    a_k=2^k, \quad k\in\amsbb{N}\,.
\end{equation}
One could think of this example as describing a system in which there is an infinite number of internally ordered levels at each site, each of which can either be occupied by an $X$ or a $Y$ boson. This is a rather exotic example that illustrates how difficult it is to encounter systems that match the fastest (exponential) growth of $a_k$ allowed by our assumption in Eq.~\eqref{eq:scaling-assumption}, let alone surpass it. For the current example, we find that the volume-law coefficient grows linearly with $n$. This is the fastest growth with $n$ allowed by our framework.

\subsection{Spin systems}

While we phrased everything in the language of particles and the particle number $\hat{N}$, our formalism equally applies to general spin-$j$ systems. The generating function for the latter systems is
\begin{align}
    \zeta(z)=1+z+\dots+z^{2j}\,.
\end{align}
The total particle-number operator $\hat N$ and the total magnetization operator $\hat{M}=\sum_i\hat{S}_i^z$ are related by
\begin{align}\label{eq:relation_N_M}
    \hat{N}=jV+\hat{M},
\end{align}
where $\hat{S}^z_i$ is the local spin operator along the $z$-direction. The maximal particle density is then $n_{\max}=2j$. The dimension $d_N=\dim\mathcal{H}^{(N)}$ in this context can be understood as a generalization of the binomial coefficients, because we have
\begin{align}
    (1+z+\dots+z^{n_{\max}})^V=\sum_N d_N z^N\,,
\end{align}
where $d_N$ would be a regular binomial for $n_{\max}=1$. In general, we can refer to $d_N$ as the $(n_{\max}+1)$-nomial coefficients, also known as extended binomial coefficient. Their asymptotics and properties have been studied in various contexts~\cite{euler1801evolutione, comtet1974advanced, andrews1975theorem, neuschel2014note, li2014asymptotic, neuschel-personal} and there even exists the closed-form sum
\begin{align}
    d_N(V)=\sum_{k}(-1)^k\binom{V}{k}\binom{V+N-k(n_{\max}+1)-1}{V-1}\,.
\end{align}

\begin{figure}[t!]
    \centering
    \begin{tikzpicture}
        \draw (0,0) node{\includegraphics[width=\linewidth]{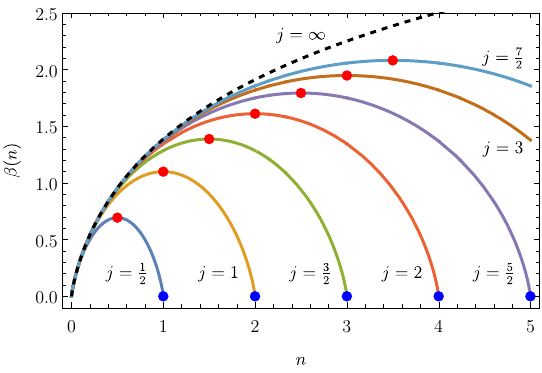}};
    \end{tikzpicture}
    \caption{{\it Spin systems.} $\beta(n)$ vs $n$ for spin-$j$ systems with different values of $j$. The red dots mark the maxima $\beta(n^*)=\ln(1+2j)$ with $n^*=\frac{n_{\max}}{2}=j$, and the blue dots mark $n_{\max}=2j$. Note that, with increasing $j$, the curves approach the $j=\infty$ result plotted as a dashed line.}
    \label{fig:spin}
\end{figure}

To find the asymptotics of $d_N$, we apply saddle point equation~\eqref{eq:solution-saddle}, where we need to find the unique positive root $z_0>0$ of the polynomial
\begin{align}
    z_0\,\zeta'(z_0)-n\, \zeta(z_0)=\sum^{n_{\max}}_{k=0}(k-n)z^k_0=0\,.
\end{align}
For $n_{\max}\leq 4$, there exist (increasingly cumbersome) closed expressions for the solutions, while for $n_{\max}\geq 5$ the solution can be efficiently evaluated numerically. In particular, we have
\begin{align}
    n_{\max}&=1:\quad z_0(n)=\frac{n}{1-n}\,,\\
    n_{\max}&=2:\quad z_0(n)=\frac{\sqrt{1+6n-3n^2}+n-1}{2(2-n)}\,,\label{eq:spin-1-z0}\\
    n_{\max}&=\infty: \quad z_0(n)=\frac{n}{1+n}\,,
\end{align}
where the first case is equivalent to spinless fermions and the last one to spinless bosons. We show $\beta(n)$ for different values of $j$ in Fig.~\ref{fig:spin}. In all cases with finite $n_{\max}$, $\beta(n)$ is symmetric under the $n\leftrightarrow n_{\max}-n$ swap, with $n^*=n_{\max}/2=j$ and $\beta(n^*)=\ln(1+n_{\max})=\ln(1+2j)$. For fixed $n$, we can take the large spin limit $j\to\infty$, in which $\beta(n)$ approaches the expression for bosons, \ie case (c) from Fig.~\ref{fig:bucket_numbers_common_systems}.

\subsection{Systems with general $U(1)$ charges}
Let us emphasize that our formalism can be equally applied to systems with general $U(1)$ charges that are multiples of some elementary charge $q$, such as particles with electrical charges (which are multiples of the elementary charge $e$). To accommodate negative charges, the generating function $\zeta(z)$ takes the form:
\begin{align}
    \zeta(z)=\sum^{k_{\max}}_{k=k_{\min}}a_k z^k\,,
\end{align}
where $a_k=\dim\mathcal{H}^{(k)}_{\mathrm{loc}}$ describes the dimension of the Hilbert space containing states with charge $kq$, and $k_{\min}$ can take negative values. In most situations, we expect charge conjugation symmetry, such that $k_{\min}=-k_{\max}$ and $a_k=a_{-k}$, but technically this is not required by the formalism. We further note that the entropy diverges if $k_{\min}=-\infty$ and $k_{\max}=\infty$, as this implies that there are infinitely ways to pair states with positive and negative charge to get any finite total charge. For $k_{\min}>-\infty$, one can always map this setting to that of the fixed total particle number by redefining $k\to k-k_{\min}$, which then starts at $k=0$ and runs up to $n_{\max}=k_{\max}-k_{\min}$.

\subsection{Composite systems}
It is straightforward to use our framework to describe composite systems with different species of particles. In that case, the total particle-number operator $\hat{N}=\sum_i\hat{N}_{(i)}$ corresponds to a sum over the total particle-number operators of different species $i$. As emphasized before, the description of such a system is completely determined once the sequence $\{a_k\}^\infty_{k=0}$ is known or, equivalently, once the generating function in Eq.~\eqref{eq:generating_function} is specified. For a composite system in which the local Hilbert space of each particle type is characterized by the sequences $\{a_k^{(i)}\}^\infty_{k=0}$, we can immediately compute
\begin{align}
    \zeta^\text{tot}(z)=\prod_i \zeta^{(i)}(z)=\prod_i\sum_k a^{(i)}_k z^k\,,
\end{align}
from which we can then determine $\beta(n)$ as described in Eq.~\eqref{eq:solution-saddle}.

\textbf{Example: Spinless fermions (a) and bosons (c).} For this combination, we find
\begin{align}
\begin{split}
    \zeta^\text{tot}(z)&=(1+z)(1+z+z^2+\dots)\\
    &=1+2z+2z^2+\dots=\frac{1+z}{1-z},
\end{split}
\end{align}
leading to $z^\text{tot}_0(n)=\frac{\sqrt{1+n^2}-1}{n}$, and
\begin{align}
    \beta^\text{tot}(n)=\operatorname{arcsinh}(n)-n\ln\frac{\sqrt{1+n^2}-1}{n}\,.
\end{align}

If we consider $m$ species of the same type of particles, the function $\zeta^\text{tot}_m(z)$ has the form
\begin{align}
    \zeta^\text{tot}_m(z)=[\zeta(z)]^m\,,
\end{align}
where $\zeta(z)$ is the single-species generating function. The resulting version of the saddle point Eq.~\eqref{eq:psi_derivative_equals_zero} then gives
\begin{align}
    \beta^{\text{tot}}_m(n)=m\beta(\tfrac{n}{m})\,,
\end{align}
where $\beta(n)$ is the single-species result.

\textbf{Example: spin-$\frac{1}{2}$ fermions.} This corresponds to two-species spinless fermions so
\begin{align}
    \beta^{\text{tot}}_2(n)=2 \left[-\frac{n}{2}\ln\left(\frac{n}{2}\right)-\left(1-\frac{n}{2}\right)\ln\left(1-\frac{n}{2}\right)\right]\,.
\end{align}

\section{Physical Hamiltonians}\label{sec:physical}

Next, we explore to which degree our analytical expressions describe the entanglement entropy of highly excited eigenstates of quantum-chaotic Hamiltonians. It is important to stress that there is no randomness in the Hamiltonians considered in this work. They are translationally invariant Hamiltonians with local interactions. We consider an extended spin-1 $XXZ$ model, which has $U(1)$ symmetry and as such the total magnetization is conserved, and the Bose-Hubbard model with or without an occupancy constraint in the lattice sites. The averages in this section, in contrast to those in Sec.~\ref{sec:analytical-results}, are carried out over a fixed number of eigenstates of those Hamiltonians.

\subsection{Spin-1 $XXZ$ model}\label{sec:Spin-1XXZ}

We focus first on the extended spin-1 $XXZ$ model (with anisotropy $\Delta$) in chains with $V$ sites, with Hamiltonian
\begin{eqnarray}
     H&=& H_0+\lambda H_1\,, \label{eq:Spin1_Hamiltonian} \\
    H_0&=& -\sum_{i=1}^V \hat{S}^x_i \hat{S}^x_{i+1} + \hat{S}^y_i \hat{S}^y_{i+1} +\Delta\hat{S}^z_i \hat{S}^z_{i+1} \,, \nonumber \\
    H_1&=& \sum_{i=1}^V (\hat{\vec{S}}_i\cdot \hat{\vec{S}}_{i+1})^2 
     -  \mu[2(\hat{S}^z_i)^2-(\hat{S}^z_i \hat{S}^z_{i+1})^2]\, \nonumber \\
     &&\hspace{0.5cm}-\nu[(\hat{S}^x_i \hat{S}^x_{i+1} + \hat{S}^y_i \hat{S}^y_{i+1})\hat{S}^z_i \hat{S}^z_{i+1} + \text{H.c.}] \,,\nonumber
\end{eqnarray}
where $\mu=\Delta-1\,,\, \nu=2-\sqrt{2(1+\Delta)}$, and $\hat{\vec{S}}_i=(\hat{S}^x_i,\hat{S}^y_i,\hat{S}^z_i)$ is the spin-$1$ operator at site $i$, and we consider periodic boundary conditions. For $\lambda=1$, this model is an integrable generalization of the spin-1 $XXZ$ model (also known as the Zamolodchikov-Fateev model)~\cite{Zamolodchikov_1980,Bytsko_2003}. For $\lambda=0$, unlike the spin-$\frac{1}{2}$ $XXZ$ model, this model is quantum-chaotic independently of the value of $\Delta$ (we set $\Delta=0.55$ to be in the maximally chaotic regime, following the discussion in Ref.~\cite{kliczkowski2023average}). The extended spin-1 $XXZ$ model allows us to probe the effect that the $U(1)$ symmetry and quantum chaosversusintegrability have on the entanglement entropy of highly excited energy eigenstates beyond the usually considered spin-$\frac{1}{2}$ case~\cite{LeBlond_19, Bianchi_2022}. 

The total magnetization $\hat M=\sum_{i=1}^V \hat{S}^z_i$ is a conserved quantity of the Hamiltonian~\eqref{eq:Spin1_Hamiltonian}. $M$ in our spin-1 model plays the role that the total particle number $N$ plays in a corresponding particle model, with $N=M+V$. The magnetization per site $m=M/V$ in the spin-1 model is the equivalent of the particle filling fraction $n=N/V$ in the particle model, with $n=m+1$. An example of a corresponding particle model is that of indistinguishable bosons with the constraint that at most two bosons may occupy a lattice site. We consider such a case in Sec.~\ref{sec:BoseHubbard} in the context of the Bose-Hubbard model. 

The generating function for our three-dimensional local Hilbert space is
\begin{equation}\label{eq:zetaspin-1}
\zeta(z)=1+z+z^2\,.
\end{equation}
The asymptotic form of the entanglement entropy follows from Eq.~\eqref{eq:spin-1-z0}. It is given by Eq.~\eqref{eq:entropy_N_asymptotic} with $n^*=1$ and
\begin{align}
     \beta(n) &= (n-2)\ln(2-n)+(n-1)\ln(2) \nonumber\\ 
     &\phantom{=}+\ln(7-3n+\sqrt{1-3n(n-2)})  \nonumber\\  
    &\phantom{=}-n\ln(n-1+\sqrt{1-3n(n-2)})\,. \label{eq:Beta_Spin-1} 
\end{align}

\begin{figure}[!t]
    \includegraphics[width=0.98\columnwidth]{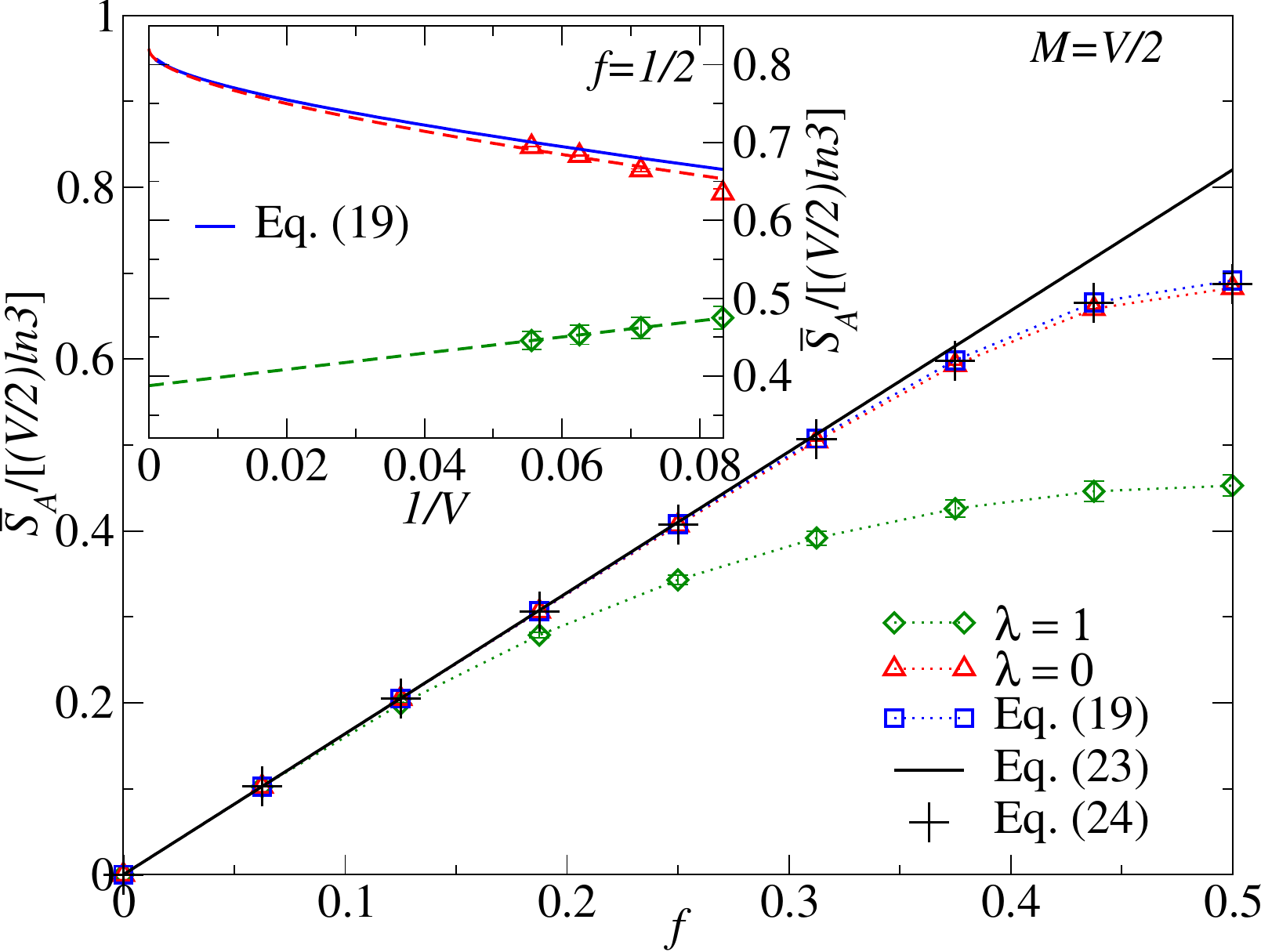}
    \caption{{\it Page curve for the extended spin-1 $XXZ$ model~\eqref{eq:Spin1_Hamiltonian} within the $M=V/2$ sector.} (Main panel) $\bar{S}_A$ vs the subsystem fraction $f$ for the quantum-chaotic ($\lambda=0$) and integrable ($\lambda=1$) Hamiltonian eigenstates for $V=16$. We also report $\braket{S_A}_N$ from the exact sum in Eq.~\eqref{eq:entropy_fixed_N_exact_sum}, as well as the leading order prediction for $\braket{S_A}_N$ from Eq.~\eqref{eq:entropy_N_asymptotic} and the double-scaling Kronecker-$\delta$-resolved expression for $\braket{S_A}_N$ from Eq.~\eqref{eq:EE-resolved} [with $\beta(n)$ from Eq.~\eqref{eq:Beta_Spin-1}], where $N=M+V$. (Inset) $\bar{S}_A$ vs $1/V$ at $f=1/2$ for the quantum-chaotic and integrable Hamiltonian eigenstates. The dashed lines following the numerical results are fits to the last three data points. For the quantum-chaotic case we use a single-parameter fit to $d_1+d_2/\sqrt{V}+p/V$, where $d_1$ and $d_2$ are set by Eq.~\eqref{eq:entropy_N_asymptotic} with $\beta(n)$ from Eq.~\eqref{eq:Beta_Spin-1}, and $p$ is our fitting parameter. For the integrable case we use a two-parameter fit to $p_1+p_2/V$. The continuous line shows $\braket{S_A}_N$ from the exact sum in Eq.~\eqref{eq:entropy_fixed_N_exact_sum}. The error bars in the numerical results are the standard deviation of the averages.}
    \label{fig:Spin-1_m=1/2}
\end{figure}

We compute the average entanglement entropy $\bar{S}_A$ of the highly excited eigenstates of the Hamiltonian~\eqref{eq:Spin1_Hamiltonian} in two magnetization sectors, $M=0$ and $M=V/2$. In our calculations we resolve all the symmetries of the Hamiltonian. Within each magnetization sector, translational invariance allows us to carry out the diagonalization within the total quasimomentum $k$ sectors, with $k\in\{2\pi \ell/V\, |\, \ell=-V/2+1,\, -V/2+2,\, \ldots,\, V/2\}$. The $k=0$ and $\pi$ sectors are further split into two subsectors (even and odd) under space reflection symmetry $P$. Furthermore, the $M=0$ sector exhibits an additional symmetry on top of the translational and space reflection symmetry, namely, the spin reflection symmetry $Z_2$. For each $M$, we use full exact diagonalization to obtain the 100 mid-spectrum energy eigenstates within each symmetry-resolved sector labeled by the applicable symmetries $\in\{k,P,Z_2\}$. Unless otherwise specified, the average entanglement entropy $\bar{S}_A$ reported for each $M$ is computed by taking the properly weighted average over all the symmetry sectors.

\begin{figure}[!t]
    \includegraphics[width=0.98\columnwidth]{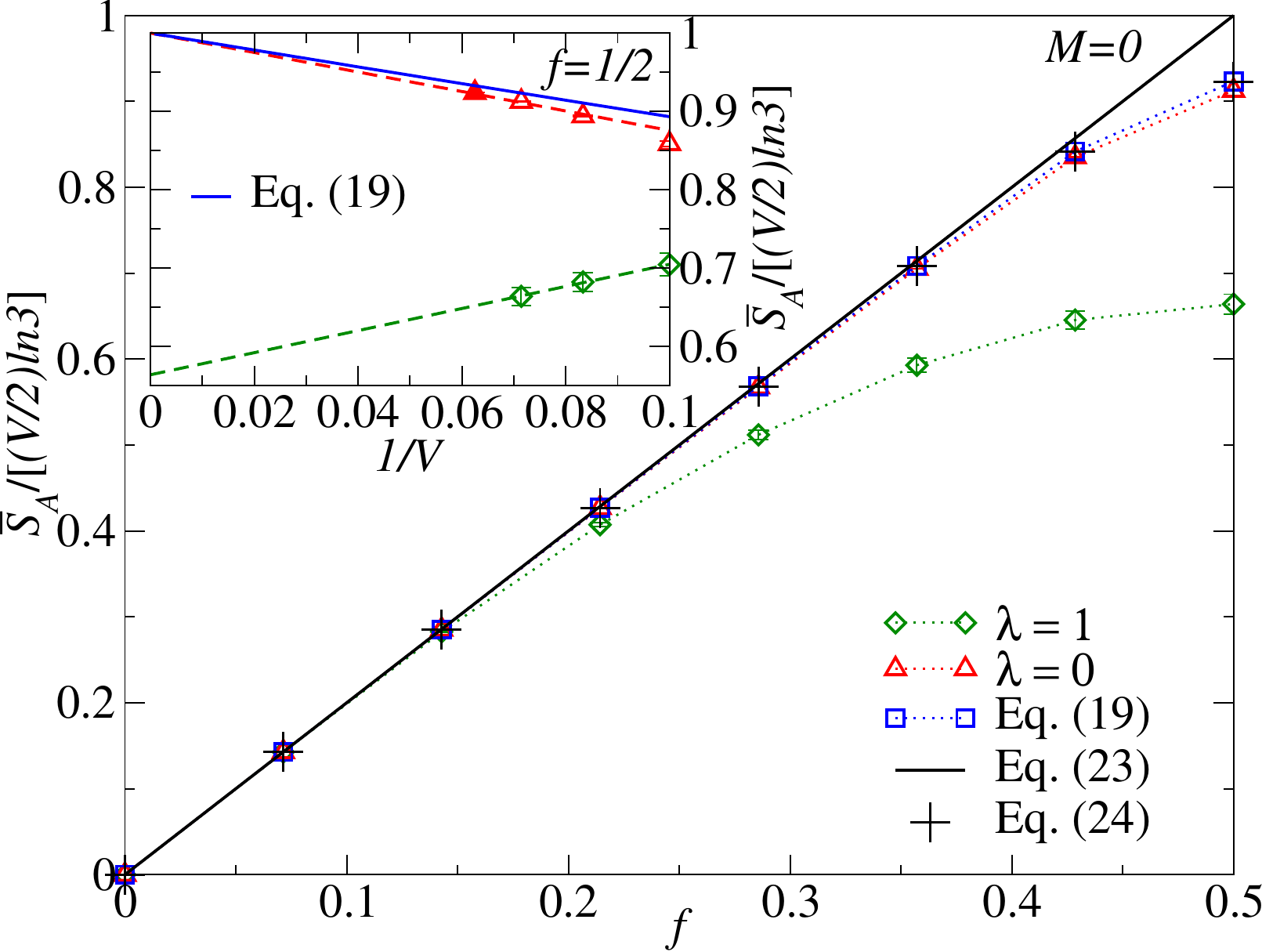}
    \caption{{\it Page curve for the extended spin-1 $XXZ$ model~\eqref{eq:Spin1_Hamiltonian} within the $M=0$ sector.} (Main panel) Same as in Fig.~\ref{fig:Spin-1_m=1/2} but for $V=14$. (Inset) Same as in Fig.~\ref{fig:Spin-1_m=1/2} except that for the quantum-chaotic case we use a single-parameter fit to $d+p/V$, where $d$ is set by Eq.~\eqref{eq:entropy_N_asymptotic} with $\beta(n)$ from Eq.~\eqref{eq:Beta_Spin-1}, and $p$ is our fitting parameter. Also, $\bar{S}_A$ for the largest chain ($V=16$, filled symbol) was obtained using only the $k=0,\,\pi$ sectors.}
    \label{fig:Spin-1_m=0}
\end{figure}

In Figs.~\ref{fig:Spin-1_m=1/2} and~\ref{fig:Spin-1_m=0}, we show our numerical results for $\bar{S}_A$ versus the subsystem fraction $f$ (also known as the Page curve) for the quantum-chaotic ($\lambda=0$) and integrable ($\lambda=1$) Hamiltonian~\eqref{eq:Spin1_Hamiltonian} eigenstates within the $m=1/2$ and $m=0$ magnetization sectors, respectively. For all subsystem fractions, $\bar{S}_A$ for the quantum-chaotic Hamiltonian eigenstates is very close to $\braket{S_A}_N$ from the exact sum in Eq.~\eqref{eq:entropy_fixed_N_exact_sum}. Furthermore, they both are nearly indistinguishable from our leading order analytical prediction (straight line) for $f\lesssim0.35$, and from the double-scaling Kronecker-$\delta$-resolved expression for $\braket{S_A}_N$ from Eq.~\eqref{eq:EE-resolved}  ($+$ symbols in the plots) for all values of $f$. However, $\bar{S}_A$ for the integrable Hamiltonian eigenstates departs from the exact sum for $\braket{S_A}_N$ as $f$ departs from $f=0$. 

Further evidence that our analytical results for $\braket{S_A}_N$ describe the leading terms of $\bar{S}_A$ for quantum-chaotic Hamiltonian eigenstates, which are distinct from those of $\bar{S}_A$ for integrable Hamiltonian eigenstates, is provided by the finite-size scaling analyses reported in the insets in Figs.~\ref{fig:Spin-1_m=1/2} and~\ref{fig:Spin-1_m=0} at $f=1/2$. Those numerical results suggest that, like in spin-1/2 systems~\cite{kliczkowski2023average}, the departure of $\bar{S}_A$ for quantum-chaotic Hamiltonian eigenstates from $\braket{S_A}_N$ occurs at the level of the $O(1)$ subleading correction, while for integrable Hamiltonian eigenstates already the leading terms are different. The same has been argued to occur in the absence of $U(1)$ symmetry~\cite{kliczkowski2023average, Haque_Khaymovich_2022, huang_21, nieva2023}, and in the presence of $SU(2)$ symmetry~\cite{patil2023}. Our results support the expectation that the average entanglement entropy of highly excited Hamiltonian eigenstates can be used as a universal diagnostics of quantum chaos and integrability in many-body systems~\cite{Bianchi_2022, LeBlond_19}.

\subsection{Bose-Hubbard model}\label{sec:BoseHubbard}

We consider next the Bose-Hubbard model in chains with $V$ sites, with Hamiltonian
\begin{equation}\label{eq:bh}
H=-\sum_{i=1}^{V}(\hat b_{i}^{\dagger} \hat b_{i+1}^{}+{\rm H.c.})+\frac{U}{2} \sum_{i=1}^V \hat n_{i}^{}(\hat n_{i}^{}-1)\,,
\end{equation}
where $\hat b_{i}^{\dagger}$ ($\hat b_{i}^{}$) is the bosonic creation (annihilation) operator at site $i$, and we consider periodic boundary conditions. The first term in Hamiltonian~\eqref{eq:bh} describes the hopping of bosons between nearest neighbor sites, and the second term describes their on-site interaction, with strength $U$ relative to the hopping amplitude (which we set to be the energy scale). Like in our analytical calculations, $N$ is the number of bosons and $n=N/V$ is the average filling. We compute the average entanglement entropy $\bar S_A$ over the 100 mid-spectrum energy eigenstates (50 for the smallest chain considered) within the even parity subsector of the $k=0$ total quasimomentum sector. 

\begin{figure}[!bt]
    \includegraphics[width=0.98\columnwidth]{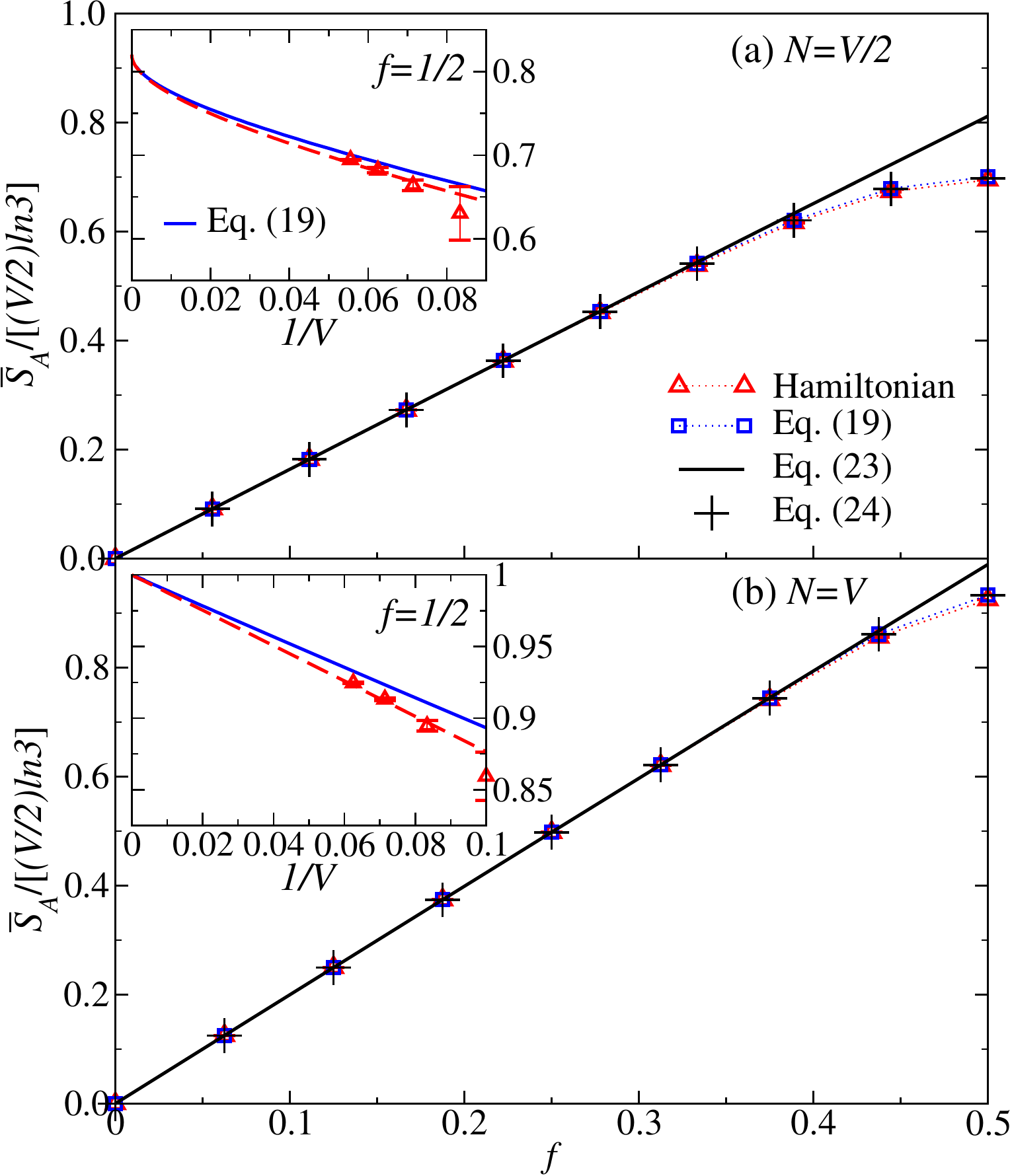}
    \caption{{\it Page curve for the Bose-Hubbard model~\eqref{eq:bh} with $n_{\max}=2$.} (Main panels) $\bar S_A$ vs $f$ for (a) $N=V/2$ with $V=18$ and $U=2.25$, and (b) $N=V$ with $V=16$ and $U=1.75$. We also report $\langle S_A\rangle_N$ from the exact sum in Eq.~\eqref{eq:entropy_fixed_N_exact_sum}, as well as the leading order prediction for $\braket{S_A}_N$ from Eq.~\eqref{eq:entropy_N_asymptotic} and the double-scaling Kronecker-$\delta$-resolved expression for $\braket{S_A}_N$ from Eq.~\eqref{eq:EE-resolved} [with $\beta(n)$ from Eq.~\eqref{eq:Beta_Spin-1}]. (Insets) $\bar S_A$ vs $1/V$ at $f=1/2$. The dashed lines following the numerical results are single parameter fits to the last three data points. In the inset in panel (a) [(b)] we use as fitting function $d_1+d_2/\sqrt{V}+p/V$ [$d+p/V$], where $d_1$ and $d_2$ [$d$] are set by Eq.~\eqref{eq:entropy_N_asymptotic}, and $p$ is our fitting parameter. The continuous lines show $\braket{S_A}_N$ from the exact sum in Eq.~\eqref{eq:entropy_fixed_N_exact_sum}. The error bars in the numerical results are the standard deviation of the averages.}
    \label{fig:BH_lm2}
\end{figure}

\begin{figure}[!bt]
    \includegraphics[width=0.98\columnwidth]{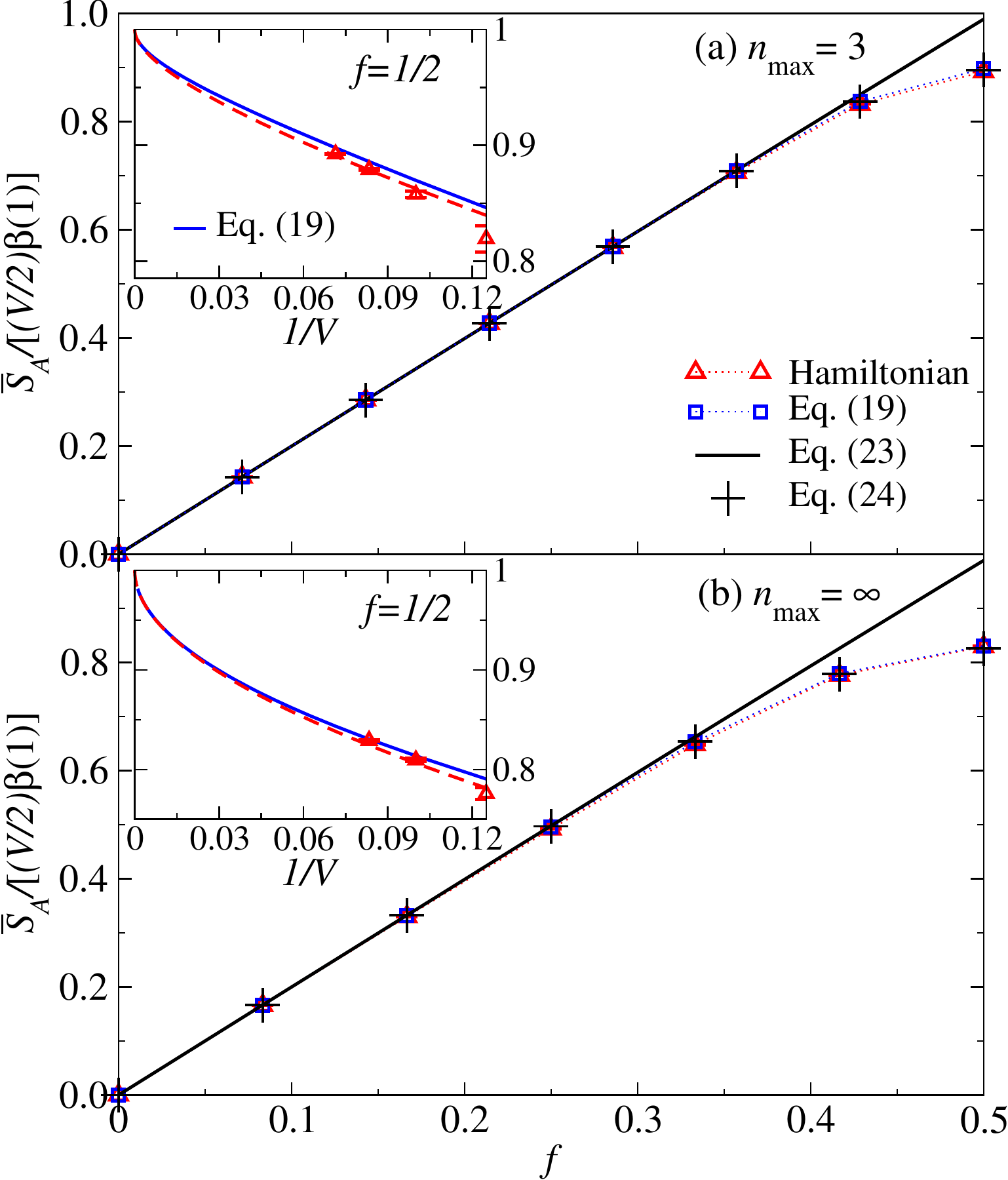}
    \caption{{\it Page curve for the Bose-Hubbard model~\eqref{eq:bh} with $N=V$.} (Main panels) Same as in Fig.~\ref{fig:BH_lm2} but for (a) $n_{\max}=3$ with $V=14$ and $U=1$, and (b) $n_{\max}=\infty$ with $V=12$ and $U=0.75$.  (Insets) The corresponding finite-size scalings as in Fig.~\ref{fig:BH_lm2}(a). Note that in this figure the normalization in the $y$ axes involves $\beta(1)$, which for $n_{\max}=3$ in panel (a) is $\beta(1)=1.284$, while for $n_{\max}=\infty$ in panel (b) is $\beta(1)=\ln 4$.}
    \label{fig:BH_lm3}
\end{figure}

We calculate the many-body eigenstates of Hamiltonian~\eqref{eq:bh} with no constraint on the maximal site occupation ($n_{\max}=\infty$), as is the case for the traditional Bose-Hubbard model, as well as with the constraint that at most $n_{\max}$ bosons may occupy a lattice site (in which case the local Hilbert space dimension is $n_{\max}+1$). When $n_{\max}=1$ the model is integrable. It describes hard-core bosons hopping on a lattice, and can be mapped onto the spin-1/2 $XX$ chain as well as onto a model of noninteracting spinless fermions~\cite{Cazalilla_2011}. The entanglement entropy of the eigenstates of those models was studied in Refs.~\cite{VidmarHackl_2017, HacklVidmar_2019}, and resembles the results in Sec.~\ref{sec:Spin-1XXZ} at the integrable point. Namely, the coefficient of the volume in the leading term is smaller than for quantum-chaotic Hamiltonian eigenstates and for random states. 

Here we focus on the cases $n_{\max}=2$, 3, and $\infty$, in which the model is quantum chaotic~\cite{Cazalilla_2011, Kollath_2010}. For each maximal site occupation $n_{\max}$ and filling $n$ considered, we select the value of $U$ to be in the maximally chaotic regime as per the discussion in Ref.~\cite{kliczkowski2023average}.

In Fig.~\ref{fig:BH_lm2} we plot Page curves for the Bose-Hubbard model with maximal site occupation $n_{\max}=2$, when $N=V/2$ with $V=18$ [Fig.~\ref{fig:BH_lm2}(a)] and when $N=V$ with $V=16$ [Fig.~\ref{fig:BH_lm2}(b)]. For $n_{\max}=2$, our model has the same generating function $\zeta(z)$ [Eq.~\eqref{eq:zetaspin-1}] and $\beta(n)$ [Eq.~\eqref{eq:Beta_Spin-1}] as the spin-1 model in Sec.~\ref{sec:Spin-1XXZ}. For both fillings one can see that $\bar{S}_A$ follows the prediction for $\langle S_A\rangle_N$ from the exact sum in Eq.~\eqref{eq:entropy_fixed_N_exact_sum}, and they both agree with leading order analytical prediction (straight line) for $f\lesssim0.35$ as well as with the double-scaling Kronecker-$\delta$-resolved expression for $\braket{S_A}_N$ from Eq.~\eqref{eq:EE-resolved}  ($+$ symbols in the plots) for all values of $f$, like in Figs.~\ref{fig:Spin-1_m=1/2} and~\ref{fig:Spin-1_m=0}. In the insets of Figs.~\ref{fig:BH_lm2}(a) and~\ref{fig:BH_lm2}(b), we carry out finite-size scaling analyses of the average entanglement entropy at $f=1/2$ that parallel the ones in the insets of Figs.~\ref{fig:Spin-1_m=1/2} and~\ref{fig:Spin-1_m=0}, respectively. The similarity of the scalings in the insets of Figs.~\ref{fig:BH_lm2}(a) and~\ref{fig:Spin-1_m=1/2} [Figs.~\ref{fig:BH_lm2}(b) and~\ref{fig:Spin-1_m=0}] is remarkable. It shows that the local Hilbert space dimension together with the filling/magnetization are the ones that control the leading terms [greater than $O(1)$] in the average entanglement entropy of highly excited energy eigenstates. Those leading terms appear to be universal independently of the model considered so long as it is quantum chaotic. As in the analytical calculations, it does not make a difference whether we deal with bosons or spins.

In Fig.~\ref{fig:BH_lm3} we plot Page curves for the Bose-Hubbard model at an average site occupation of one boson per site ($N=V$), when the maximal site occupation $n_{\max}=3$ with $V=14$ [Fig.~\ref{fig:BH_lm3}(a)] and when the maximal site occupation $n_{\max}=\infty$ with $V=12$ [Fig.~\ref{fig:BH_lm3}(b)]. For $n_{\max}=3$ the generating function is $\zeta=1+z+z^2+z^3$ ($n^*=3/2$), while for $n_{\max}=\infty$ it is $\zeta=1/(1-z)$ (there is no $n^*$). The agreement between the numerical results for Hamiltonian eigenstates and the analytical predictions for random states in Fig.~\ref{fig:BH_lm3} is similar to that in Figs.~\ref{fig:Spin-1_m=1/2}--\ref{fig:BH_lm2}. This supports the expectation that our analytical results predict the leading terms [greater than $O(1)$] in the average entanglement entropy of highly excited energy eigenstates of quantum-chaotic Hamiltonians with arbitrary local Hilbert spaces in the presence of particle-number conservation.

\section{Summary and Discussion}\label{sec:summary}

We calculated the bipartite entanglement entropy of typical pure states with a fixed number of particle, under the assumption that the total Hilbert space is constructed from identical local Hilbert spaces at individual sites. Our setup covers the vast majority of lattice systems of interest in physics, which involve fermions, bosons, spins, and their mixtures. We showed that our framework allows to straightforwardly predict what happens when one changes intrinsic properties, such as the spin of the particles.

We derived a general formula for the average entanglement entropy $\braket{S_A}_N$ up to $O(1)$ in $V$ in the thermodynamic limit, and showed that the variance $(\Delta S_A)^2_N$ vanishes exponentially fast in that limit. The latter finding implies that the computed average is also the typical entanglement entropy among all states with fixed particle number. Our result only depends on the asymptotic behavior of the dimension $d_N(V)=\frac{\alpha(n)}{\sqrt{V}}e^{\beta(n)V}$ of the Hilbert space for $N$ particles in $V$ sites on the particle density $n=N/V$. 

To use our results to predict the typical entanglement entropy of pure states in an arbitrary system, we provide a simple recipe based on the generating function $\zeta(z)$ in Eq.~\eqref{eq:generating_function}: (i) Find the sequence of Hilbert space dimensions $a_k=\dim \mathcal{H}^{(k)}_{\mathrm{loc}}$, where $\mathcal{H}^{(k)}_{\mathrm{loc}}$ is the Hilbert space of a site with $k$ particles, \ie $a_k$ tells us how many ways there are to place $k$ particles at a site. (ii) Write the generating function $\zeta(z)$ in Eq.~\eqref{eq:generating_function}, which will always be well-defined in a neighborhood of $z=0$. (iii) Find the unique saddle point solution $z_0(n)\geq 0$, such that $z_0\,\zeta'(z_0)=n\,\zeta(z_0)$. (iv) The asymptotics of the entanglement entropy is determined by $\beta(n)=\ln[\zeta(z_0)]-n\ln(z_0)$ and $\alpha(n)=\sqrt{-\beta''(n)/(2\pi)}$.

Moreover, the procedure introduced to find $\beta(n)$ allows us to relate the asymptotic properties of $\beta(n)$ to the behavior of the sequence $\{a_k\}^{\infty}_{k=0}$. Unsurprisingly, $\beta(n)$ is a bounded function for finite $n_{\max}$. Remarkably, we find that for $n_{\max}=\infty$, and for at most polynomially growing $a_k=O(k^p)$, $\beta(n)=O(\ln{n})$. Only in the rather exotic case of exponentially growing $a_k=O(e^k)$, we find $\beta(n)=O(n)$. Indistinguishably plays a fundamental role in those results. In Appendix~\ref{app:case-study-distinguishable}, we show that the entanglement entropy of a system with a fixed number of \emph{distinguishable} particles grows as $V\ln{V}$, \ie faster than volume law. Such a super-extensive behavior of the entanglement entropy is unphysical on a fundamental level. It parallels the well known result for the entropy of an ideal gas of distinguishable particles resulting in the Gibbs' paradox.

While our expressions for $\braket{S_A}_N$ and $(\Delta S_A)^2_N$ are important in their own right, they quantify the typical entanglement entropy in Hilbert spaces with fixed numbers of particles, an important motivation for this study is gaining a better understanding of the entanglement entropy of highly excited eigenstates of generic particle-number conserving quantum-chaotic models. Using numerical calculations, we found evidence that the leading terms [greater than $O(1)$] in the average entanglement entropy of mid-spectrum eigenstates of two paradigmatic quantum-chaotic models, the spin-$1$ $XXZ$ model and the Bose-Hubbard model, are described by our analytical expression for $\braket{S_A}_N$. In contrast, when we repeat the analysis for the spin-$1$ $XXZ$ model at an integrable point ($\lambda=1$), we find an increasingly larger discrepancy at leading order as the subsystem fraction $f=V_A/V$ approaches $f=1/2$, consistent with previous findings in Refs.~\cite{LeBlond_19, Bianchi_2022}. Our results indicate that eigenstate entanglement entropy is a universal diagnostic of quantum chaos and integrability in many-body quantum systems for arbitrary local Hilbert space sizes, complementing previous studies for the case of local two-dimensional Hilbert spaces~\cite{LeBlond_19,Bianchi_2022}.

An important finding of our work is the insight that the term $[f+\ln(1-f)]/2$ in Eq.~\eqref{eq:entropy_N_asymptotic} is \emph{universal}, i.e., independent of the specifics of the particles/spins involved. This term was found (within a ``mean-field'' calculation) for the case of a local two-dimensional Hilbert space in Ref.~\cite{vidmar2017entanglement}. Our work establishes that it is a universal consequence of particle-number conservation, as it is not present once one removes such a constraint~\cite{page1993information, Bianchi_2022}. A well known example in which only the $O(1)$ term is universal, while the leading order depends on the specifics of the system, is the ground-state entanglement entropy of topologically ordered two-dimensional models~\cite{kitaev2006topological}. In their seminal work~\cite{kitaev2006topological}, Kitaev and Preskill related such an $O(1)$ term to the so-called total quantum dimension and coined it the \emph{topological entanglement entropy}. Remarkably, increasing the symmetry of the system from $U(1)$ to $SU(2)$ changes this $O(1)$ term, as proved in Ref.~\cite{patil2023} for the total spin $J=0$ case, in which case the $O(1)$ term is $3[f+\ln(1-f)]/2$. It is therefore a natural question for future work to explore whether the $O(1)$ correction is universal for each symmetry selected, \eg whether for all $SU(2)$ symmetric systems one has $3[f+\ln(1-f)]/2$.

Our work also establishes that, in the presence of $U(1)$ symmetry, there is always a $\sqrt{V}$ term if the system is split into two equal halves, \ie at $f=1/2$. This term was found for the case of a local two-dimensional Hilbert space in Ref.~\cite{vidmar2017entanglement}. We unveil two important facts about this term. The first one that it is also a universal consequence of particle-number conservation, as it is not present in its absence~\cite{page1993information, Bianchi_2022}. The second one is that the prefactor of $\sqrt{V}$ is completely fixed by the same function $\beta(n)$ that determines the leading order behavior. 

Finally, we identified the general location of Page's $-1/2$ correction in the presence of $U(1)$ symmetry. It is controlled via the special Kronecker $\delta_{n,n^*}$, which only appears at a filling density $n^*$ (of which there exists at most one) such that $\beta'(n^*)=0$. This term is mutually exclusive with the $\sqrt{V}$-term. The latter is proportional to $|\beta'(n)|$ and thus vanishes at $n=n^*$. Remarkably, our expression for $\braket{S_A}_N$ demonstrates that it suffices to find the functional form of $\beta(n)$ from the leading volume-law term to get the full asymptotics of the typical entanglement entropy up to $O(1)$ in $V$. This is striking as $\beta(n)$ only captures the leading order behavior of the Hilbert space dimension $d_N$. 

Interesting directions for future work include studying the symmetry-resolved entanglement entropy within our general framework, to generalize recent results obtained in the context of local two-dimensional Hilbert spaces~\cite{murciano2022symmetry}. Another interesting direction is to generalize our results, in which the total number of particles $N$ was fixed for all species at once, to the case in which the particle numbers are fixed independently for different species. Our framework opens the door to address many interesting questions in the context of the entanglement entropy of composite systems.

\begin{acknowledgments}
We thank Eugenio Bianchi, Mario Kieburg, and Lev Vidmar for inspiring discussions. We thank Thorsten Neuschel for a suggestion leading to our asymptotic expansion of $d_N$, and Mario Kieburg for suggesting to consider the case of distinguishable particles. Y.Y.~acknowledges support by the Vacation Scholars Program of the School of Mathematics and Statistics at the University of Melbourne. R.P.~and M.R.~acknowledge support from the National Science Foundation under Grants No.~PHY-2012145 and No.~PHY-2309146. Y.Z.~acknowledges the support from the Dodge Family Postdoc Fellowship at the University of Oklahoma. LH acknowledges support by the Alexander von Humboldt Foundation, by grant $\#$62312 from the John Templeton Foundation, as part of the \href{https://www.templeton.org/grant/the-quantuminformation-structure-ofspacetime-qiss-second-phase}{‘The Quantum Information Structure of Spacetime’ Project (QISS)}, by grant $\#$63132 from the John Templeton Foundation and an Australian Research Council Australian Discovery Early Career Researcher Award (DECRA) DE230100829 funded by the Australian Government.
\end{acknowledgments}

\appendix

\section{Saddle point analysis}\label{app:saddle-point-analysis}
Here, we provide the necessary mathematical proofs to establish the properties of the saddle point $z_0(n)$ solving Eq.~\eqref{eq:psi_derivative_equals_zero} that are used in the main text.

\subsection{Generating function}
\label{app:assumptions}
In Eq.~\eqref{eq:generating_function}, we introduced the generating function $\zeta(z)=\sum^{\infty}_{k=0}a_kz^k$ with the rather mild requirement from Eq.~\eqref{eq:scaling-assumption} that the coefficients $a_k$ scale at most exponentially with $k$. Apart from this, the coefficients $a_k$ satisfy the following natural properties:
\begin{itemize}
	\item All $a_k$ are nonnegative integers, as they represent Hilbert space dimensions.
	\item We have $a_0>0$, as there must be at least one type of vacuum (zero particles) at a site. Usually $a_0=1$, \ie there is only one way to place zero particles at a site representing a \emph{unique} vacuum.
    \item There exists at least one $a_k\neq 0$ for $k>0$, as otherwise the system would not accommodate any particles.
\end{itemize}
For the growth of $a_k$, we can distinguish the following three cases, which are all compatible with the mild requirement of at most exponential growth discussed in the main text:
\begin{itemize}
	\item[(a)] \textbf{Series is finite with $a_k=0$ for $k>n_{\max}$.} Note that this is equivalent to $\lim_{k\to\infty} a_k=0$, as $a_k\in\amsbb{N}$. The function $\zeta$ is defined everywhere on the complex plane and there is a maximal total particle number given by $N_{\max}=Vn_{\max}$ that the system can accommodate. It corresponds to a particle density of $n_{\max}$ particles per site.
	\item[(b)] \textbf{Sequence of $a_k$ grows subexponentially with $a_k=o(e^k)$, but $\lim_{k\to\infty}a_k\neq 0$.} The series defining $\zeta$ converges inside a disk of radius $R=1$.
	\item[(c)] \textbf{Sequence of $a_k$ grows exponentially, such that $R=\exp(-\lim_{l\to\infty}\sup\{\ln(a_k)/k|k<l\})$ with $0<R<1$.} The series defining $\zeta$ converges inside a disk of radius $R$.\\
\end{itemize}

\subsection{Positive real saddle point $z_0(n)$}
\label{app:saddle_analysis}
We can rewrite the saddle point Eq.~\eqref{eq:psi_derivative_equals_zero} as
\begin{equation}
    \label{eq:function_Z_definition}
    Z(z_0)=n\quad \text{with}\quad Z(z)=z\frac{\zeta'(z)}{\zeta(z)}\,,
\end{equation}
where we introduced the function $Z: \{z\in\amsbb{C}: |z|<R\}\to\amsbb{C}$. Now consider the restriction of $Z$ to the real interval $z\in\left(0,R\right)$, which we will study in the following.

\begin{lemma}[Monotonicity of $Z$]\label{lem:monotonic}
The function $Z(z)=z\zeta'(z)/\zeta(z)$ is strictly increasing on the interval $(0,R)$.
\end{lemma}
\begin{proof}
We compute the derivative
\begin{equation}
    Z'(z)=\frac{\left[\zeta'(z)+z\zeta''(z)\right]\zeta(z)-z\zeta'(z)^2}{\zeta(z)^2}\,.
\end{equation}
where $\zeta$ was introduced in Eq.~\eqref{eq:generating_function}. The denominator clearly satisfies $\zeta(z)>0$ for $z\in(0,R)$, as all coefficients $a_k$ are nonnegative and we have $a_0>0$ and at least one other coefficient. Therefore, we look at the numerator whose expansion gives
\begin{widetext}
\begin{align}
    \left[\zeta'(z)+z\zeta''(z)\right]\zeta(z)-z\zeta'(z)^2&=\sum_{k=0}^\infty (k+1)^2a_{k+1}z^k\sum_{k=0}^\infty a_kz^k-\sum_{k=0}^\infty (k+1)a_{k+1}z^{k}\sum_{k=0}^\infty k a_kz^k \\
    &=\sum_{k=0}^\infty \left(\sum_{l=0}^k (l+1)^2a_{l+1}a_{k-l}\right)z^k-\sum_{k=0}^\infty \left(\sum_{l=0}^k (l+1)(k-l)a_{l+1}a_{k-l}\right)z^k \\
    \label{eq:Z_derivative_numerator}
    &=\sum_{k=0}^\infty \left( \sum_{l=0}^k(l+1)(2l+1-k)a_{l+1}a_{k-l}\right)z^k\\
    &=\sum_{k=0}^\infty \left( (k+1)^2a_{k+1}a_0+\sum^{\lfloor\frac{k}{2}\rfloor}_{l=1}(2l-1-k)^2a_{l}a_{k-l+1}\right)z^k\label{eq:summand-term}
\end{align}
\end{widetext}
where in the second line we used the Cauchy product of a power series and, in the fourth line, we took out the $l=k$ term and then combined the $l$-term with the $k-l-1$ term. Note that for odd $k$, the $l=(k-1)/2$ term vanishes. As all coefficients $a_k$ are nonnegative with $a_0>0$ and at least one other $a_k>0$, all summands in Eq.~\eqref{eq:summand-term} must be nonnegative and at least one must be positive. Therefore, we have $Z'(z)>0$ for all $z\in(0,R)$, from which the claim follows.
\end{proof}

\begin{lemma}[Boundary limits of $Z$]\label{lem:boundary}
On the interval $(0,R)$, the function $Z$ has the limits
\begin{align}
    \lim_{z\to 0}Z(z)=0\qquad\text{and}\qquad
    \lim_{z\to R}Z(z)=n_{\max}\,.
\end{align}
\end{lemma}
\begin{proof}
For the first limit, we use $a_0>0$ and compute
\begin{equation}
        \lim_{z\to 0^+}Z(z)=\lim_{z\to0^+} z\frac{\zeta'(z)}{\zeta(z)}=\lim_{z\to0^+} z\frac{\sum_{k=1}^{\infty}ka_kz^{k-1}}{a_0+\sum_{k=1}^{\infty}a_kz^k}=0\,.
\end{equation}
For the second limit, we first consider $n_{\max}<\infty$. In this case, $\zeta(z)$ is a finite polynomial and direct evaluation yields
\begin{align}
    \lim_{z\to\infty}Z(z)=\lim_{z\to\infty} z\frac{\zeta'(z)}{\zeta(z)}=\lim_{z\to\infty} \frac{\sum_{k=1}^{n_{\max}}ka_kz^k}{\sum_{k=0}^{n_{\max}}a_kz^k}=n_{\max}\,.
\end{align}
For $n_{\max}=\infty$, we want to show that $\lim_{z\to R}Z(z)=\infty$, for which we use Abel's theorem for diverging series. Assume, for a contradiction, that $\lim_{z\to R}\frac{\zeta'(z)}{\zeta(z)}<\infty$. It follows that $\int^R_0\frac{\zeta'(z)}{\zeta(z)}\dd{z}$ is finite, but $\int^R_0\frac{\zeta'(z)}{\zeta(z)}\dd{z}=\ln[\lim_{z\to R^-}\zeta(z)]-\ln[\zeta(0)]$ diverges. Hence, we must have $\lim_{z\to R}\frac{\zeta'(z)}{\zeta(z)}=\infty$ and thus $\lim_{z\to R}Z(z)=\infty$.
\end{proof}

Together, the previous two lemmata establish that there exists a unique real solution $z_0(n)>0$ of the saddle point equation, which grows monotonically with $n$.

\begin{proposition}[Existence and monotonicity of $z_0$]\label{prop:existence}
For $n\in(0,n_{\max})$, there exists a unique positive solution $z_0(n)$ of the saddle point equation. Moreover, $z_0(n)$ increases monotonically with $n$, so that $z_0'(n)>0$.
\end{proposition}
\begin{proof}
Recall from Eq.~\eqref{eq:function_Z_definition} that $Z(z_0(n))=n$. Lemma~\ref{lem:monotonic} establishes that $Z$ is strictly increasing and lemma~\ref{lem:boundary} shows that the range of $Z$ is given by $(0,n_{\max})$. Therefore, there exists a unique solution $z_0=z_0(n)$, such that $Z(z_0)=n$. As the function $Z(z)$ is strictly increasing, the argument $z_0(n)$ must increase when increase $n$. This means that $z_0(n)$ is a strictly increasing function of $n$, so that $z_0'(n)>0$.
\end{proof} 

The saddle point defining equation $Z(z_0(n))=n$, along with the results of the two previous sections, mean that $z_0(0)=0$ and $\lim_{n\to {n_{\max}}}z_0(n)=R$.

\subsection{Analyzing the exponential scaling $\beta(n)$}
In the following discussion, we analyze the derivatives of $\beta(n)$ to understand its behavior.

\begin{proposition}[Derivatives of $\beta$]
The derivatives of $\beta(n)$ are given by
\begin{align}
\label{eq:beta_derivatives}
    \beta'(n)=-\ln(z_0)\quad\text{and}\quad \beta''(n)=-\frac{z'_0}{z_0}\,.
\end{align}
\begin{proof}
We compute $\beta'(n)$ straight from its definition in Eq.~\eqref{eq:solution-saddle} and get
\begin{align}
    \beta'(n)&=\pdv{n}\psi(z_0)=\psi'(z_0)z_0'-\ln(z_0)\,.
\end{align}
The first term is zero because it is precisely the saddle point condition, so we see that $\beta'(n)=-\ln(z_0)$. Then the second derivative of $\beta$ trivially follows by taking another derivative with respect to $n$.
\end{proof} 
\end{proposition}

\begin{proposition}[Concavity of $\beta$]
The function $\beta(n)$ is for all $n\in (0,n_{\max})$ concave, \ie $\beta''(n)<0$ for all $n\in (0,n_{\max})$.
\end{proposition}
\begin{proof}
We recall from Eq.~\eqref{eq:beta_derivatives} that $\beta''(n)=-\frac{z_0'}{z_0}$. The saddle point is positive and in Proposition~\ref{prop:existence} we showed that $z'_0>0$ for $n\in (0,n_{\max})$, from which the claim follows.
\end{proof}

\begin{proposition}[Monotonicity of $\beta$]
There exists a unique $n^*=\frac{\zeta'(1)}{\zeta(1)}$, such that $\beta'(n^*)=0$, with $\beta'(n)>0$ for $n\in (0,n^*)$ and $\beta'(n)<0$ for $n\in(n^*,n_{\max})$, if and only if $n_{\max}<\infty$. Otherwise, we have $\beta'(n)>0$ for all $n\in (0,\infty)$.
\end{proposition}
\begin{proof}
When $n_{\max}=\infty$, $z_0\in\interval[open right]{0}{R}$, where $R\leq1$. Hence, $\beta'(n)=-\ln(z_0)>0$. When $n_{\max}<\infty$, $z_0\in\interval[open right]{0}{\infty}$. Since $z_0$ is monotonically increasing, $\beta'$ is monotonically decreasing and $\beta'\in\interval[open]{-\infty}{\infty}$. By the intermediate value theorem, there exists a unique point $n^*\in\interval[open]{0}{n_{\max}}$ with $\beta'(n^*)=0$, with $\beta'$ changing sign either side of $n^*$. In fact, we can directly calculate $n^*$. We note first that $\beta'(n^*)=0$ implies $z_0(n^*)=1$. This saddle point satisfies Eq.~\eqref{eq:psi_derivative_equals_zero}, which can be solved for $n^*$, giving $n^*=\frac{\zeta'(1)}{\zeta(1)}$.
\end{proof}

\begin{proposition}[Boundary points of $\beta$]
We have $\beta(0)=\ln(a_0)$, $\beta(n_{\max})=\ln(a_{n_{\max}})$ for $n_{\max}<\infty$ and $\lim_{n\to\infty}\beta(n)=\infty$ for $n_{\max}=\infty$.
\end{proposition}
\begin{proof}
From the definition of $\beta$ in Eq.~\eqref{eq:solution-saddle}, it is clear that $\beta(0)=\ln(a_0)$ since in the limit $n\to 0$, we have $\lim_{n\to 0}n \ln{z_0(n)}=0$. In the limit $n\to n_{\max}$ for finite $n_{\max}$,  we compute
\begin{eqnarray}
    \lim_{n\to n_{\max}}\beta(n)&=&\lim_{n\to n_{\max}} \ln\left(\frac{\sum_{k=0}^{n_{\max}}a_kz_0^k}{z_0^n}\right)\nonumber \\&=& \ln\left(\lim_{n\to n_{\max}}\frac{\sum_{k=0}^{n_{\max}}a_kz_0^k}{z_0^n}\right)\,.
\end{eqnarray}
If $n_{\max}<\infty$, the argument is just a rational function so that $\lim_{n\to n_{\max}}\beta(n)=\ln(a_{n_{\max}})$. If $n_{\max}=\infty$, we have $\lim_{n\to n_{\max}}z_0(n)=R$ with $\lim_{z\to R}\psi(z)=\infty$.
\end{proof}

\subsection{Relationship between $\alpha$ and $\beta$}
\label{app:alpha_beta_proof}
A key result of our analysis is that the function $\beta$ and its derivatives provide all the relevant information when studying average entanglement entropy up to $O(1)$ in $V$. This is due to the fact that the parameters $\alpha(n)$ and $\beta(n)$ in the asymptotics $d_N=\frac{\alpha(n)}{\sqrt{V}}e^{\beta(n)V}+o(1)$ are not independent, as the following proposition shows.

\begin{proposition}
Given the saddle point described in Eq.~\eqref{eq:saddle-asymp-sol} with $d_N=\frac{\alpha(n)}{\sqrt{V}}e^{\beta(n)V}$, we have the relation
\begin{equation}
    \alpha(n)=\sqrt{\frac{-\beta''(n)}{2\pi}}\,.\label{eq:alpha_relation_beta_1}
\end{equation}
\end{proposition}
\begin{proof}
We present two different versions of the proof highlighting two different perspectives.\\
\textbf{Proof (version 1).} Recall that we have the relation
\begin{widetext}
\begin{align}
    d_N=\sum^N_{N_A=0} d_Ad_B\quad\text{with}\quad \begin{array}{l}
    d_A=d_N(N_A,V_A)=\frac{\alpha(\frac{N_A}{V_A})}{\sqrt{V_A}}e^{\beta(\frac{N_A}{V_A}) V_A}=\frac{\alpha(\frac {n_A}{f})}{\sqrt{fV}}e^{\beta(\frac{n_A}{f})fV} \\
    d_B=d_N(N_B,V_B)=\frac{\alpha(\frac{N_B}{V_B})}{\sqrt{V_B}}e^{\beta(\frac{N_B}{V_B}) V_B}=\frac{\alpha(\frac{n-n_A}{1-f})}{\sqrt{(1-f)V}}e^{\beta(\frac{n-n_A}{1-f})(1-f)V}
    \end{array}\,,\label{eq:dN-consistency}
\end{align}
\end{widetext}
where $d_A$ and $d_B$ have the same functional form as $d_N$ and we used the relations $f=V_A/V$, $n_A=N_A/V$, $N_B=N-N_A$ and $V_B=V-V_A$. This yields the asymptotics
\begin{eqnarray}
    && d_Ad_B=\frac{\alpha(\frac{n_A}{f})\alpha(\frac{n-n_A}{1-f})}{\sqrt{f(1-f)}V}e^{V\lambda(n_A)},\quad\text{with}\nonumber \\&& \lambda(n_A)=f\beta\!\left(\frac{n_A}{f}\right)+(1-f)\beta\!\left(\frac{n-n_A}{1-f}\right).\quad
\end{eqnarray}
We can convert the sum into an integral $\sum^N_{N_A=0}\to V\int_0^1dn_A$, which will be dominated by the saddle point with saddle point equation
\begin{align}
    \lambda'(n_A)=\beta'\!\left(\frac{n_A}{f}\right)-\beta'\!\left(\frac{n-n_A}{1-f}\right)=0\,,
\end{align}
which has the unique solution $n_A=f n$. This yields
\begin{eqnarray}
    d_N&=&V\int_0^1\frac{\alpha(\frac{n_A}{f})\alpha(\frac{n-n_A}{1-f})}{\sqrt{f(1-f)}V}e^{V\lambda(n_A)}\dd{n_A}\nonumber\\&=&\frac{\alpha^2(n)}{\sqrt{f(1-f)}}\sqrt{\frac{2\pi}{-V\lambda''(fn)}}e^{V\lambda(fn)}+o(1).\qquad
\end{eqnarray}
Using $\lambda(fn)=\beta(n)$ and $\lambda''(fn)=\frac{\beta''(n)}{f(1-f)}$, setting $d_N=\frac{\alpha(n)}{\sqrt{V}} e^{\beta(n)V}+o(1)$ on the l.h.s., allows us to solve for $\alpha(n)$ yielding Eq.~\eqref{eq:alpha_relation_beta_1}. This result is thus a consequence of Eq.~\eqref{eq:dN-consistency}, which can be interpreted as a consistency relation when splitting a system into subsystems.\\
\textbf{Proof (version 2).} The saddle point approximation from Eq.~\eqref{eq:saddle-asymp-sol} yielded
\begin{align}
    \alpha(n)=\frac{1}{\sqrt{2\pi z_0^2\psi''(z)}}\,,
\end{align}
which we would like to relate to $\beta''(n)$ from Eq.~\eqref{eq:solution-saddle}. We can compute $\psi''(z_0)$ to give
\begin{align}
    \psi''(z)=\frac{n}{z_0^2}+\frac{\zeta''(z_0)}{\zeta(z_0)}-\left[\frac{\zeta'(z_0)}{\zeta(z_0)}\right]^2\,.\label{eq:doublederivativepsi}
\end{align}
To simplify expressions, we would like to get rid of the derivative terms $\zeta'(z_0)$ and $\zeta''(z_0)$. For this, we can use the saddle point Eq.~\eqref{eq:psi_derivative_equals_zero} and its derivative with respect to $n$ (recall: $z_0$ depends on $n$) to give
\begin{align}
    \zeta'(z_0)=\frac{n \zeta(z_0)}{z_0}\ \ \text{and}\ \ \zeta''(z_0)=\frac{\zeta(z_0)+(n-1)\zeta'(z_0)z_0'}{z_0 z_0'}\,,\label{eq:zeta-derivatives}
\end{align}
where we solved for $\zeta'(z_0)$ and $\zeta''(z_0)$, respectively. Plugging Eq.~\eqref{eq:zeta-derivatives} into Eq.~\eqref{eq:doublederivativepsi} yields $\psi''(z_0)=\frac{1}{z_0z_0'}$, which gives the desired result
\begin{align}
    \alpha(n)=\sqrt{\frac{1}{2\pi}\frac{z_0'}{z_0}}=\sqrt{\frac{-\beta''(n)}{2\pi}}\,,
\end{align}
where we used $\beta''(n)=-\frac{z'_0}{z_0}$ from Eq.~\eqref{eq:beta_derivatives}.
\end{proof}

\section{Resolution of Kronecker $\delta$s}
\label{app:resolution_kronecker_delta}

The Kronecker $\delta$ corrections in Eq.~\eqref{eq:entropy_N_asymptotic} are nonanalytic, but can be resolved using a double-scaling limit, meaning that we take the limit of two variables simultaneously. Of interest in our case are the limits $f\to1/2$ and $n\to n^*$. We will see that the correction depends on the scaling of the distances $f-1/2$ and $n-n^*$.

We consider Eq.~\eqref{eq:entropy_fixed_N_exact_sum}, where we have two nonanalytical points, due to the $\max$ and $\min$ functions. We consider the splitting 
\begin{align}
    \varphi_1&=-\min\left[\frac{d_A-1}{2d_B},\frac{d_B-1}{2d_A}\right]\,, \\ 
    \label{eq:varphi_2_definition}
    \varphi_2&=\Psi(d_N\!+\!1)-\Psi(\max[d_A,d_B]\!+\!1)\,,
\end{align}
such that $\varphi=\varphi_1+\varphi_2$. In particular, each of $\varphi_1$ and $\varphi_2$ contain nonanalytical functions, that when summed (integrated) over with $\varrho_{N_A}$, yield a Kronecker $\delta$. We refer to these functions as
\begin{align}
    x_1&=\min\left[\frac{d_A-1}{d_B},\frac{d_B-1}{d_A}\right]\,,\label{eq:x1_definition}\\
    \begin{split}
    \label{eq:x2_definition}
    x_2&=\ln(\frac{d_B}{d_A})[\Theta(n^*-n)\Theta(n_A-n_\text{crit}) \\ 
    &\phantom{=\ln(\frac{d_B}{d_A})}+\Theta(n-n^*)\Theta(n_\text{crit}-n_A)]\,,
    \end{split}
\end{align}
where $\Theta(x)$ is the Heaviside step function. While one easily sees that $\varphi_1=-x_1/2$, the relationship between $\varphi_2$ and $x_2$ is not as obvious, but is explained in the second section below.

Both nonanalytical points require that $d_A=d_B$, which can only occur for $f=1/2$ and $n_A=fn$ in the limit $V\to\infty$, as can be seen from setting the exponents for $d_A$ and $d_B$ equal in Eq.~\eqref{eq:general_saddle_dA_dB} leading to
\begin{align}
    \beta\!\left(\frac{n_A}{f}\right)fV=\beta\!\left(\frac{n-n_A}{1-f}\right)(1-f)V\,,
\end{align}
which is only satisfied for all $\beta$ and $n$ at $f=1/2$ and $n_A=fn$. We therefore expand the two nonanalytical terms around this point.

Because of the $\max$ and $\min$ functions, it is useful to determine which of $d_A$ or $d_B$ is larger for different regimes of $n_A$, $f$ and $n$. We consider $d_A/d_B\propto \exp[V(f\beta(\tfrac{n_A}{f})-(1-f)\beta(\tfrac{n-n_A}{1-f}))]$ by using Eq.~\eqref{eq:general_saddle_dA_dB}, where we shall define the factor in the exponent to be
\begin{equation}
    Y=f\beta(\frac{n_A}{f})-(1-f)\beta(\frac{n-n_A}{1-f})\,.\label{eq:def-Y}
\end{equation}
For large $V$, finding the larger dimension is equivalent to determining the sign of $Y$, that is to say
\begin{equation}
    d_A<d_B\iff Y<0\quad \text{and} \quad d_A>d_B\iff Y>0\,.
\end{equation}
Note that by the concavity of the function $\beta(n)$, Eq.~\eqref{eq:def-Y} has at most one root for $n_A$, which we shall call $n_\text{crit}$ if it exists.

Since we are integrating the $\max$ and $\min$ functions against a Gaussian, it is relevant to consider this near the mean of the Gaussian. Expanding Eq.~\eqref{eq:def-Y} around $\overline{n}_A=fn$, 
\begin{eqnarray}
\label{eq:Y_expansion_fn}
    Y&=&2\left(f-\tfrac{1}{2}\right)\beta(n)+2\beta'(n)(n_A-fn) \\&& + \frac{\left(f-\frac{1}{2}\right)\abs{\beta ''(n)}}{f(1-f)}\left(n_A-f n\right)^2 +O(n_A-fn)^3\,,\nonumber
\end{eqnarray}
where the absolute value is for notational convenience since $\beta''(n)$ is always negative. This expression will be useful in the following two sections when we split integrals into two regimes.

\subsection{Kronecker $\delta$ at $n=n^*$ and $f=\frac{1}{2}$}
\label{app:resolve_ratio_kronecker_delta}
We first consider the effect of $x_1=\min\left[\frac{d_A-1}{d_B},\frac{d_B-1}{d_A}\right]$ from Eq.~\eqref{eq:x1_definition} by defining
\begin{equation}
    X_1=\sum^N_{N_A=0}\varrho_{N_A}x_1=\int \varrho(n_A)x_1 dn_A+o(1)\,,
\end{equation}
recalling that $\varrho(n_A)$ is a Gaussian with mean $\bar{n}_A=fn$ and standard deviation $\sqrt{f(1-f)/(\abs{\beta''(n)}V)}$. We use Eq.~\eqref{eq:general_saddle_dA_dB} to expand the dimensions and find that the minimum may be reexpressed as
\begin{equation}
    x_1=\min\left[\frac{d_A}{d_B},\frac{d_B}{d_A}\right]+o(1)=\exp[-V\abs{Y}]+o(1)\,,
\end{equation}
where $Y$ is from Eq.~\eqref{eq:def-Y} and we ignored the square-root factors from Eq.~\eqref{eq:general_saddle_dA_dB}, since we are integrating this function against the Gaussian $\varrho(n_A)$, and the square-root factor is unity at the mean of the Gaussian.

In order for $X_1$ to not vanish in the limit of large $V$, we must require the constant and first order terms of $Y$ in Eq.~\eqref{eq:Y_expansion_fn} to vanish in this limit. This means we must have $f-1/2=O(1/V)$. The linear term vanishes only at $n^*$, which is defined as the value of $n$ such that $\beta'(n^*)=0$. Further expansion around $n^*$ yields
\begin{eqnarray}
    \label{eq:Y_expansion_nstar}
    Y&=&2\left(f-\tfrac{1}{2}\right)\beta(n^*)-2\abs{\beta''(n^*)}(n-n^*)(n_A-fn)\nonumber\\&&+O(n-n^*)^2.
\end{eqnarray}
The quadratic term in Eq.~\eqref{eq:Y_expansion_fn} may be ignored since it is only relevant if the first two terms vanish, which enforces $f-1/2=O(1/V)$. But then this term is already small compared to the quadratic term in the Gaussian which is of order $V$. We rewrite our integral as
\begin{eqnarray}
\label{eq:X1_integral_reduced}
    X_1&=&\int_{-\infty}^\infty \varrho(n_A) \exp \left[-2V\,\bigg|\!\left(f-\frac{1}{2}\right) \beta(n^*)\right.\\ && \left. \quad - \abs{\beta''(n^*)}(n-n^*)(n_A-fn) \bigg| \right]\dd{n_A}+o(1)\,.\nonumber
\end{eqnarray}
We shall define the point, at which the absolute value switches sign (equivalent to $d_A=d_B$) to be $n_\text{crit}$, whose expansions we compute as
\begin{equation}
\label{eq:ncrit_expansion_half_nstar}
    n_\text{crit}=fn+\frac{\left(f-\frac{1}{2}\right)\beta(n^*)}{\abs{\beta''(n^*)}(n-n^*)}\,.
\end{equation}
From here we need to distinguish between the cases $n-n^*<0$ and $n-n^*>0$. First we consider $n-n^*<0$. To deal with the absolute value, we must split the integral into two parts. One verifies that $n_A<n_\mathrm{crit}$ implies $Y<0$ using Eq.~\eqref{eq:Y_expansion_nstar}. With this, we arrive at the integral
\begin{widetext}
\begin{align}
    \begin{split}
    X_1&=\int_{-\infty}^{n_\text{crit}} \varrho(n_A)\exp\left[2V\left(f-\frac{1}{2}\right)\beta(n^*)-2V\abs{\beta''(n^*)}(n-n^*)(n_A-fn)\right]\dd{n_A} \\
    &\phantom{={}}+\int_{n_\text{crit}}^\infty \varrho(n_A)\exp\left[-2V\left(f-\frac{1}{2}\right)\beta(n^*)+2V\abs{\beta''(n^*)}(n-n^*)(n_A-fn)\right]\dd{n_A}\,,
    \end{split}
\end{align}
which we can evaluate as
\begin{eqnarray}
\label{eq:case1-X1}
    X_1\!\!&=&\!\!\frac{1}{2}\exp\!\left[\frac{(n-n^*)^2}{2}V\abs{\beta''(n^*)}\right]\!\left[\exp\!\left[2\left(f-\frac{1}{2}\right)V\beta(n^*)\right]\erfc\!\left(\sqrt{\frac{\abs{\beta''(n^*)}V}{2}}\frac{(n-n^*)^2\abs{\beta''(n^*)}+2\left(f-\frac{1}{2}\right)\beta(n^*)}{(n^*-n)\abs{\beta''(n^*)}}\right)\right.\nonumber \\
    &&\hspace{1.3cm}\left.+\exp\!\left[-2\left(f-\frac{1}{2}\right)V\beta(n^*)\right]\erfc\!\left(\sqrt{\frac{\abs{\beta''(n^*)}V}{2}}\frac{(n-n^*)^2\abs{\beta''(n^*)}-2\left(f-\frac{1}{2}\right)\beta(n^*)}{(n^*-n)\abs{\beta''(n^*)}}\right) \right]+o(1)\,.
\end{eqnarray}
For $n-n^*>0$, we have $n_A<n_\text{crit}$ implying $Y>0$. The integrals swap places, which yields the same result as Eq.~\eqref{eq:case1-X1} but with $n^*-n$ replaced by $n-n^*$. We can thus describe both cases in a single formula
\begin{eqnarray}
    \label{eq:resolved_Kronecker_delta}
    X_1\!\!&=&\!\!\frac{1}{2}\exp\!\left[\frac{(n-n^*)^2}{2}V\abs{\beta''(n^*)}\right]\!\left[\exp\!\left[2\left(f-\frac{1}{2}\right)V\beta(n^*)\right]\erfc\!\left(\sqrt{\frac{\abs{\beta''(n^*)}V}{2}}\frac{(n-n^*)^2\abs{\beta''(n^*)}+2\left(f-\frac{1}{2}\right)\beta(n^*)}{\abs{(n-n^*)\beta''(n^*)}}\right)\right.\nonumber \\
    &&\hspace{1.3cm}\left.+\exp\!\left[-2\left(f-\frac{1}{2}\right)V\beta(n^*)\right]\erfc\!\left(\sqrt{\frac{\abs{\beta''(n^*)}V}{2}}\frac{(n-n^*)^2\abs{\beta''(n^*)}-2\left(f-\frac{1}{2}\right)\beta(n^*)}{\abs{(n-n^*)\beta''(n^*)}}\right) \right]+o(1)\,,
\end{eqnarray}
which is valid for any $n-n^*\neq 0$.

\begin{figure*}
\centering
\begin{tikzpicture}
\draw (0,0) node{\includegraphics[width=.7\linewidth]{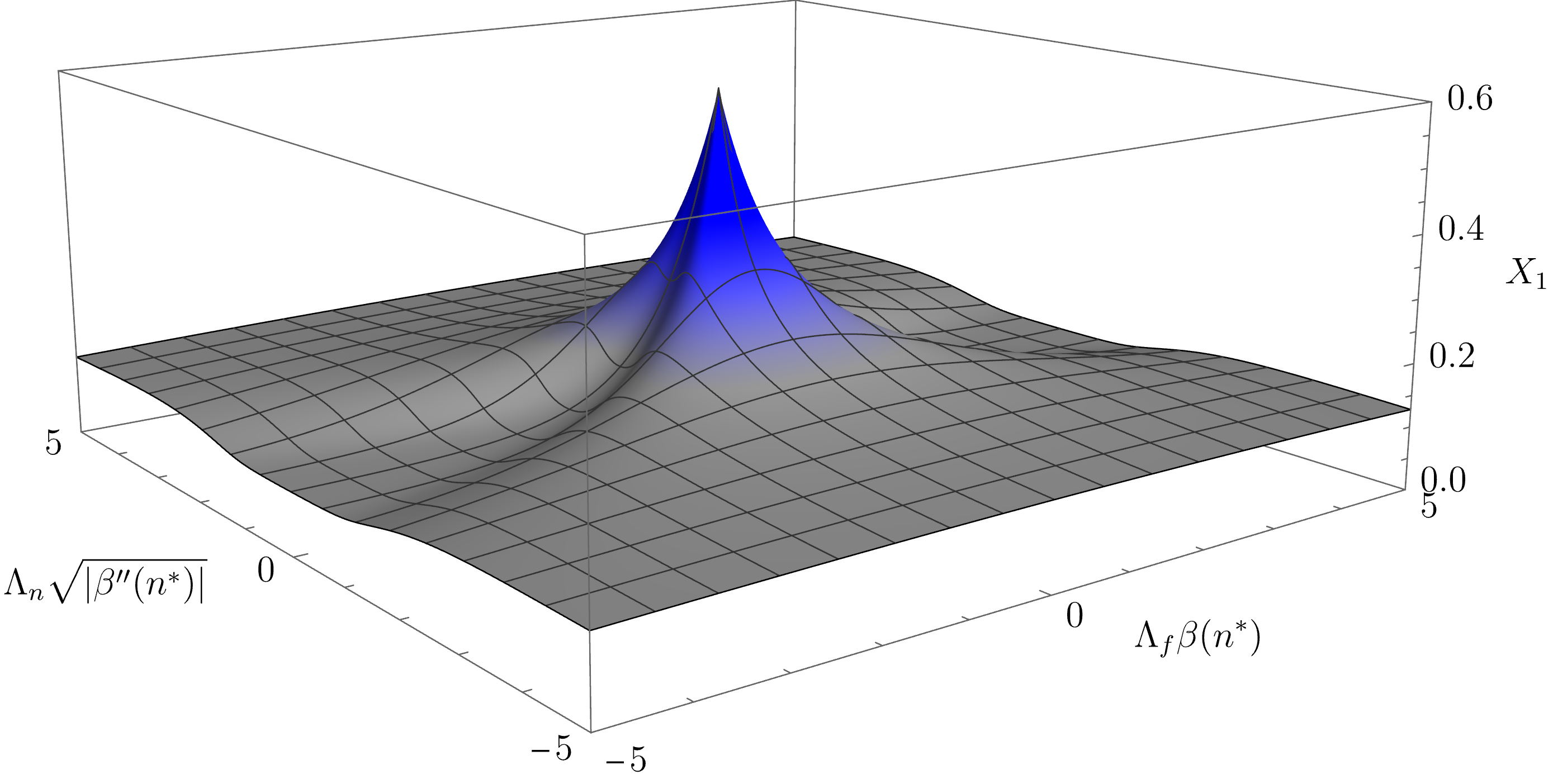}};
\end{tikzpicture}
\caption{\emph{Resolving the Kronecker $\delta$ at $f=1/2$ and $n=n^*$.} This is the case $s=1$ and $t=1/2$ from Eq.~\eqref{eq:X1-full-resolution}. We plot $X_1$ as a function of $\Lambda_f \beta(n^*)$ and $\Lambda_n\sqrt{|\beta''(n^*)|}$, which makes the functional form independent of the specifics of the system.}
\label{fig:X1-full-resolution}
\end{figure*}

Clearly, Eq.~\eqref{eq:resolved_Kronecker_delta} is nonzero only if we have simultaneously $\lim_{V\to\infty}n=n^*$ and $\lim_{V\to\infty}f=1/2$, as only in this double scaling limit (where also $n$ and $f$ have an implicit $V$ dependence) we will be near the Kronecker $\delta$. However, there are many ways to approach these limits, which is why we analyze different power laws
\begin{align}
        f=\frac{1}{2}+\frac{\Lambda_f}{V^s}\quad\text{and}\quad n=n^*+\frac{\Lambda_n}{V^t}\,,\label{eq:powerlaws-X1}
\end{align}
where $s>0$ and $t>0$ are the respective powers and $\Lambda_f$ and $\Lambda_n$ are free real parameters allowing us to map out the neighborhood around the Kronecker $\delta$ in this double scaling limit. Plugging Eq.~\eqref{eq:powerlaws-X1} into Eq.~\eqref{eq:resolved_Kronecker_delta} yields
\begin{eqnarray}
\label{eq:resolved-Kronecker-scaling}
    X_1&=&\frac{1}{2}\exp\left[\frac{\Lambda_n^2}{2}V^{1-2t}\abs{\beta''(n^*)}\right]\left[\exp\left[2\Lambda_fV^{1-s}\beta(n^*)\right]\erfc\left(\sqrt{\frac{\abs{\beta''(n^*)}}{2}}\left(\abs{\Lambda_n}V^{\frac{1}{2}-t}+\frac{2\Lambda_f}{\abs{\Lambda_n}}V^{\frac{1}{2}+t-s}\frac{\beta(n^*)}{\abs{\beta''(n^*)}}\right)\right)\right.\nonumber \\
    &&\hspace{1.5cm}\left.+\exp\left[-2\Lambda_fV^{1-s}\beta(n^*)\right]\erfc\left(\sqrt{\frac{\abs{\beta''(n^*)}}{2}}\left(\abs{\Lambda_n}V^{\frac{1}{2}-t}-\frac{2\Lambda_f}{\abs{\Lambda_n}}V^{\frac{1}{2}+t-s}\frac{\beta(n^*)}{\abs{\beta''(n^*)}}\right)\right) \right]+o(1)\,,
\end{eqnarray}
which can be simplified by considering different regimes for the power parameters $s$ and $t$. We find
\begin{align}\label{eq:X1-full-resolution}
    X_1=\begin{cases}
    0 & s<1\text{ or }t<\frac{1}{2}\\[1mm] \!\!
    {\Large\substack{\frac{1}{2}\exp\left[\frac{\Lambda_n^2}{2}\abs{\beta''(n^*)}\right]\left[\exp\left[2\Lambda_f\beta(n^*)\right]\erfc\left(\sqrt{\frac{\abs{\beta''(n^*)}}{2}}\frac{\Lambda_n^2\abs{\beta''(n^*)}+2\Lambda_f\beta(n^*)}{\abs{\Lambda_n\beta''(n^*)}}\right)\right.\\[1mm]
    \left.\hspace{2.6cm}+ \exp\left[-2\Lambda_f\beta(n^*)\right]\erfc\left(\sqrt{\frac{\abs{\beta''(n^*)}}{2}}\frac{\Lambda_n^2\abs{\beta''(n^*)}-2\Lambda_f\beta(n^*)}{\abs{\Lambda_n\beta''(n^*)}}\right)\right]}}  & s=1\text{ and }t=\frac{1}{2}\\[1mm]
    \exp\left[\frac{\Lambda_n^2}{2}\abs{\beta''(n^*)}\right]\erfc\left(\sqrt{\frac{\abs{\beta''(n^*)}}{2}}\abs{\Lambda_n}\right) & s>1\text{ and }t=\frac{1}{2}\\[1mm]
    \exp\left[-2\Lambda_f\beta(n^*)\right]  & s=1\text{ and }t>\frac{1}{2}\\[1mm]
    1 & s>1\text{ and }t>\frac{1}{2}
    \end{cases}
\end{align}
\end{widetext}
We note that the most interesting case corresponds to $s=1$ and $t=1/2$ shown in Fig.~\ref{fig:X1-full-resolution}, from which the other cases can be deduced by taking the appropriate limits $\Lambda_f\to \infty$ or $\Lambda_n\to \infty$ for $s<1$ or $t<1/2$, respectively, and $\Lambda_f\to 0$ or $\Lambda_n\to 0$ for $s>1$ or $t>1/2$, respectively. The function is mirror-symmetric with respect to both $\Lambda_f$ and $\Lambda_n$.

\subsection{Kronecker $\delta$ at $f=\frac{1}{2}$}
We now consider the effect of $\varphi_2=\ln(\min[\tfrac{d_N}{d_A},\tfrac{d_N}{d_B}])+o(1)$ from Eq.~\eqref{eq:varphi_2_definition} by studying
\begin{align}
    I&=\sum_{N_A}\varrho_{N_A}\varphi_2 \\
    &=\int_{-\infty}^\infty \varrho(n_A)\ln(\min\left[\frac{d_N}{d_A},\frac{d_N}{d_B}\right])\dd{n_A}+o(1)\,.
\end{align}
Recall that $n_\text{crit}$ is the point where $d_A(n_\text{crit})=d_B(n-n_\text{crit})$. One trick to evaluating this integral is to split it into two integrals around $n_\text{crit}$ to deal with the minimum function. Without loss of generality, we may restrict our analysis to $f<1/2$, as symmetry arguments will cover the case $f>1/2$.

To determine which dimension is larger in the splitting of the integral, we again study the sign of $Y$ as defined in Eq.~\eqref{eq:def-Y}. The argument is as follows:
\begin{itemize}
    \item At the mean of the Gaussian, Eq.~\eqref{eq:Y_expansion_fn} reduces to $Y\rvert_{n_A=fn}=2(f-1/2)\beta(n)$, so $Y<0 \implies d_A<d_B$.
    \item Using Eq.~\eqref{eq:ncrit_expansion_half_nstar}, we see that $n_\text{crit}>fn$ for $n<n^*$ and $n_\text{crit}<fn$ for $n>n^*$.
    \item The dimension inequality can only flip at $n_\text{crit}$. Put differently, one dimension is always larger for $n<n_\text{crit}$ and the other dimension is larger for $n>n_\text{crit}$.
    \item Thus we can conclude that when $n<n^*$, $d_A<d_B$ in the region $n_A<n_\text{crit}$. When $n>n^*$, $d_A>d_B$ in the region $n_A<n_\text{crit}$. The other dimension is larger for $n_A>n_\text{crit}$ in both cases.
\end{itemize}
The upshot of this analysis is that it enables us to write the integral for $n<n^*$ as
\begin{align}
    I=\int_{-\infty}^{n_\text{crit}}\!\! \varrho(n_A)\ln\!\left(\frac{d_N}{d_B}\right)\!\dd{n_A}+\int^{\infty}_{n_\text{crit}}\!\! \varrho(n_A)\ln\!\left(\frac{d_N}{d_A}\right)\!\dd{n_A},
\end{align}
and for $n>n^*$ as
\begin{align}
    I=\int_{-\infty}^{n_\text{crit}}\!\! \varrho(n_A)\ln\!\left(\frac{d_N}{d_A}\right)\!\dd{n_A}+\int^{\infty}_{n_\text{crit}}\!\! \varrho(n_A)\ln\!\left(\frac{d_N}{d_B}\right)\!\dd{n_A},
\end{align}
or, more compactly, as
\begin{widetext}
\begin{align}
    I&=\int_{-\infty}^\infty \varrho(n_A)\ln\left(\frac{d_N}{d_B}\right)\dd{n_A}+\begin{cases}
        \int^{\infty}_{n_\text{crit}} \varrho(n_A)\ln\left(\frac{d_B}{d_A}\right)\dd{n_A}\,, & n<n^* \\
        \int_{-\infty}^{n_\text{crit}} \varrho(n_A)\ln\left(\frac{d_B}{d_A}\right)\dd{n_A}\,, & n>n^*
    \end{cases} \\
    &=\int_{-\infty}^\infty \varrho(n_A)\ln\left(\frac{d_N}{d_B}\right)\dd{n_A}+\int_{-\infty}^\infty \varrho(n_A)x_2\dd{n_A} \,.
\end{align}
Here we see why it was useful to introduce $x_2$ in Eq.~\eqref{eq:x2_definition}. We can evaluate the first integral via the expansion
\begin{equation}
    \ln\left(\frac{d_N}{d_B}\right)=V\left[f\beta(n)+\beta'(n)(n_A-fn)-\frac{\beta''(n)}{2(1-f)}(n_A-fn)^2+O(n_A-fn)^3\right]+\frac{\ln(1-f)}{2}\,,
\end{equation}
\end{widetext}
and integrating against $\varrho(n_A)$ yields the leading-order and $O(1)$ term from Eq.~\eqref{eq:entropy_N_asymptotic}, which is given by the expression
\begin{equation}
    \int_{-\infty}^\infty\!\! \varrho(n_A)\ln\!\left(\frac{d_N}{d_B}\right)\!\dd{n_A}=Vf\beta(n)+\frac{f+\ln(1-f)}{2}+o(1)\,.
\end{equation}
The Kronecker $\delta$ is then encoded in the remaining integral
\begin{align}
    \begin{split}
    \label{eq:X2_definition}
    X_2&=\int_{-\infty}^\infty \varrho(n_A)x_2\dd{n_A} \\
    &=\begin{cases}
        \int^{\infty}_{n_\text{crit}} \varrho(n_A)\ln\left(\frac{d_B}{d_A}\right)\dd{n_A}\,, & n<n^* \\
        \int_{-\infty}^{n_\text{crit}} \varrho(n_A)\ln\left(\frac{d_B}{d_A}\right)\dd{n_A}\,, & n>n^*
    \end{cases}\,,
    \end{split}
\end{align}
where we need to take into account the scaling of the end-points. In particular, we require that $\delta n_\text{crit}= n_\text{crit}-fn=O(1/\sqrt{V})$, otherwise the exponential suppression of the Gaussian will cause the integral to vanish. Hence, from the definition of $n_\text{crit}$ that $d_A(n_\text{crit})=d_B(n-n_\text{crit})$, we expand to linear order and solve for $f-1/2$. We have essentially already studied this equation and we see that it is equivalent to
\begin{eqnarray}
    \label{eq:kronecker_log_f_scaling}
    &&\!\!\!\!Y\rvert_{n_A=n_\text{crit}}\!\!=2(f\!-\!\tfrac{1}{2})\beta(\frac{n_\text{crit}}{f})\!+\!2\beta'(\frac{n_\text{crit}}{f})(n_\text{crit}\!-\!fn)=0,\nonumber \\
    &&\!\!\!\!\text{and}\quad f-\frac{1}{2}=-\frac{\beta'(\frac{n_\text{crit}}{f})(n_\text{crit}-fn)}{\beta(\frac{n_\text{crit}}{f})}\,.
\end{eqnarray}
We need $f-1/2=O(1/\sqrt{V})$ in order for $X_2$ to not vanish. We ignored the quadratic term because it is retrospectively sub-leading to $O(1)$ and linear term, and would only contribute terms of order $o(1/\sqrt{V})$ to the r.h.s.~of Eq.~\eqref{eq:kronecker_log_f_scaling}. This yields an expression for $n_\text{crit}$, now for general $n$, valid only near $f=1/2$, as
\begin{align}
    n_\text{crit}=fn-\frac{\left(f-\frac{1}{2}\right)\beta(n)}{\beta'(n)}\ \text{and}\ \delta n_\text{crit}=-\frac{\left(f-\frac{1}{2}\right)\beta(n)}{\beta'(n)},
\end{align}
where we can expand to logarithm to find
\begin{eqnarray}
    &&\!\!\!\!\!\!\ln\left(\frac{d_B}{d_A}\right)=V\left[(1-f)\beta\left(\frac{n-n_A}{1-f}\right)-f\beta\left(\frac{n_A}{f}\right)\right] \\ &&\!\!\!\!\!\!=\!-2V\!\left[\!\left(f-\frac{1}{2}\right)\!\beta(n)\!+\!\beta'(n)(n_A-fn)\!+\!O(n_A-fn)^2\right]\!. \nonumber
\end{eqnarray}
Evaluating the integrals in Eq.~\eqref{eq:X2_definition} and after some algebra, we find
\begin{eqnarray}
   &&\!\!\!\!\!\!X_2=V|f-\tfrac{1}{2}|\beta(n)\erfc\left(\sqrt{2V\abs{\beta''(n)}}\frac{|f-\frac{1}{2}|\beta(n)}{\abs{\beta'(n)}}\right)\quad \nonumber \\ &&\!\!\!\!\!\! -\abs{\beta'(n)}\sqrt{\frac{V}{2\pi\abs{\beta''(n)}}}\exp\left[-2V\abs{\beta''(n)}\frac{\left(f-\frac{1}{2}\right)^2\beta(n)^2}{\abs{\beta'(n)}^2}\right],\nonumber \\ \label{eq:Kronecker-delta-resolved2}  
\end{eqnarray}
where the sign change arises from using the symmetry of the entanglement entropy, as $f-1/2$ will be nonpositive for $f\leq1/2$. Equation~\eqref{eq:Kronecker-delta-resolved2} resolves the Kronecker $\delta$ associated to the term of order $\sqrt{V}$. Note, however, that the first summand could have also different power laws depending on how $f-1/2$ scales with $V$, as we will see in a moment.

In analogy to Eq.~\eqref{eq:resolved-Kronecker-scaling}, we can plug a general power law scaling $f=1/2+\Lambda_f/V^s$ into Eq.~\eqref{eq:Kronecker-delta-resolved2} to find
\begin{eqnarray}\label{eq:resolved-Kronecker-scaling2}
    &&\!\!\!\! X_2=|\Lambda_f|V^{1-s}\beta(n)\erfc\left(\sqrt{2\abs{\beta''(n)}}\frac{|\Lambda_f|\beta(n)}{\abs{\beta'(n)}}V^{\frac{1}{2}-s}\right)\nonumber \\ &&\!\!\!\! -\abs{\beta'(n)}\sqrt{\frac{V}{2\pi\abs{\beta''(n)}}}\exp\left[-2\abs{\beta''(n)}\frac{\Lambda_f^2\beta(n)^2}{\abs{\beta'(n)}^2}V^{1-2s}\right]. \nonumber \\
\end{eqnarray}
Again, we can consider the various power laws to find
\begin{widetext}
\begin{align}
    X_2=\begin{cases}
        0 & s<\frac{1}{2}\\[1mm]
    \sqrt{V}\left(|\Lambda_f|\beta(n)\erfc\left(\sqrt{2\abs{\beta''(n)}}\frac{|\Lambda_f|\beta(n)}{\abs{\beta'(n)}}\right)-\abs{\beta'(n)}\sqrt{\frac{1}{2\pi\abs{\beta''(n)}}}\exp\left[-2\abs{\beta''(n)}\frac{\Lambda_f^2\beta(n)^2}{\abs{\beta'(n)}^2}\right]\right) & s=\frac{1}{2}\\[1mm]
    |\Lambda_f| V^{1-s}\beta(n)-\abs{\beta'(n)}\sqrt{\frac{V}{2\pi\abs{\beta''(n)}}} & \frac{1}{2}<s\leq1\\[1mm]
    -\abs{\beta'(n)}\sqrt{\frac{V}{2\pi\abs{\beta''(n)}}} & s> 1
    \end{cases}
\end{align}
\end{widetext}
where we ignore any terms of order $o(1)$. Here, $s=1/2$ is the most interesting case and, again, we can get the other limits by taking $\Lambda_f$ to zero or infinity. An interesting effect for $s>1/2$ is that the term $|\Lambda_f|V^{1-s}\beta(n)$ will exactly cancel the respective contribution from the leading order term $V\beta(n)\min(f,1-f)=V\beta(n)(\frac{1}{2}-\frac{|\Lambda_f|}{V^s})$, such that there will not be a term proportional to $V^{1-s}$ for $1/2<s<1$.

\section{Distinguishable particles}\label{app:case-study-distinguishable}
In this case study we retain our previous setup of a set of $V$ sites, among which we place $N$ particles. Let us assume that each site can hold an arbitrary number of particles. However, we now treat the particles as distinguishable, which means that it matters which particle is placed on which site. We label particles by elements of the set $U=\left\{ 1, 2, \dots, N\right\}$ and let $P\subseteq U$ represent a subset of particles.

Bipartitioning the system into subsystem $A$ and $B$ yields a Hilbert space of fixed particle number decomposed as
\begin{equation}
    \mathcal{H}^{(N)}=\bigoplus_{\abs{P}=0}^N \mathcal{H}_{A}^{(P)} \otimes \mathcal{H}_{B}^{(U\setminus P)},
\end{equation}
where the direct sum is over all possible subsets of particles $P$ containing $0\leq \abs{P}\leq N$ particles. Here, $\mathcal{H}_A^{(P)}$ denotes the Hilbert space describing the particles of $P$ to be in subsystem $A$. The remaining particles, $U\setminus P$, are then in subsystem $B$, described by the Hilbert space $\mathcal{H}_B^{(U\setminus P)}$.

In distributing $N_A$ distinguishable particles over $V_A$ sites, we first need to choose $N_A$ particles out of a total of $N$. This additional step introduces a binomial coefficient into the dimension, such that the Hilbert space dimension is
\begin{equation}
    \label{eq:dN_relation_dA_dB_distinguishable}
    d_N\!=\!\dim\mathcal{H}^{(N)}\!=\!\sum_{N_A=0}^N \binom{N}{N_A}\,d_A(N_A) d_B(N-N_A)\!=\!V^N\!.
\end{equation}
where we recognized that $d_N$ is the number of ways to place $N$ distinguishable particles over $V$ distinguishable sites, since from each particle's perspective there are $V$ sites to choose from (as there is no restriction on how many particles a site can hold). Similarly, we have $d_A(N_A)=V_A^{N_A}$ and $d_B(N_B)=(V-V_A)^{N_B}$.

The average entanglement entropy can be computed in analogy to Eq.~\eqref{eq:entropy_fixed_N_exact_sum}, containing an extra binomial factor, as
\begin{widetext}
\begin{equation}
\label{eq:entropy_distinguishable_relation_varrho_varphi}
    \braket{S_A}_N=\sum^{N}_{N_A=0}\underbrace{\binom{N}{N_A}\frac{d_Ad_B}{d_N}}_{\varrho_{N_A}}\underbrace{\left(\Psi(d_N+1)-\Psi(\max(d_A,d_B)+1)-\min(\tfrac{d_A-1}{2d_B},\tfrac{d_B-1}{2d_A})\right)}_{\varphi_{N_A}}\,,
\end{equation}
\end{widetext}
where we introduced the probability function $\varrho_{N_A}$ obeying $\sum_{N_A=0}^N\varrho_{N_A}=1$ and $\varphi_{N_A}$ in close analogy to Eq.~\eqref{eq:entropy_fixed_N_exact_sum}. Again, we evaluate this sum by approximating it as an integral in the quasi-continuous variable $n_A=N_A/V$ and identifying the saddle point of the density function $\varrho(n_A)=V \varrho_{Vn_A}$. We find that
\begin{equation}
    \label{eq:varrho_distinguishable}
    \varrho(n_A)=\frac{1}{\sqrt{2\pi f(1-f)n V}}\exp\left[-\frac{V}{2}\frac{(n_A-f n)^2}{f(1-f)n}\right]\,,
\end{equation}
which is simply the Gaussian approximation to the binomial distribution with a mean $\bar{n}_A=f n$.

The function $\varphi(n_A)=\varphi_{Vn_A}$ is a piecewise function with nonanalycity at the point $N_\text{crit}$, defined by the condition $d_A(N_\text{crit})=d_B(N-N_\text{crit})$. We can solve for it explicitly and find
\begin{align}
        N_\text{crit}\!=\!\frac{N\ln(V-V_A)}{\ln(V-V_A)+\ln(V_A)}\ \text{and}\ 
        n_\text{crit}\!=\!\frac{n\ln[(1-f)V]}{\ln[f(1-f)V^2]}.
\end{align}
Without loss of generality, we restrict to the case $f\leq \frac{1}{2}$, as the entanglement entropy is symmetric under $f\to 1-f$. With this assumption, we clearly have at $\bar{n}_A=f n$ the inequality $d_A(\bar{N}_A)=(fV)^{fnV}\leq ((1-f)V)^{(1-f)nV}=d_B(\bar{N}_B)$ with $\bar{N}_A=\bar{n}_AV$ and $\bar{N}_B=N-\bar{N}_A$. Therefore, at leading order it suffices to use $\varphi_{N_A}=\Psi(d_N+1)-\Psi(d_B+1)$ when evaluating the integral. We note that for $f<\frac{1}{2}$, we have $\delta n_{\mathrm{crit}}>0$, which implies that we can just integrate against $\varphi(n_A)$ for $n_A<n_{\mathrm{crit}}$. We can ignore the term $\min(\tfrac{d_A-1}{2d_B},\tfrac{d_B-1}{2d_A})$ inside $\varphi(n_A)$ as its integration against $\varrho(n_A)$ will be subleading of order $o(1)$. Therefore, our evaluation can use $\varphi(n_A)=\varphi_{n_A V}=\ln(d_N/d_B)+o(1)$, where we used the first-order approximation $\Psi(x)=\ln(x)+o(1)$. Plugging in the expressions for the appropriate dimensions from Eq.~\eqref{eq:dN_relation_dA_dB_distinguishable} and simplifying, we find
\begin{equation}
    \label{eq:varphi_distinguishable}
    \varphi(n_A)=n_AV\ln(V)-(n-n_A)V\ln(1-f)+o(1).
\end{equation}
If $f=\frac{1}{2}$, the nonanalycity of $\varphi_{N_A}$ at $n_\mathrm{crit}=\frac{n}{2}$ coincides with the peak of the Gaussian at $\bar{n}_A=\frac{n}{2}$, which means that we must break the integral approximation of Eq.~\eqref{eq:entropy_distinguishable_relation_varrho_varphi} into two integrals for $n_A<\frac{n}{2}$ and $n_A>\frac{n}{2}$. This yields a contribution of the order $\sqrt{V}\ln(V)$.

Finally, we combine Eqs.~\eqref{eq:varrho_distinguishable} and~\eqref{eq:varphi_distinguishable} to obtain the average entanglement entropy
\begin{eqnarray}
    \braket{S_A}_N&=&nf \,V\ln{V}-n(1-f)\ln(1-f)V \nonumber \\ &&+\sqrt{\frac{n}{2\pi}}\ln(2)\,\delta_{f,\frac{1}{2}}\sqrt{V}\ln{V}+o(1),
\end{eqnarray}
for $0\leq f\leq \frac{1}{2}$. One immediately sees that the leading order is not volume-law, it grows as $V\ln(V)$. It is still linear in $f$ leading to the typical Page curve triangle. In particular, comparing with Eq.~\eqref{eq:entropy_N_asymptotic}, there is no $\frac{f+\ln(1-f)}{2}$ term, which we showed to be universal for indistinguishable particle systems. We find a Kronecker $\delta$ at $f=\frac{1}{2}$ with a $\sqrt{V} \ln{V}$ prefactor so, compared to the case of indistinguishable particles, both the $V$ and the $\sqrt{V}$ terms are logarithmically enhanced. The Kronecker $\delta$ may be resolved at $f=\frac{1}{2}$ using the techniques outlined in Appendix~\ref{app:resolution_kronecker_delta}.

\bibliography{references}

\begin{thebibliography}{56}%
\makeatletter
\providecommand \@ifxundefined [1]{%
 \@ifx{#1\undefined}
}%
\providecommand \@ifnum [1]{%
 \ifnum #1\expandafter \@firstoftwo
 \else \expandafter \@secondoftwo
 \fi
}%
\providecommand \@ifx [1]{%
 \ifx #1\expandafter \@firstoftwo
 \else \expandafter \@secondoftwo
 \fi
}%
\providecommand \natexlab [1]{#1}%
\providecommand \enquote  [1]{``#1''}%
\providecommand \bibnamefont  [1]{#1}%
\providecommand \bibfnamefont [1]{#1}%
\providecommand \citenamefont [1]{#1}%
\providecommand \href@noop [0]{\@secondoftwo}%
\providecommand \href [0]{\begingroup \@sanitize@url \@href}%
\providecommand \@href[1]{\@@startlink{#1}\@@href}%
\providecommand \@@href[1]{\endgroup#1\@@endlink}%
\providecommand \@sanitize@url [0]{\catcode `\\12\catcode `\$12\catcode
  `\&12\catcode `\#12\catcode `\^12\catcode `\_12\catcode `\%12\relax}%
\providecommand \@@startlink[1]{}%
\providecommand \@@endlink[0]{}%
\providecommand \url  [0]{\begingroup\@sanitize@url \@url }%
\providecommand \@url [1]{\endgroup\@href {#1}{\urlprefix }}%
\providecommand \urlprefix  [0]{URL }%
\providecommand \Eprint [0]{\href }%
\providecommand \doibase [0]{https://doi.org/}%
\providecommand \selectlanguage [0]{\@gobble}%
\providecommand \bibinfo  [0]{\@secondoftwo}%
\providecommand \bibfield  [0]{\@secondoftwo}%
\providecommand \translation [1]{[#1]}%
\providecommand \BibitemOpen [0]{}%
\providecommand \bibitemStop [0]{}%
\providecommand \bibitemNoStop [0]{.\EOS\space}%
\providecommand \EOS [0]{\spacefactor3000\relax}%
\providecommand \BibitemShut  [1]{\csname bibitem#1\endcsname}%
\let\auto@bib@innerbib\@empty
\bibitem [{\citenamefont {Eisert}\ \emph {et~al.}(2007)\citenamefont {Eisert},
  \citenamefont {Brandão},\ and\ \citenamefont
  {Audenaert}}]{eisert2007quantitative}%
  \BibitemOpen
  \bibfield  {author} {\bibinfo {author} {\bibfnamefont {J.}~\bibnamefont
  {Eisert}}, \bibinfo {author} {\bibfnamefont {F.~G. S.~L.}\ \bibnamefont
  {Brandão}},\ and\ \bibinfo {author} {\bibfnamefont {K.~M.~R.}\ \bibnamefont
  {Audenaert}},\ }\bibfield  {title} {\bibinfo {title} {Quantitative
  entanglement witnesses},\ }\href {https://doi.org/10.1088/1367-2630/9/3/046}
  {\bibfield  {journal} {\bibinfo  {journal} {New J. Phys.}\ }\textbf {\bibinfo
  {volume} {9}},\ \bibinfo {pages} {46} (\bibinfo {year} {2007})}\BibitemShut
  {NoStop}%
\bibitem [{\citenamefont {Zheng}\ and\ \citenamefont
  {Guo}(2000)}]{zheng2000efficient}%
  \BibitemOpen
  \bibfield  {author} {\bibinfo {author} {\bibfnamefont {S.-B.}\ \bibnamefont
  {Zheng}}\ and\ \bibinfo {author} {\bibfnamefont {G.-C.}\ \bibnamefont
  {Guo}},\ }\bibfield  {title} {\bibinfo {title} {Efficient scheme for two-atom
  entanglement and quantum information processing in cavity {QED}},\ }\href
  {https://doi.org/10.1103/PhysRevLett.85.2392} {\bibfield  {journal} {\bibinfo
   {journal} {Phys. Rev. Lett.}\ }\textbf {\bibinfo {volume} {85}},\ \bibinfo
  {pages} {2392} (\bibinfo {year} {2000})}\BibitemShut {NoStop}%
\bibitem [{\citenamefont {Pollmann}\ \emph {et~al.}(2010)\citenamefont
  {Pollmann}, \citenamefont {Turner}, \citenamefont {Berg},\ and\ \citenamefont
  {Oshikawa}}]{pollmann2010entanglement}%
  \BibitemOpen
  \bibfield  {author} {\bibinfo {author} {\bibfnamefont {F.}~\bibnamefont
  {Pollmann}}, \bibinfo {author} {\bibfnamefont {A.~M.}\ \bibnamefont
  {Turner}}, \bibinfo {author} {\bibfnamefont {E.}~\bibnamefont {Berg}},\ and\
  \bibinfo {author} {\bibfnamefont {M.}~\bibnamefont {Oshikawa}},\ }\bibfield
  {title} {\bibinfo {title} {Entanglement spectrum of a topological phase in
  one dimension},\ }\href {https://doi.org/10.1103/PhysRevB.81.064439}
  {\bibfield  {journal} {\bibinfo  {journal} {Phys. Rev. B}\ }\textbf {\bibinfo
  {volume} {81}},\ \bibinfo {pages} {064439} (\bibinfo {year}
  {2010})}\BibitemShut {NoStop}%
\bibitem [{\citenamefont {Page}(1993{\natexlab{a}})}]{page1993information}%
  \BibitemOpen
  \bibfield  {author} {\bibinfo {author} {\bibfnamefont {D.~N.}\ \bibnamefont
  {Page}},\ }\bibfield  {title} {\bibinfo {title} {Information in black hole
  radiation},\ }\href {https://doi.org/10.1103/PhysRevLett.71.3743} {\bibfield
  {journal} {\bibinfo  {journal} {Phys. Rev. Lett.}\ }\textbf {\bibinfo
  {volume} {71}},\ \bibinfo {pages} {3743} (\bibinfo {year}
  {1993}{\natexlab{a}})}\BibitemShut {NoStop}%
\bibitem [{\citenamefont {Ryu}\ and\ \citenamefont
  {Takayanagi}(2006)}]{ryu2006aspects}%
  \BibitemOpen
  \bibfield  {author} {\bibinfo {author} {\bibfnamefont {S.}~\bibnamefont
  {Ryu}}\ and\ \bibinfo {author} {\bibfnamefont {T.}~\bibnamefont
  {Takayanagi}},\ }\bibfield  {title} {\bibinfo {title} {Aspects of holographic
  entanglement entropy},\ }\href
  {https://doi.org/10.1088/1126-6708/2006/08/045} {\bibfield  {journal}
  {\bibinfo  {journal} {J. High Energy Phys.}\ }\textbf {\bibinfo {volume}
  {2006}}\bibinfo  {number} { (08)},\ \bibinfo {pages} {045}}\BibitemShut
  {NoStop}%
\bibitem [{\citenamefont {Mej\'{\i}a-Monasterio}\ \emph
  {et~al.}(2005)\citenamefont {Mej\'{\i}a-Monasterio}, \citenamefont {Benenti},
  \citenamefont {Carlo},\ and\ \citenamefont {Casati}}]{mejia_05}%
  \BibitemOpen
\bibfield  {number} {  }\bibfield  {author} {\bibinfo {author} {\bibfnamefont
  {C.}~\bibnamefont {Mej\'{\i}a-Monasterio}}, \bibinfo {author} {\bibfnamefont
  {G.}~\bibnamefont {Benenti}}, \bibinfo {author} {\bibfnamefont {G.~G.}\
  \bibnamefont {Carlo}},\ and\ \bibinfo {author} {\bibfnamefont
  {G.}~\bibnamefont {Casati}},\ }\bibfield  {title} {\bibinfo {title}
  {Entanglement across a transition to quantum chaos},\ }\href
  {https://doi.org/10.1103/PhysRevA.71.062324} {\bibfield  {journal} {\bibinfo
  {journal} {Phys. Rev. A}\ }\textbf {\bibinfo {volume} {71}},\ \bibinfo
  {pages} {062324} (\bibinfo {year} {2005})}\BibitemShut {NoStop}%
\bibitem [{\citenamefont {Santos}\ \emph {et~al.}(2012)\citenamefont {Santos},
  \citenamefont {Polkovnikov},\ and\ \citenamefont {Rigol}}]{santos_12}%
  \BibitemOpen
  \bibfield  {author} {\bibinfo {author} {\bibfnamefont {L.~F.}\ \bibnamefont
  {Santos}}, \bibinfo {author} {\bibfnamefont {A.}~\bibnamefont
  {Polkovnikov}},\ and\ \bibinfo {author} {\bibfnamefont {M.}~\bibnamefont
  {Rigol}},\ }\bibfield  {title} {\bibinfo {title} {Weak and strong typicality
  in quantum systems},\ }\href {https://doi.org/10.1103/PhysRevE.86.010102}
  {\bibfield  {journal} {\bibinfo  {journal} {Phys. Rev. E}\ }\textbf {\bibinfo
  {volume} {86}},\ \bibinfo {pages} {010102} (\bibinfo {year}
  {2012})}\BibitemShut {NoStop}%
\bibitem [{\citenamefont {Deutsch}\ \emph {et~al.}(2013)\citenamefont
  {Deutsch}, \citenamefont {Li},\ and\ \citenamefont {Sharma}}]{deutsch_li_13}%
  \BibitemOpen
  \bibfield  {author} {\bibinfo {author} {\bibfnamefont {J.~M.}\ \bibnamefont
  {Deutsch}}, \bibinfo {author} {\bibfnamefont {H.}~\bibnamefont {Li}},\ and\
  \bibinfo {author} {\bibfnamefont {A.}~\bibnamefont {Sharma}},\ }\bibfield
  {title} {\bibinfo {title} {Microscopic origin of thermodynamic entropy in
  isolated systems},\ }\href {https://doi.org/10.1103/PhysRevE.87.042135}
  {\bibfield  {journal} {\bibinfo  {journal} {Phys. Rev. E}\ }\textbf {\bibinfo
  {volume} {87}},\ \bibinfo {pages} {042135} (\bibinfo {year}
  {2013})}\BibitemShut {NoStop}%
\bibitem [{\citenamefont {Beugeling}\ \emph {et~al.}(2015)\citenamefont
  {Beugeling}, \citenamefont {Andreanov},\ and\ \citenamefont
  {Haque}}]{beugeling_andreanov_15}%
  \BibitemOpen
  \bibfield  {author} {\bibinfo {author} {\bibfnamefont {W.}~\bibnamefont
  {Beugeling}}, \bibinfo {author} {\bibfnamefont {A.}~\bibnamefont
  {Andreanov}},\ and\ \bibinfo {author} {\bibfnamefont {M.}~\bibnamefont
  {Haque}},\ }\bibfield  {title} {\bibinfo {title} {Global characteristics of
  all eigenstates of local many-body {H}amiltonians: participation ratio and
  entanglement entropy},\ }\href
  {https://doi.org/10.1088/1742-5468/2015/02/P02002} {\bibfield  {journal}
  {\bibinfo  {journal} {J. Stat. Mech.}\ }\textbf {\bibinfo {volume} {{\rm
  (2015)}}},\ \bibinfo {pages} {P02002} (\bibinfo {year} {2015})}\BibitemShut
  {NoStop}%
\bibitem [{\citenamefont {Yang}\ \emph {et~al.}(2015)\citenamefont {Yang},
  \citenamefont {Chamon}, \citenamefont {Hamma},\ and\ \citenamefont
  {Mucciolo}}]{yang_chamon_15}%
  \BibitemOpen
  \bibfield  {author} {\bibinfo {author} {\bibfnamefont {Z.-C.}\ \bibnamefont
  {Yang}}, \bibinfo {author} {\bibfnamefont {C.}~\bibnamefont {Chamon}},
  \bibinfo {author} {\bibfnamefont {A.}~\bibnamefont {Hamma}},\ and\ \bibinfo
  {author} {\bibfnamefont {E.~R.}\ \bibnamefont {Mucciolo}},\ }\bibfield
  {title} {\bibinfo {title} {Two-component structure in the entanglement
  spectrum of highly excited states},\ }\href
  {https://doi.org/10.1103/PhysRevLett.115.267206} {\bibfield  {journal}
  {\bibinfo  {journal} {Phys. Rev. Lett.}\ }\textbf {\bibinfo {volume} {115}},\
  \bibinfo {pages} {267206} (\bibinfo {year} {2015})}\BibitemShut {NoStop}%
\bibitem [{\citenamefont {Vidmar}\ and\ \citenamefont
  {Rigol}(2017)}]{vidmar2017entanglement}%
  \BibitemOpen
  \bibfield  {author} {\bibinfo {author} {\bibfnamefont {L.}~\bibnamefont
  {Vidmar}}\ and\ \bibinfo {author} {\bibfnamefont {M.}~\bibnamefont {Rigol}},\
  }\bibfield  {title} {\bibinfo {title} {Entanglement entropy of eigenstates of
  quantum chaotic {H}amiltonians},\ }\href
  {https://doi.org/10.1103/PhysRevLett.119.220603} {\bibfield  {journal}
  {\bibinfo  {journal} {Phys. Rev. Lett.}\ }\textbf {\bibinfo {volume} {119}},\
  \bibinfo {pages} {220603} (\bibinfo {year} {2017})}\BibitemShut {NoStop}%
\bibitem [{\citenamefont {Dymarsky}\ \emph {et~al.}(2018)\citenamefont
  {Dymarsky}, \citenamefont {Lashkari},\ and\ \citenamefont
  {Liu}}]{dymarsky2018subsystem}%
  \BibitemOpen
  \bibfield  {author} {\bibinfo {author} {\bibfnamefont {A.}~\bibnamefont
  {Dymarsky}}, \bibinfo {author} {\bibfnamefont {N.}~\bibnamefont {Lashkari}},\
  and\ \bibinfo {author} {\bibfnamefont {H.}~\bibnamefont {Liu}},\ }\bibfield
  {title} {\bibinfo {title} {Subsystem eigenstate thermalization hypothesis},\
  }\href {https://doi.org/10.1103/PhysRevE.97.012140} {\bibfield  {journal}
  {\bibinfo  {journal} {Phys. Rev. E}\ }\textbf {\bibinfo {volume} {97}},\
  \bibinfo {pages} {012140} (\bibinfo {year} {2018})}\BibitemShut {NoStop}%
\bibitem [{\citenamefont {Garrison}\ and\ \citenamefont
  {Grover}(2018)}]{garrison2018does}%
  \BibitemOpen
  \bibfield  {author} {\bibinfo {author} {\bibfnamefont {J.~R.}\ \bibnamefont
  {Garrison}}\ and\ \bibinfo {author} {\bibfnamefont {T.}~\bibnamefont
  {Grover}},\ }\bibfield  {title} {\bibinfo {title} {Does a single eigenstate
  encode the full {H}amiltonian?},\ }\href
  {https://doi.org/10.1103/PhysRevX.8.021026} {\bibfield  {journal} {\bibinfo
  {journal} {Phys. Rev. X}\ }\textbf {\bibinfo {volume} {8}},\ \bibinfo {pages}
  {021026} (\bibinfo {year} {2018})}\BibitemShut {NoStop}%
\bibitem [{\citenamefont {Nakagawa}\ \emph {et~al.}(2018)\citenamefont
  {Nakagawa}, \citenamefont {Watanabe}, \citenamefont {Fujita},\ and\
  \citenamefont {Sugiura}}]{nakagawa_watanabe_18}%
  \BibitemOpen
  \bibfield  {author} {\bibinfo {author} {\bibfnamefont {Y.~O.}\ \bibnamefont
  {Nakagawa}}, \bibinfo {author} {\bibfnamefont {M.}~\bibnamefont {Watanabe}},
  \bibinfo {author} {\bibfnamefont {H.}~\bibnamefont {Fujita}},\ and\ \bibinfo
  {author} {\bibfnamefont {S.}~\bibnamefont {Sugiura}},\ }\bibfield  {title}
  {\bibinfo {title} {Universality in volume-law entanglement of scrambled pure
  quantum states},\ }\href {https://doi.org/10.1038/s41467-018-03883-9}
  {\bibfield  {journal} {\bibinfo  {journal} {Nat. Comm.}\ }\textbf {\bibinfo
  {volume} {9}},\ \bibinfo {pages} {1635} (\bibinfo {year} {2018})}\BibitemShut
  {NoStop}%
\bibitem [{\citenamefont {Liu}\ \emph {et~al.}(2018)\citenamefont {Liu},
  \citenamefont {Chen},\ and\ \citenamefont {Balents}}]{liu2018quantum}%
  \BibitemOpen
  \bibfield  {author} {\bibinfo {author} {\bibfnamefont {C.}~\bibnamefont
  {Liu}}, \bibinfo {author} {\bibfnamefont {X.}~\bibnamefont {Chen}},\ and\
  \bibinfo {author} {\bibfnamefont {L.}~\bibnamefont {Balents}},\ }\bibfield
  {title} {\bibinfo {title} {Quantum entanglement of the {Sachdev-Ye-Kitaev}
  models},\ }\href {https://doi.org/10.1103/PhysRevB.97.245126} {\bibfield
  {journal} {\bibinfo  {journal} {Phys. Rev. B}\ }\textbf {\bibinfo {volume}
  {97}},\ \bibinfo {pages} {245126} (\bibinfo {year} {2018})}\BibitemShut
  {NoStop}%
\bibitem [{\citenamefont {Lu}\ and\ \citenamefont
  {Grover}(2019)}]{lu_grover_19}%
  \BibitemOpen
  \bibfield  {author} {\bibinfo {author} {\bibfnamefont {T.-C.}\ \bibnamefont
  {Lu}}\ and\ \bibinfo {author} {\bibfnamefont {T.}~\bibnamefont {Grover}},\
  }\bibfield  {title} {\bibinfo {title} {Renyi entropy of chaotic
  eigenstates},\ }\href {https://doi.org/10.1103/PhysRevE.99.032111} {\bibfield
   {journal} {\bibinfo  {journal} {Phys. Rev. E}\ }\textbf {\bibinfo {volume}
  {99}},\ \bibinfo {pages} {032111} (\bibinfo {year} {2019})}\BibitemShut
  {NoStop}%
\bibitem [{\citenamefont {Murthy}\ and\ \citenamefont
  {Srednicki}(2019)}]{murthy_19}%
  \BibitemOpen
  \bibfield  {author} {\bibinfo {author} {\bibfnamefont {C.}~\bibnamefont
  {Murthy}}\ and\ \bibinfo {author} {\bibfnamefont {M.}~\bibnamefont
  {Srednicki}},\ }\bibfield  {title} {\bibinfo {title} {Structure of chaotic
  eigenstates and their entanglement entropy},\ }\href
  {https://doi.org/10.1103/PhysRevE.100.022131} {\bibfield  {journal} {\bibinfo
   {journal} {Phys. Rev. E}\ }\textbf {\bibinfo {volume} {100}},\ \bibinfo
  {pages} {022131} (\bibinfo {year} {2019})}\BibitemShut {NoStop}%
\bibitem [{\citenamefont {Huang}(2019)}]{huang2019universal}%
  \BibitemOpen
  \bibfield  {author} {\bibinfo {author} {\bibfnamefont {Y.}~\bibnamefont
  {Huang}},\ }\bibfield  {title} {\bibinfo {title} {{Universal eigenstate
  entanglement of chaotic local {H}amiltonians}},\ }\href
  {https://doi.org/https://doi.org/10.1016/j.nuclphysb.2018.09.013} {\bibfield
  {journal} {\bibinfo  {journal} {Nuc. Phys. B}\ }\textbf {\bibinfo {volume}
  {938}},\ \bibinfo {pages} {594 } (\bibinfo {year} {2019})}\BibitemShut
  {NoStop}%
\bibitem [{\citenamefont {LeBlond}\ \emph {et~al.}(2019)\citenamefont
  {LeBlond}, \citenamefont {Mallayya}, \citenamefont {Vidmar},\ and\
  \citenamefont {Rigol}}]{LeBlond_19}%
  \BibitemOpen
  \bibfield  {author} {\bibinfo {author} {\bibfnamefont {T.}~\bibnamefont
  {LeBlond}}, \bibinfo {author} {\bibfnamefont {K.}~\bibnamefont {Mallayya}},
  \bibinfo {author} {\bibfnamefont {L.}~\bibnamefont {Vidmar}},\ and\ \bibinfo
  {author} {\bibfnamefont {M.}~\bibnamefont {Rigol}},\ }\bibfield  {title}
  {\bibinfo {title} {Entanglement and matrix elements of observables in
  interacting integrable systems},\ }\href
  {https://doi.org/10.1103/PhysRevE.100.062134} {\bibfield  {journal} {\bibinfo
   {journal} {Phys. Rev. E}\ }\textbf {\bibinfo {volume} {100}},\ \bibinfo
  {pages} {062134} (\bibinfo {year} {2019})}\BibitemShut {NoStop}%
\bibitem [{\citenamefont {Kaneko}\ \emph {et~al.}(2020)\citenamefont {Kaneko},
  \citenamefont {Iyoda},\ and\ \citenamefont {Sagawa}}]{kaneko_iyoda_20}%
  \BibitemOpen
  \bibfield  {author} {\bibinfo {author} {\bibfnamefont {K.}~\bibnamefont
  {Kaneko}}, \bibinfo {author} {\bibfnamefont {E.}~\bibnamefont {Iyoda}},\ and\
  \bibinfo {author} {\bibfnamefont {T.}~\bibnamefont {Sagawa}},\ }\bibfield
  {title} {\bibinfo {title} {Characterizing complexity of many-body quantum
  dynamics by higher-order eigenstate thermalization},\ }\href
  {https://doi.org/10.1103/PhysRevA.101.042126} {\bibfield  {journal} {\bibinfo
   {journal} {Phys. Rev. A}\ }\textbf {\bibinfo {volume} {101}},\ \bibinfo
  {pages} {042126} (\bibinfo {year} {2020})}\BibitemShut {NoStop}%
\bibitem [{\citenamefont {Huang}(2021)}]{huang_21}%
  \BibitemOpen
  \bibfield  {author} {\bibinfo {author} {\bibfnamefont {Y.}~\bibnamefont
  {Huang}},\ }\bibfield  {title} {\bibinfo {title} {Universal entanglement of
  mid-spectrum eigenstates of chaotic local {H}amiltonians},\ }\href
  {https://doi.org/10.1016/j.nuclphysb.2021.115373} {\bibfield  {journal}
  {\bibinfo  {journal} {Nuc. Phys. B}\ }\textbf {\bibinfo {volume} {966}},\
  \bibinfo {pages} {115373} (\bibinfo {year} {2021})}\BibitemShut {NoStop}%
\bibitem [{\citenamefont {Haque}\ \emph {et~al.}(2022)\citenamefont {Haque},
  \citenamefont {McClarty},\ and\ \citenamefont
  {Khaymovich}}]{Haque_Khaymovich_2022}%
  \BibitemOpen
  \bibfield  {author} {\bibinfo {author} {\bibfnamefont {M.}~\bibnamefont
  {Haque}}, \bibinfo {author} {\bibfnamefont {P.~A.}\ \bibnamefont
  {McClarty}},\ and\ \bibinfo {author} {\bibfnamefont {I.~M.}\ \bibnamefont
  {Khaymovich}},\ }\bibfield  {title} {\bibinfo {title} {Entanglement of
  midspectrum eigenstates of chaotic many-body systems: {R}easons for deviation
  from random ensembles},\ }\href {https://doi.org/10.1103/PhysRevE.105.014109}
  {\bibfield  {journal} {\bibinfo  {journal} {Phys. Rev. E}\ }\textbf {\bibinfo
  {volume} {105}},\ \bibinfo {pages} {014109} (\bibinfo {year}
  {2022})}\BibitemShut {NoStop}%
\bibitem [{\citenamefont {Kliczkowski}\ \emph {et~al.}(2023)\citenamefont
  {Kliczkowski}, \citenamefont {\ifmmode \acute{S}\else
  \'{S}\fi{}wi\ifmmode~\mbox{\k{e}}\else \k{e}\fi{}tek}, \citenamefont
  {Vidmar},\ and\ \citenamefont {Rigol}}]{kliczkowski2023average}%
  \BibitemOpen
  \bibfield  {author} {\bibinfo {author} {\bibfnamefont {M.}~\bibnamefont
  {Kliczkowski}}, \bibinfo {author} {\bibfnamefont {R.}~\bibnamefont {\ifmmode
  \acute{S}\else \'{S}\fi{}wi\ifmmode~\mbox{\k{e}}\else \k{e}\fi{}tek}},
  \bibinfo {author} {\bibfnamefont {L.}~\bibnamefont {Vidmar}},\ and\ \bibinfo
  {author} {\bibfnamefont {M.}~\bibnamefont {Rigol}},\ }\bibfield  {title}
  {\bibinfo {title} {Average entanglement entropy of midspectrum eigenstates of
  quantum-chaotic interacting {H}amiltonians},\ }\href
  {https://doi.org/10.1103/PhysRevE.107.064119} {\bibfield  {journal} {\bibinfo
   {journal} {Phys. Rev. E}\ }\textbf {\bibinfo {volume} {107}},\ \bibinfo
  {pages} {064119} (\bibinfo {year} {2023})}\BibitemShut {NoStop}%
\bibitem [{\citenamefont {Rodriguez-Nieva}\ \emph {et~al.}()\citenamefont
  {Rodriguez-Nieva}, \citenamefont {Jonay},\ and\ \citenamefont
  {Khemani}}]{nieva2023}%
  \BibitemOpen
  \bibfield  {author} {\bibinfo {author} {\bibfnamefont {J.~F.}\ \bibnamefont
  {Rodriguez-Nieva}}, \bibinfo {author} {\bibfnamefont {C.}~\bibnamefont
  {Jonay}},\ and\ \bibinfo {author} {\bibfnamefont {V.}~\bibnamefont
  {Khemani}},\ }\href {https://arxiv.org/abs/2305.11940} {\bibinfo {title}
  {Quantifying quantum chaos through microcanonical distributions of
  entanglement}},\ \bibinfo {howpublished} {arXiv:2305.11940}\BibitemShut
  {NoStop}%
\bibitem [{\citenamefont {Patil}\ \emph {et~al.}()\citenamefont {Patil},
  \citenamefont {Hackl}, \citenamefont {Fagan},\ and\ \citenamefont
  {Rigol}}]{patil2023}%
  \BibitemOpen
  \bibfield  {author} {\bibinfo {author} {\bibfnamefont {R.}~\bibnamefont
  {Patil}}, \bibinfo {author} {\bibfnamefont {L.}~\bibnamefont {Hackl}},
  \bibinfo {author} {\bibfnamefont {G.~R.}\ \bibnamefont {Fagan}},\ and\
  \bibinfo {author} {\bibfnamefont {M.}~\bibnamefont {Rigol}},\ }\href
  {https://arxiv.org/abs/2305.11211} {\bibinfo {title} {Average pure-state
  entanglement entropy in spin systems with {SU(2)} symmetry}},\ \bibinfo
  {howpublished} {arXiv:2305.11211}\BibitemShut {NoStop}%
\bibitem [{\citenamefont {Alba}\ \emph {et~al.}(2009)\citenamefont {Alba},
  \citenamefont {Fagotti},\ and\ \citenamefont {Calabrese}}]{alba09}%
  \BibitemOpen
  \bibfield  {author} {\bibinfo {author} {\bibfnamefont {V.}~\bibnamefont
  {Alba}}, \bibinfo {author} {\bibfnamefont {M.}~\bibnamefont {Fagotti}},\ and\
  \bibinfo {author} {\bibfnamefont {P.}~\bibnamefont {Calabrese}},\ }\bibfield
  {title} {\bibinfo {title} {Entanglement entropy of excited states},\ }\href
  {https://doi.org/10.1088/1742-5468/2009/10/P10020} {\bibfield  {journal}
  {\bibinfo  {journal} {J. Stat. Mech.}\ }\textbf {\bibinfo {volume} {{\rm
  (2009)}}},\ \bibinfo {pages} {P10020} (\bibinfo {year} {2009})}\BibitemShut
  {NoStop}%
\bibitem [{\citenamefont {M\"{o}lter}\ \emph {et~al.}(2014)\citenamefont
  {M\"{o}lter}, \citenamefont {Barthel}, \citenamefont {Schollw\"{o}ck},\ and\
  \citenamefont {Alba}}]{moelter_barthel_14}%
  \BibitemOpen
  \bibfield  {author} {\bibinfo {author} {\bibfnamefont {J.}~\bibnamefont
  {M\"{o}lter}}, \bibinfo {author} {\bibfnamefont {T.}~\bibnamefont {Barthel}},
  \bibinfo {author} {\bibfnamefont {U.}~\bibnamefont {Schollw\"{o}ck}},\ and\
  \bibinfo {author} {\bibfnamefont {V.}~\bibnamefont {Alba}},\ }\bibfield
  {title} {\bibinfo {title} {Bound states and entanglement in the excited
  states of quantum spin chains},\ }\href
  {https://doi.org/10.1088/1742-5468/2014/10/P10029} {\bibfield  {journal}
  {\bibinfo  {journal} {J. Stat. Mech.}\ }\textbf {\bibinfo {volume}
  {{\rm(2014)}}},\ \bibinfo {pages} {P10029} (\bibinfo {year}
  {2014})}\BibitemShut {NoStop}%
\bibitem [{\citenamefont {Storms}\ and\ \citenamefont
  {Singh}(2014)}]{storms_singh_14}%
  \BibitemOpen
  \bibfield  {author} {\bibinfo {author} {\bibfnamefont {M.}~\bibnamefont
  {Storms}}\ and\ \bibinfo {author} {\bibfnamefont {R.~R.~P.}\ \bibnamefont
  {Singh}},\ }\bibfield  {title} {\bibinfo {title} {Entanglement in ground and
  excited states of gapped free-fermion systems and their relationship with
  fermi surface and thermodynamic equilibrium properties},\ }\href
  {https://doi.org/10.1103/PhysRevE.89.012125} {\bibfield  {journal} {\bibinfo
  {journal} {Phys. Rev. E}\ }\textbf {\bibinfo {volume} {89}},\ \bibinfo
  {pages} {012125} (\bibinfo {year} {2014})}\BibitemShut {NoStop}%
\bibitem [{\citenamefont {Lai}\ and\ \citenamefont {Yang}(2015)}]{lai_yang_15}%
  \BibitemOpen
  \bibfield  {author} {\bibinfo {author} {\bibfnamefont {H.-H.}\ \bibnamefont
  {Lai}}\ and\ \bibinfo {author} {\bibfnamefont {K.}~\bibnamefont {Yang}},\
  }\bibfield  {title} {\bibinfo {title} {Entanglement entropy scaling laws and
  eigenstate typicality in free fermion systems},\ }\href
  {https://doi.org/10.1103/PhysRevB.91.081110} {\bibfield  {journal} {\bibinfo
  {journal} {Phys. Rev. B}\ }\textbf {\bibinfo {volume} {91}},\ \bibinfo
  {pages} {081110} (\bibinfo {year} {2015})}\BibitemShut {NoStop}%
\bibitem [{\citenamefont {Nandy}\ \emph {et~al.}(2016)\citenamefont {Nandy},
  \citenamefont {Sen}, \citenamefont {Das},\ and\ \citenamefont
  {Dhar}}]{nandy_sen_16}%
  \BibitemOpen
  \bibfield  {author} {\bibinfo {author} {\bibfnamefont {S.}~\bibnamefont
  {Nandy}}, \bibinfo {author} {\bibfnamefont {A.}~\bibnamefont {Sen}}, \bibinfo
  {author} {\bibfnamefont {A.}~\bibnamefont {Das}},\ and\ \bibinfo {author}
  {\bibfnamefont {A.}~\bibnamefont {Dhar}},\ }\bibfield  {title} {\bibinfo
  {title} {Eigenstate gibbs ensemble in integrable quantum systems},\ }\href
  {https://doi.org/10.1103/PhysRevB.94.245131} {\bibfield  {journal} {\bibinfo
  {journal} {Phys. Rev. B}\ }\textbf {\bibinfo {volume} {94}},\ \bibinfo
  {pages} {245131} (\bibinfo {year} {2016})}\BibitemShut {NoStop}%
\bibitem [{\citenamefont {Vidmar}\ \emph {et~al.}(2017)\citenamefont {Vidmar},
  \citenamefont {Hackl}, \citenamefont {Bianchi},\ and\ \citenamefont
  {Rigol}}]{VidmarHackl_2017}%
  \BibitemOpen
  \bibfield  {author} {\bibinfo {author} {\bibfnamefont {L.}~\bibnamefont
  {Vidmar}}, \bibinfo {author} {\bibfnamefont {L.}~\bibnamefont {Hackl}},
  \bibinfo {author} {\bibfnamefont {E.}~\bibnamefont {Bianchi}},\ and\ \bibinfo
  {author} {\bibfnamefont {M.}~\bibnamefont {Rigol}},\ }\bibfield  {title}
  {\bibinfo {title} {Entanglement entropy of eigenstates of quadratic fermionic
  {H}amiltonians},\ }\href {https://doi.org/10.1103/PhysRevLett.119.020601}
  {\bibfield  {journal} {\bibinfo  {journal} {Phys. Rev. Lett.}\ }\textbf
  {\bibinfo {volume} {119}},\ \bibinfo {pages} {020601} (\bibinfo {year}
  {2017})}\BibitemShut {NoStop}%
\bibitem [{\citenamefont {Vidmar}\ \emph {et~al.}(2018)\citenamefont {Vidmar},
  \citenamefont {Hackl}, \citenamefont {Bianchi},\ and\ \citenamefont
  {Rigol}}]{vidmar2018volume}%
  \BibitemOpen
  \bibfield  {author} {\bibinfo {author} {\bibfnamefont {L.}~\bibnamefont
  {Vidmar}}, \bibinfo {author} {\bibfnamefont {L.}~\bibnamefont {Hackl}},
  \bibinfo {author} {\bibfnamefont {E.}~\bibnamefont {Bianchi}},\ and\ \bibinfo
  {author} {\bibfnamefont {M.}~\bibnamefont {Rigol}},\ }\bibfield  {title}
  {\bibinfo {title} {Volume law and quantum criticality in the entanglement
  entropy of excited eigenstates of the quantum {Ising} model},\ }\href
  {https://doi.org/10.1103/PhysRevLett.121.220602} {\bibfield  {journal}
  {\bibinfo  {journal} {Phys. Rev. Lett.}\ }\textbf {\bibinfo {volume} {121}},\
  \bibinfo {pages} {220602} (\bibinfo {year} {2018})}\BibitemShut {NoStop}%
\bibitem [{\citenamefont {Zhang}\ \emph {et~al.}(2018)\citenamefont {Zhang},
  \citenamefont {Vidmar},\ and\ \citenamefont {Rigol}}]{zhang_vidmar_18}%
  \BibitemOpen
  \bibfield  {author} {\bibinfo {author} {\bibfnamefont {Y.}~\bibnamefont
  {Zhang}}, \bibinfo {author} {\bibfnamefont {L.}~\bibnamefont {Vidmar}},\ and\
  \bibinfo {author} {\bibfnamefont {M.}~\bibnamefont {Rigol}},\ }\bibfield
  {title} {\bibinfo {title} {Information measures for a local quantum phase
  transition: {L}attice fermions in a one-dimensional harmonic trap},\ }\href
  {https://doi.org/10.1103/PhysRevA.97.023605} {\bibfield  {journal} {\bibinfo
  {journal} {Phys. Rev. A}\ }\textbf {\bibinfo {volume} {97}},\ \bibinfo
  {pages} {023605} (\bibinfo {year} {2018})}\BibitemShut {NoStop}%
\bibitem [{\citenamefont {Hackl}\ \emph {et~al.}(2019)\citenamefont {Hackl},
  \citenamefont {Vidmar}, \citenamefont {Rigol},\ and\ \citenamefont
  {Bianchi}}]{HacklVidmar_2019}%
  \BibitemOpen
  \bibfield  {author} {\bibinfo {author} {\bibfnamefont {L.}~\bibnamefont
  {Hackl}}, \bibinfo {author} {\bibfnamefont {L.}~\bibnamefont {Vidmar}},
  \bibinfo {author} {\bibfnamefont {M.}~\bibnamefont {Rigol}},\ and\ \bibinfo
  {author} {\bibfnamefont {E.}~\bibnamefont {Bianchi}},\ }\bibfield  {title}
  {\bibinfo {title} {Average eigenstate entanglement entropy of the {XY} chain
  in a transverse field and its universality for translationally invariant
  quadratic fermionic models},\ }\href
  {https://doi.org/10.1103/PhysRevB.99.075123} {\bibfield  {journal} {\bibinfo
  {journal} {Phys. Rev. B}\ }\textbf {\bibinfo {volume} {99}},\ \bibinfo
  {pages} {075123} (\bibinfo {year} {2019})}\BibitemShut {NoStop}%
\bibitem [{\citenamefont {Jafarizadeh}\ and\ \citenamefont
  {Rajabpour}(2019)}]{jafarizadeh_rajabpour_19}%
  \BibitemOpen
  \bibfield  {author} {\bibinfo {author} {\bibfnamefont {A.}~\bibnamefont
  {Jafarizadeh}}\ and\ \bibinfo {author} {\bibfnamefont {M.~A.}\ \bibnamefont
  {Rajabpour}},\ }\bibfield  {title} {\bibinfo {title} {Bipartite entanglement
  entropy of the excited states of free fermions and harmonic oscillators},\
  }\href {https://doi.org/10.1103/PhysRevB.100.165135} {\bibfield  {journal}
  {\bibinfo  {journal} {Phys. Rev. B}\ }\textbf {\bibinfo {volume} {100}},\
  \bibinfo {pages} {165135} (\bibinfo {year} {2019})}\BibitemShut {NoStop}%
\bibitem [{\citenamefont {\L{}yd\ifmmode~\dot{z}\else \.{z}\fi{}ba}\ \emph
  {et~al.}(2020)\citenamefont {\L{}yd\ifmmode~\dot{z}\else \.{z}\fi{}ba},
  \citenamefont {Rigol},\ and\ \citenamefont {Vidmar}}]{lydzba2020eigenstate}%
  \BibitemOpen
  \bibfield  {author} {\bibinfo {author} {\bibfnamefont {P.}~\bibnamefont
  {\L{}yd\ifmmode~\dot{z}\else \.{z}\fi{}ba}}, \bibinfo {author} {\bibfnamefont
  {M.}~\bibnamefont {Rigol}},\ and\ \bibinfo {author} {\bibfnamefont
  {L.}~\bibnamefont {Vidmar}},\ }\bibfield  {title} {\bibinfo {title}
  {Eigenstate entanglement entropy in random quadratic {Hamiltonians}},\ }\href
  {https://doi.org/10.1103/PhysRevLett.125.180604} {\bibfield  {journal}
  {\bibinfo  {journal} {Phys. Rev. Lett.}\ }\textbf {\bibinfo {volume} {125}},\
  \bibinfo {pages} {180604} (\bibinfo {year} {2020})}\BibitemShut {NoStop}%
\bibitem [{\citenamefont {\L{}yd\ifmmode~\dot{z}\else \.{z}\fi{}ba}\ \emph
  {et~al.}(2021)\citenamefont {\L{}yd\ifmmode~\dot{z}\else \.{z}\fi{}ba},
  \citenamefont {Rigol},\ and\ \citenamefont
  {Vidmar}}]{lydzba2021entanglement}%
  \BibitemOpen
  \bibfield  {author} {\bibinfo {author} {\bibfnamefont {P.}~\bibnamefont
  {\L{}yd\ifmmode~\dot{z}\else \.{z}\fi{}ba}}, \bibinfo {author} {\bibfnamefont
  {M.}~\bibnamefont {Rigol}},\ and\ \bibinfo {author} {\bibfnamefont
  {L.}~\bibnamefont {Vidmar}},\ }\bibfield  {title} {\bibinfo {title}
  {Entanglement in many-body eigenstates of quantum-chaotic quadratic
  {Hamiltonians}},\ }\href {https://doi.org/10.1103/PhysRevB.103.104206}
  {\bibfield  {journal} {\bibinfo  {journal} {Phys. Rev. B}\ }\textbf {\bibinfo
  {volume} {103}},\ \bibinfo {pages} {104206} (\bibinfo {year}
  {2021})}\BibitemShut {NoStop}%
\bibitem [{\citenamefont {Frey}\ \emph {et~al.}()\citenamefont {Frey},
  \citenamefont {Mikhail}, \citenamefont {Rachel},\ and\ \citenamefont
  {Hackl}}]{frey2023probing}%
  \BibitemOpen
  \bibfield  {author} {\bibinfo {author} {\bibfnamefont {P.}~\bibnamefont
  {Frey}}, \bibinfo {author} {\bibfnamefont {D.}~\bibnamefont {Mikhail}},
  \bibinfo {author} {\bibfnamefont {S.}~\bibnamefont {Rachel}},\ and\ \bibinfo
  {author} {\bibfnamefont {L.}~\bibnamefont {Hackl}},\ }\href
  {https://arxiv.org/abs/2309.03632} {\bibinfo {title} {Probing {Hilbert} space
  fragmentation and the block inverse participation ratio}},\ \bibinfo
  {howpublished} {arXiv:2309.03632}\BibitemShut {NoStop}%
\bibitem [{\citenamefont {Eisert}\ \emph {et~al.}(2010)\citenamefont {Eisert},
  \citenamefont {Cramer},\ and\ \citenamefont
  {Plenio}}]{eisert_colloquium_2010}%
  \BibitemOpen
  \bibfield  {author} {\bibinfo {author} {\bibfnamefont {J.}~\bibnamefont
  {Eisert}}, \bibinfo {author} {\bibfnamefont {M.}~\bibnamefont {Cramer}},\
  and\ \bibinfo {author} {\bibfnamefont {M.~B.}\ \bibnamefont {Plenio}},\
  }\bibfield  {title} {\bibinfo {title} {Colloquium: {A}rea laws for the
  entanglement entropy},\ }\href {https://doi.org/10.1103/RevModPhys.82.277}
  {\bibfield  {journal} {\bibinfo  {journal} {Rev. Mod. Phys.}\ }\textbf
  {\bibinfo {volume} {82}},\ \bibinfo {pages} {277} (\bibinfo {year}
  {2010})}\BibitemShut {NoStop}%
\bibitem [{\citenamefont {Page}(1993{\natexlab{b}})}]{Page_1993}%
  \BibitemOpen
  \bibfield  {author} {\bibinfo {author} {\bibfnamefont {D.~N.}\ \bibnamefont
  {Page}},\ }\bibfield  {title} {\bibinfo {title} {Average entropy of a
  subsystem},\ }\href {https://doi.org/10.1103/physrevlett.71.1291} {\bibfield
  {journal} {\bibinfo  {journal} {Phys. Rev. Lett.}\ }\textbf {\bibinfo
  {volume} {71}},\ \bibinfo {pages} {1291} (\bibinfo {year}
  {1993}{\natexlab{b}})}\BibitemShut {NoStop}%
\bibitem [{\citenamefont {Bianchi}\ \emph {et~al.}(2022)\citenamefont
  {Bianchi}, \citenamefont {Hackl}, \citenamefont {Kieburg}, \citenamefont
  {Rigol},\ and\ \citenamefont {Vidmar}}]{Bianchi_2022}%
  \BibitemOpen
  \bibfield  {author} {\bibinfo {author} {\bibfnamefont {E.}~\bibnamefont
  {Bianchi}}, \bibinfo {author} {\bibfnamefont {L.}~\bibnamefont {Hackl}},
  \bibinfo {author} {\bibfnamefont {M.}~\bibnamefont {Kieburg}}, \bibinfo
  {author} {\bibfnamefont {M.}~\bibnamefont {Rigol}},\ and\ \bibinfo {author}
  {\bibfnamefont {L.}~\bibnamefont {Vidmar}},\ }\bibfield  {title} {\bibinfo
  {title} {Volume-law entanglement entropy of typical pure quantum states},\
  }\href {https://doi.org/10.1103/PRXQuantum.3.030201} {\bibfield  {journal}
  {\bibinfo  {journal} {PRX Quantum}\ }\textbf {\bibinfo {volume} {3}},\
  \bibinfo {pages} {030201} (\bibinfo {year} {2022})}\BibitemShut {NoStop}%
\bibitem [{\citenamefont {Bianchi}\ \emph {et~al.}(2021)\citenamefont
  {Bianchi}, \citenamefont {Hackl},\ and\ \citenamefont
  {Kieburg}}]{bianchi2021page}%
  \BibitemOpen
  \bibfield  {author} {\bibinfo {author} {\bibfnamefont {E.}~\bibnamefont
  {Bianchi}}, \bibinfo {author} {\bibfnamefont {L.}~\bibnamefont {Hackl}},\
  and\ \bibinfo {author} {\bibfnamefont {M.}~\bibnamefont {Kieburg}},\
  }\bibfield  {title} {\bibinfo {title} {Page curve for fermionic {Gaussian}
  states},\ }\href {https://doi.org/10.1103/PhysRevB.103.L241118} {\bibfield
  {journal} {\bibinfo  {journal} {Phys. Rev. B}\ }\textbf {\bibinfo {volume}
  {103}},\ \bibinfo {pages} {L241118} (\bibinfo {year} {2021})}\BibitemShut
  {NoStop}%
\bibitem [{\citenamefont {Bianchi}\ and\ \citenamefont
  {Don\`a}(2019)}]{Bianchi_2019}%
  \BibitemOpen
  \bibfield  {author} {\bibinfo {author} {\bibfnamefont {E.}~\bibnamefont
  {Bianchi}}\ and\ \bibinfo {author} {\bibfnamefont {P.}~\bibnamefont
  {Don\`a}},\ }\bibfield  {title} {\bibinfo {title} {Typical entanglement
  entropy in the presence of a center: {P}age curve and its variance},\ }\href
  {https://doi.org/10.1103/PhysRevD.100.105010} {\bibfield  {journal} {\bibinfo
   {journal} {Phys. Rev. D}\ }\textbf {\bibinfo {volume} {100}},\ \bibinfo
  {pages} {105010} (\bibinfo {year} {2019})}\BibitemShut {NoStop}%
\bibitem [{sup(2023)}]{supp}%
  \BibitemOpen
  \href@noop {} {\emph {\bibinfo {title} {{See Supplemental Material for
  further details regarding the derivations of the average and standard
  deviation of the entanglement entropy}}}} (\bibinfo {year}
  {2023})\BibitemShut {NoStop}%
\bibitem [{\citenamefont {Euler}(1801)}]{euler1801evolutione}%
  \BibitemOpen
  \bibfield  {author} {\bibinfo {author} {\bibfnamefont {L.}~\bibnamefont
  {Euler}},\ }\bibfield  {title} {\bibinfo {title} {De evolutione potestatis
  polynomialis cuiuscunque $(1+ x+ x^2+ x^3+ x^4+ etc.)^n$},\ }\href@noop {}
  {\bibfield  {journal} {\bibinfo  {journal} {Nova Acta Academiae Scientiarum
  Imperialis Petropolitanae}\ ,\ \bibinfo {pages} {47}} (\bibinfo {year}
  {1801})}\BibitemShut {NoStop}%
\bibitem [{\citenamefont {Comtet}(1974)}]{comtet1974advanced}%
  \BibitemOpen
  \bibfield  {author} {\bibinfo {author} {\bibfnamefont {L.}~\bibnamefont
  {Comtet}},\ }\href@noop {} {\emph {\bibinfo {title} {Advanced Combinatorics:
  The art of finite and infinite expansions}}}\ (\bibinfo  {publisher}
  {Springer Science \& Business Media},\ \bibinfo {year} {1974})\BibitemShut
  {NoStop}%
\bibitem [{\citenamefont {Andrews}(1975)}]{andrews1975theorem}%
  \BibitemOpen
  \bibfield  {author} {\bibinfo {author} {\bibfnamefont {G.~E.}\ \bibnamefont
  {Andrews}},\ }\bibfield  {title} {\bibinfo {title} {A theorem on reciprocal
  polynomials with applications to permutations and compositions},\ }\href
  {https://doi.org/10.2307/2319803} {\bibfield  {journal} {\bibinfo  {journal}
  {Am. Math. Mon.}\ }\textbf {\bibinfo {volume} {82}},\ \bibinfo {pages} {830}
  (\bibinfo {year} {1975})}\BibitemShut {NoStop}%
\bibitem [{\citenamefont {Neuschel}(2014)}]{neuschel2014note}%
  \BibitemOpen
  \bibfield  {author} {\bibinfo {author} {\bibfnamefont {T.}~\bibnamefont
  {Neuschel}},\ }\bibfield  {title} {\bibinfo {title} {A note on extended
  binomial coefficients},\ }\href
  {https://cs.uwaterloo.ca/journals/JIS/VOL17/Neuschel/neuschel4.html}
  {\bibfield  {journal} {\bibinfo  {journal} {J. Integer Seq.}\ }\textbf
  {\bibinfo {volume} {17}},\ \bibinfo {pages} {14.10.4} (\bibinfo {year}
  {2014})}\BibitemShut {NoStop}%
\bibitem [{\citenamefont {Li}()}]{li2014asymptotic}%
  \BibitemOpen
  \bibfield  {author} {\bibinfo {author} {\bibfnamefont {J.}~\bibnamefont
  {Li}},\ }\href {https://arxiv.org/abs/1405.1803} {\bibinfo {title}
  {Asymptotic estimate for the polynomial coefficients}},\ \bibinfo
  {howpublished} {arXiv:1405.1803}\BibitemShut {NoStop}%
\bibitem [{\citenamefont {Neuschel}()}]{neuschel-personal}%
  \BibitemOpen
  \bibfield  {author} {\bibinfo {author} {\bibfnamefont {T.}~\bibnamefont
  {Neuschel}},\ }\href@noop {} {\bibinfo {title} {Note on extended binomial
  coefficients, private communication}}\BibitemShut {NoStop}%
\bibitem [{\citenamefont {Zamolodchikov}\ and\ \citenamefont
  {Fateev}(1980)}]{Zamolodchikov_1980}%
  \BibitemOpen
  \bibfield  {author} {\bibinfo {author} {\bibfnamefont {A.~B.}\ \bibnamefont
  {Zamolodchikov}}\ and\ \bibinfo {author} {\bibfnamefont {V.~A.}\ \bibnamefont
  {Fateev}},\ }\bibfield  {title} {\bibinfo {title} {Model factorized
  {S}-matrix and an integrable spin-1 {H}eisenberg chain},\ }\href
  {https://www.osti.gov/biblio/6190897} {\bibfield  {journal} {\bibinfo
  {journal} {Sov. J. Nucl. Phys.(Engl. Transl.);(United States)}\ }\textbf
  {\bibinfo {volume} {32}} (\bibinfo {year} {1980})}\BibitemShut {NoStop}%
\bibitem [{\citenamefont {Bytsko}(2003)}]{Bytsko_2003}%
  \BibitemOpen
  \bibfield  {author} {\bibinfo {author} {\bibfnamefont {A.~G.}\ \bibnamefont
  {Bytsko}},\ }\bibfield  {title} {\bibinfo {title} {{On integrable
  Hamiltonians for higher spin {XXZ} chain}},\ }\href
  {https://doi.org/10.1063/1.1591054} {\bibfield  {journal} {\bibinfo
  {journal} {J. Math. Phys.}\ }\textbf {\bibinfo {volume} {44}},\ \bibinfo
  {pages} {3698} (\bibinfo {year} {2003})}\BibitemShut {NoStop}%
\bibitem [{\citenamefont {Cazalilla}\ \emph {et~al.}(2011)\citenamefont
  {Cazalilla}, \citenamefont {Citro}, \citenamefont {Giamarchi}, \citenamefont
  {Orignac},\ and\ \citenamefont {Rigol}}]{Cazalilla_2011}%
  \BibitemOpen
  \bibfield  {author} {\bibinfo {author} {\bibfnamefont {M.~A.}\ \bibnamefont
  {Cazalilla}}, \bibinfo {author} {\bibfnamefont {R.}~\bibnamefont {Citro}},
  \bibinfo {author} {\bibfnamefont {T.}~\bibnamefont {Giamarchi}}, \bibinfo
  {author} {\bibfnamefont {E.}~\bibnamefont {Orignac}},\ and\ \bibinfo {author}
  {\bibfnamefont {M.}~\bibnamefont {Rigol}},\ }\bibfield  {title} {\bibinfo
  {title} {One dimensional bosons: {F}rom condensed matter systems to ultracold
  gases},\ }\href {https://doi.org/10.1103/RevModPhys.83.1405} {\bibfield
  {journal} {\bibinfo  {journal} {Rev. Mod. Phys.}\ }\textbf {\bibinfo {volume}
  {83}},\ \bibinfo {pages} {1405} (\bibinfo {year} {2011})}\BibitemShut
  {NoStop}%
\bibitem [{\citenamefont {Kollath}\ \emph {et~al.}(2010)\citenamefont
  {Kollath}, \citenamefont {Roux}, \citenamefont {Biroli},\ and\ \citenamefont
  {L\"auchli}}]{Kollath_2010}%
  \BibitemOpen
  \bibfield  {author} {\bibinfo {author} {\bibfnamefont {C.}~\bibnamefont
  {Kollath}}, \bibinfo {author} {\bibfnamefont {G.}~\bibnamefont {Roux}},
  \bibinfo {author} {\bibfnamefont {G.}~\bibnamefont {Biroli}},\ and\ \bibinfo
  {author} {\bibfnamefont {A.~M.}\ \bibnamefont {L\"auchli}},\ }\bibfield
  {title} {\bibinfo {title} {Statistical properties of the spectrum of the
  extended {Bose–Hubbard} model},\ }\href
  {https://doi.org/10.1088/1742-5468/2010/08/P08011} {\bibfield  {journal}
  {\bibinfo  {journal} {J. Stat. Mech.}\ }\textbf {\bibinfo {volume} {2010}},\
  \bibinfo {pages} {P08011} (\bibinfo {year} {2010})}\BibitemShut {NoStop}%
\bibitem [{\citenamefont {Kitaev}\ and\ \citenamefont
  {Preskill}(2006)}]{kitaev2006topological}%
  \BibitemOpen
  \bibfield  {author} {\bibinfo {author} {\bibfnamefont {A.}~\bibnamefont
  {Kitaev}}\ and\ \bibinfo {author} {\bibfnamefont {J.}~\bibnamefont
  {Preskill}},\ }\bibfield  {title} {\bibinfo {title} {Topological entanglement
  entropy},\ }\href {https://doi.org/10.1103/PhysRevLett.96.110404} {\bibfield
  {journal} {\bibinfo  {journal} {Phys. Rev. Lett.}\ }\textbf {\bibinfo
  {volume} {96}},\ \bibinfo {pages} {110404} (\bibinfo {year}
  {2006})}\BibitemShut {NoStop}%
\bibitem [{\citenamefont {Murciano}\ \emph {et~al.}(2022)\citenamefont
  {Murciano}, \citenamefont {Calabrese},\ and\ \citenamefont
  {Piroli}}]{murciano2022symmetry}%
  \BibitemOpen
  \bibfield  {author} {\bibinfo {author} {\bibfnamefont {S.}~\bibnamefont
  {Murciano}}, \bibinfo {author} {\bibfnamefont {P.}~\bibnamefont
  {Calabrese}},\ and\ \bibinfo {author} {\bibfnamefont {L.}~\bibnamefont
  {Piroli}},\ }\bibfield  {title} {\bibinfo {title} {Symmetry-resolved {Page}
  curves},\ }\href {https://doi.org/10.1103/PhysRevD.106.046015} {\bibfield
  {journal} {\bibinfo  {journal} {Phys. Rev. D}\ }\textbf {\bibinfo {volume}
  {106}},\ \bibinfo {pages} {046015} (\bibinfo {year} {2022})}\BibitemShut
  {NoStop}%
\end{thebibliography}%

\end{document}


\title{Supplementary material:\texorpdfstring{\\}{}
Typical entanglement entropy in systems with particle-number conservation}

\author{Yale Yauk}
\affiliation{School of Mathematics and Statistics, The University of Melbourne, Parkville, VIC 3010, Australia}
\affiliation{School of Physics, The University of Melbourne, Parkville, VIC 3010, Australia}
\affiliation{Perimeter Institute for Theoretical Physics, 31 Caroline St. N, Waterloo, ON N2L 2Y5, Canada}

\author{Rohit Patil}
\affiliation{Department of Physics, The Pennsylvania State University, University Park, PA 16802, USA}

\author{Yicheng Zhang}
\affiliation{Homer L. Dodge Department of Physics and Astronomy,\\
The University of Oklahoma, Norman, OK 73019, USA}
\affiliation{Center for Quantum Research and Technology, The University of Oklahoma, Norman, OK 73019, USA}

\author{Marcos Rigol}
\affiliation{Department of Physics, The Pennsylvania State University, University Park, PA 16802, USA}

\author{Lucas Hackl}
\affiliation{School of Mathematics and Statistics, The University of Melbourne, Parkville, VIC 3010, Australia}
\affiliation{School of Physics, The University of Melbourne, Parkville, VIC 3010, Australia}

\maketitle

\section{Details: Derivation of average $\braket{S_A}_N$}
\label{app:derivative_of_average_entropy}
In order to evaluate the formula Eq.~(19) from the main text, we approximate the sum with an integral, which can then be evaluated explicitly up to constant order.

\subsection{Probability weight $\varrho(n_A)$}
We define
\begin{align}
    \varrho(n_A)=V\frac{d_Ad_B}{d_N}=\sqrt{\frac{V}{2\pi f(1-f)} \abs{\frac{\beta''\left(\frac{n_A}{f}\right)\beta''\left(\frac{n-n_A}{1-f}\right)}{\beta''(n)}}}\exp\left[V\left(f\beta\left(\frac{n_A}{f}\right)+(1-f)\beta\left(\frac{n-n_A}{1-f}\right)-\beta(n)\right)\right]\,.
\end{align}
We find a saddle-point approximation to $\varrho(n_A)$, which requires us to solve for the stationary point of the exponent. This leads to the equation
\begin{equation}
    \beta'\left(\frac{n_A}{f}\right)-\beta'\left(\frac{n-n_A}{1-f}\right)=0\,,
\end{equation}
which by the concavity of $\beta$ has a unique saddle point solution $n_A=fn$, leading to the approximation
\begin{equation}
    \varrho(n_A)=\sqrt{\frac{V\abs{\beta''(n)}}{2\pi f(1-f)}}e^{-\frac{V}{2}\frac{\abs{\beta''(n)}(n_A-fn)^2}{f(1-f)}}\,.
\end{equation}
The central moments of $\varrho(n_A)$ will be key in simplifying integral expressions. They are
\begin{align}
\label{eq:varrho_moments}
\begin{split}
    \mathcal{M}_0&=\int_{-\infty}^\infty \varrho(n_A)\dd{n_A}=1\,, \\
    \mathcal{M}_1&=\int_{-\infty}^\infty \varrho(n_A)(n_A-fn)\dd{n_A}=0\,, \\
    \mathcal{M}_2&=\int_{-\infty}^\infty \varrho(n_A)(n_A-fn)^2\dd{n_A}=\frac{f(1-f)}{\abs{\beta''(n)}V}\,.
\end{split}
\end{align}
We shall also need the half-moments, or one-sided moments:
\begin{align}
\label{eq:varrho_half_moments}
    \begin{split}
        \mathcal{M}_0^-&=\int_{-\infty}^{fn} \varrho(n_A)\dd{n_A}=\frac{1}{2}\,, \\
        \mathcal{M}_1^-&=\int_{-\infty}^{fn} \varrho(n_A)(n_A-fn)\dd{n_A}=-\frac{1}{\sqrt{8\pi \abs{\beta''(n)}V}}\,, \\
        \mathcal{M}_2^-&=\int_{-\infty}^{fn} \varrho(n_A)(n_A-fn)^2\dd{n_A}=\frac{1}{8\abs{\beta''(n)}V}\,.
    \end{split}
    \quad 
    \begin{split}
        \mathcal{M}_0^+&=\int_{fn}^{\infty} \varrho(n_A)\dd{n_A}=\frac{1}{2}\,, \\
        \mathcal{M}_1^+&=\int_{fn}^{\infty} \varrho(n_A)(n_A-fn)\dd{n_A}=\frac{1}{\sqrt{8\pi \abs{\beta''(n)}V}}\,, \\
        \mathcal{M}_2^+&=\int_{fn}^{\infty} \varrho(n_A)(n_A-fn)^2\dd{n_A}=\frac{1}{8\abs{\beta''(n)}V}\,.
    \end{split}
\end{align}

\subsection{Dimension analysis}
Recall that the dimensions of Hilbert space, $d_A$ and $d_B$, depend on the intrinsic variables $n$, $n_A$ and $f$. In this digression, we determine which dimension is larger in the different regimes of $n$, $n_A$ and $f$.

Consider the ratio
\begin{equation}
    \frac{d_A}{d_B}\propto \exp \left[V\left(f\beta\left(\frac{n_A}{f}\right)-(1-f)\beta\left(\frac{n-n_A}{1-f}\right)\right)\right].
\end{equation}
Define the exponent to be $Y=f\beta\left(\frac{n_A}{f}\right)-(1-f)\beta\left(\frac{n-n_A}{1-f}\right)$, which when expanded around the mean of the Gaussian $\overline{n}_A=fn$ gives
\begin{equation}
    \label{eq:Y_definition}
    Y=2\left(f-\frac{1}{2}\right)\beta(n)+2\beta'(n)(n_A-fn)+\frac{\left(f-\frac{1}{2}\right)\abs{\beta ''(n)}}{f(1-f)}\left(n_A-f n\right)^2 +O(n_A-fn)^3\,.
\end{equation}
The value of $n_A$ at which the dimensions are equal is exactly the root of $Y$, which we shall call $n_\text{crit}$. To linear order\footnote{See the argument following Eq.~(B7) from the main text for why the linear order is sufficient.}, we get that
\begin{equation}
    \label{eq:ncrit_definition}
    n_\text{crit}=fn-\frac{\left(f-\frac{1}{2}\right)\beta(n)}{\beta'(n)}\,.
\end{equation}
Expanding $Y$ now around $n_\text{crit}$, one finds that
\begin{equation}
    Y=2\beta'(n)(n_A-n_\text{crit})\,.
\end{equation}
From this simple result, we now know the sign of $Y$ in the different regimes. Specifically, it is positive whenever $n_A>n_\text{crit}$ and $n>n^*$, or $n_A<n_\text{crit}$ and $n<n^*$. It is negative whenever $n_A>n_\text{crit}$ and $n<n^*$, or $n_A<n_\text{crit}$ and $n>n^*$. These results are tabulated in Figure~\ref{fig:dimension_argument}, along with plots showing the dimensions as a function of $n_A$.

\begin{figure}
\begingroup
\def\arraystretch{1.8}
\centering
\begin{tikzpicture}
    \node at (0,0) {
        \centering
        \setlength\tabcolsep{0pt}
        \begin{tabular}{M{1.4cm}|M{1.8cm} M{1.8cm}}
        any $f$ & $n_A<n_\text{crit}$ & $n_A>n_\text{crit}$ \\
        \hline
        $n<n^*$ & \cellcolor{yellow(ncs)!10} $d_A<d_B$ & \cellcolor{yellow(ncs)!40}$d_A>d_B$ \\ 
        $n>n^*$ & \cellcolor{green!10}$d_A>d_B$ & \cellcolor{green!40}$d_A<d_B$ \\
        \end{tabular}};
    \node at (5.5,-0.5) {
        \centering
        \includegraphics{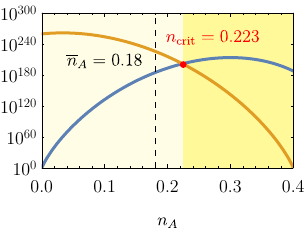}
    };
    \node at (5.5,-2.9) {(a) $n=0.4<\frac{2}{3}=n^*$};
    \node at (11.6,-0.5) {
        \centering
        \includegraphics{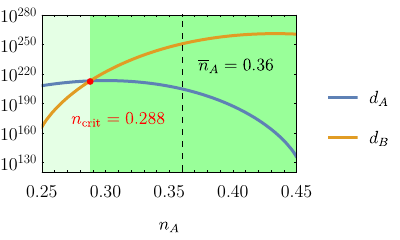}
    };
    \node at (11.1,-2.9) {
        (b) $n=0.8>\frac{2}{3}=n^*$
    };
\end{tikzpicture}
    \caption{A log-plot of the dimensions $d_A$ and $d_B$, one for (a) $n<n^*$ and (b) $n>n^*$, and an accompanying table displaying which dimension is larger in different regimes. The colours of the plot corresponding to the cells of the table indicate the region it represents. For the sake of an explicit example, the system used is hardcore bosons (2 species) with $f=0.45$ and $V=1000$ - see example b) of Figure~2 from the main text.}
    \label{fig:dimension_argument}
\endgroup
\end{figure}

In some situations we are also interested in the expansion of Eq.~\eqref{eq:Y_definition} around $n=n^*$:
\begin{equation}
    Y=2\left(f-\frac{1}{2}\right)\beta(n^*)-2\abs{\beta''(n^*)}(n_A-fn)(n-n^*)+O(n_A-fn)^2\,,
\end{equation}
where now $n_\text{crit}$ takes the form
\begin{equation}
    n_\text{crit}=fn+\frac{\left(f-\frac{1}{2}\right)\beta(n^*)}{\abs{\beta''(n^*)}(n-n^*)}\,.
\end{equation}

Lastly, we note that from Eq~\eqref{eq:Y_definition}, at the mean of the Gaussian $\overline{n}_A$, we have $Y<0$ if $f<\frac{1}{2}$ and $Y>0$ if $f>\frac{1}{2}$.

\subsection{Observable $\varphi(n_A)$}
Let us define
\begin{equation}
    \label{eq:varphi_1_2_definition}
    \begin{array}{l} \varphi_1(n_A)=-\min\left[\frac{d_A+1}{2d_B},\frac{d_B+1}{2d_A}\right]\\ \varphi_2(n_A)=\Psi(d_N+1)-\Psi(\max[d_A,d_B]+1) \end{array}\quad \text{so that} \quad \varphi=\varphi_1+\varphi_2\,.
\end{equation}
For large dimensions, which occur at large $V$, we can write $\varphi_1\approx -\frac{1}{2}\min(\frac{d_A}{d_B},\frac{d_B}{d_A})$. This term contains an exponential, which is suppressed unless the argument of the exponential vanishes, which is when $f=\frac{1}{2}$ and $n=n^*$. Thus 
\begin{equation}
    \varphi_1=-\frac{1}{2}\delta_{f,\frac{1}{2}}\delta_{n,n^*}\,.
\end{equation}
This argument is made rigorous in Appendix~B.1 from the main text.

For large $V$, we have $\varphi_2=\ln\left(\min\left[\frac{d_N}{d_A},\frac{d_N}{d_B}\right]\right)+o(1)$. Since we are integrating $\varphi$ against $\varrho$, we will want to expand $\varphi_2$ around $\overline{n}_A$:
\begin{equation}
    \varphi_2(n_A)=\begin{cases}\beta(n)fV+\frac{\ln(1-f)}{2}+\left(\beta'(n)V+O(1)\right)(n_A-fn)+\left(\frac{\abs{\beta''(n)}}{2(1-f)}V+O(1)\right)(n_A-fn)^2+O(n_A-fn)^3, & f<\frac{1}{2} \\
    \beta(n)(1-f)V+\frac{\ln(f)}{2}-\left(\beta'(n)V+O(1)\right)(n_A-fn)+\left(\frac{\abs{\beta''(n)}}{2f}V+O(1)\right)(n_A-fn)^2+O(n_A-fn)^3, & f>\frac{1}{2}
    \end{cases}\,.
\end{equation}

\subsection{Integral}
Finally, putting the results of the previous sections together, we want to evaluate
\begin{equation}
    \langle S_A\rangle_N=\int \varrho(n_A)(\varphi_1(n_A)+\varphi_2(n_A))\dd{n_A}\,.
\end{equation}
Except for the case $f=\frac{1}{2}$, we can use the moments Eq.~\eqref{eq:varrho_moments} to evaluate the average. For $f<\frac{1}{2}$, we have
\begin{align}
    \left\langle S_A\right\rangle_N&=\left(\beta(n)fV+\frac{\ln(1-f)}{2}\right)\mathcal{M}_0+\frac{\abs{\beta''(n)}}{2(1-f)}V\mathcal{M}_2 \nonumber \\
    &=\beta(n)fV+\frac{\ln(1-f)}{2}+\frac{\abs{\beta''(n)}}{2(1-f)}V\frac{f(1-f)}{\abs{\beta''(n)}V} \nonumber \\
    &=\beta(n)fV+\frac{f+\ln(1-f)}{2}\,,
\end{align}
and a similar calculation holds for $f>\frac{1}{2}$.

For $f=\frac{1}{2}$, the non-analycity of $\varphi_2$ lies exactly on the mean of the Gaussian, which means we need to use the half-moments instead.
\begin{align}
    \braket{S_A}_N&=\left(\frac{\beta(n)}{2}V+\frac{\ln(\frac{1}{2})}{2}\right)\mathcal{M}_0^-+\beta'(n)V\mathcal{M}_1^-+\abs{\beta''(n)}V\mathcal{M}_2^- \nonumber \\
    &\phantom{={}}+\left(\frac{\beta(n)}{2}V+\frac{\ln(\frac{1}{2})}{2}\right)\mathcal{M}_0^+-\beta'(n)V\mathcal{M}_1^++\abs{\beta''(n)}V\mathcal{M}_2^+-\frac{1}{2}\delta_{f,\frac{1}{2}}\delta_{n,n^*} \nonumber \\
    &=\beta(n)fV-\frac{\abs{\beta'(n)}}{\sqrt{2\pi\abs{\beta''(n)}}}\sqrt{V}+\frac{1}{2}\left(\frac{1}{2}+\ln(\frac{1}{2})-\delta_{f,\frac{1}{2}}\delta_{n,n^*}\right)\,.
\end{align}
Then we can write the full result, valid for $f\leq\frac{1}{2}$, as
\begin{equation}
    \braket{S_A}_N=\beta(n)fV-\frac{\abs{\beta'(n)}}{\sqrt{2\pi\abs{\beta''(n)}}}\sqrt{V}\delta_{f,\frac{1}{2}}+\frac{1}{2}\left(f+\ln(1-f)-\delta_{f,\frac{1}{2}}\delta_{n,n^*}\right)\,,
\end{equation}
which is the result of Eq.~(23) from the main text.

\newpage
\section{Details: Derivation of standard deviation $\Delta S_A$}\label{app:details-standard-deviation}
The variance of the entanglement entropy has the exact form
\begin{equation}
    \label{eq:variance_sum_definition_r}
    \left(\Delta S_A\right)^2_N=\frac{1}{d_N+1}\left[\sum_{N_A}\varrho_{N_A}\left(\varphi_{N_A}^2+\chi_{N_A}\right)-\left(\sum_{N_A}\varrho_{N_A}\varphi_{N_A}\right)^2\right]\,,
    \end{equation}
with $\varrho_{N_A}$ and $\varphi_{N_A}$ defined in Eq.~(19) of the main text and $\chi_{N_A}$ defined by
\begin{equation}
    \label{eq:varchi_definition}
    \chi_{N_A}=\begin{cases}
        (d_A+d_B)\Psi'(d_B+1)-(d_N+1)\Psi'(d_N+1)-\frac{(d_A-1)(d_A+2d_B-1)}{4d_B^2}, \quad d_A\leq d_B \\ 
        (d_A+d_B)\Psi'(d_A+1)-(d_N+1)\Psi'(d_N+1)-\frac{(d_B-1)(d_B+2d_A-1)}{4d_A^2}, \quad d_A> d_B
    \end{cases},
\end{equation}
where $\Psi'(x)=\dv[2]{x}\ln(\Gamma(x))$ is the derivative of the digamma function. In the following discussion we drop the $n_A$ dependence of $\varrho$ and $\varphi$ and the differential $\dd{n_A}$ for notational convenience. For example, $\int \varrho$ is understood to be $\int \varrho(n_A)\dd{n_A}$.

It is easier to evaluate Eq.~\eqref{eq:variance_sum_definition_r} by evaluating the sums $\sum\varrho\varphi^2$, $\sum\varrho \chi$ and $\left(\sum\varrho\varphi\right)^2$ separately.

For the rest of the section, we assume that $f\leq \frac{1}{2}$ unless otherwise stated.

\subsection{Evaluating the sum $\sum \varrho\varphi^2$}
We can further break up the sum into three components $\sum\varrho\varphi_1^2$, $\sum\varrho\varphi_2^2$ and $\sum\varrho\varphi_1\varphi_2$, with $\varphi_1$ and $\varphi_2$ from Eq.~\eqref{eq:varphi_1_2_definition}. The first two components may be approximated by an integral, though the third does not allow for such a simplification.

For $f<\frac{1}{2}$, we always have $d_A(\overline{n}_A)<d_B(n-\overline{n}_A)$ according to Figure~\ref{fig:dimension_argument}. In a neighbourhood of $\overline{n}_A$, we have
\begin{align}
    \begin{split}
        \varphi_1(n_A)&=-\frac{1}{2}\frac{d_A}{d_B}\,, \\
        \varphi_2(n_A)&=\ln\left(\frac{d_N}{d_B}\right)+o(1)\,.
    \end{split}
\end{align}

\subsubsection{Integrating $\int \varrho\varphi_1^2$}
We see that $\varphi_1^2\propto \exp\left[-2V\abs{Y}\right]$, with $Y$ defined as in Eq.~\eqref{eq:Y_definition}, the integral $\int\varrho\varphi_1^2$ contributes only when $f$ is close to $\frac{1}{2}$ and $n$ is close to $n^*$. Expanding in a Taylor series and evaluating the integral gives
\begin{equation}
    \int\varrho\varphi_1^2=\frac{1}{4}\delta_{f,\frac{1}{2}}\delta_{n,n^*}+o(1)\,,
\end{equation}
so the term is $o(V)$ anyways.

\subsubsection{Integrating $\int\varrho\varphi_2^2$}
In a similar procedure to Appendix~\ref{app:derivative_of_average_entropy}, we can Taylor expand $\varphi_2^2$ around $\overline{n}_A=fn$ to get
\begin{align}
\begin{split}
    \varphi_2^2&=\left(\beta(n)fV+\frac{1}{2}\ln(1-f)\right)^2+\left(2\beta(n)\beta'(n)fV^2+o(V^2)\right)(n_A-fn) \\
    &\phantom{={}} +\left(\left(\beta'(n)^2+\frac{f\beta(n)\abs{\beta''(n)}}{1-f}\right)V^2+o(V^2)\right)(n_A-fn)^2+O(n_A-fn)^3\,.
\end{split}
\end{align}
By using the moments listed in Eq.~\eqref{eq:varrho_moments}, we have that
\begin{equation}
    \int\varrho\varphi_2^2=\beta(n)^2f^2V^2+\left(\beta(n)f(f+\ln(1-f))+\frac{f(1-f)\beta'(n)^2}{\abs{\beta''(n)}}\right)V+o(V)\,.
\end{equation}

When $f=\frac{1}{2}$, we find that $n_\text{crit}$ lies exactly on the mean of the Gaussian. Thus we must Taylor expand $\varphi_2^2$ on either side of $n_\text{crit}$, where now the linear term in the expansion will contribute to a new term. From the table in Figure~\ref{fig:dimension_argument}, we need to treat the cases $n<n^*$ and $n>n^*$ separately. This new term is given by
\begin{equation}\hspace{-6mm}
    \begin{cases}
        2\beta(n)\beta'(n)fV^2\mathcal{M}_1^--2\beta(n)\beta'(n)(1-f)V^2\mathcal{M}_1^+, & n<n^* \\
        -2\beta(n)\beta'(n)(1-f)V^2\mathcal{M}_1^-+2\beta(n)\beta'(n)fV^2\mathcal{M}_1^+, & n>n^*
    \end{cases}=\begin{cases}
            -\frac{\beta(n)\beta'(n)}{\sqrt{2\pi\abs{\beta''(n)}}}V^{\frac{3}{2}}, & n<n^* \\ 
            \frac{\beta(n)\beta'(n)}{\sqrt{2\pi\abs{\beta''(n)}}}V^{\frac{3}{2}}, & n>n^* 
    \end{cases} 
    =-\frac{\beta(n)\abs{\beta'(n)}}{\sqrt{2\pi\abs{\beta''(n)}}}V^{\frac{3}{2}}\,.
\end{equation}
Then the final expression is
\begin{equation}
    \int\varrho\varphi_2^2=\beta(n)^2f^2V^2-\frac{\beta(n)\abs{\beta'(n)}}{\sqrt{2\pi\abs{\beta''(n)}}}\delta_{f,\frac{1}{2}}V^\frac{3}{2}+\left(\beta(n)f(f+\ln(1-f))+\frac{f(1-f)\beta'(n)^2}{\abs{\beta''(n)}}\right)V+o(V)\,.
\end{equation}

\subsubsection{Summing $\sum \varrho\varphi_1\varphi_2$}
We use the saddle-point approximation for $\varphi_1\varphi_2$ to get
\begin{align}
\label{eq:expand_phi1phi2}
    \varphi_1\varphi_2&=\begin{cases}
        \frac{1}{2} e^{VY}\sqrt{\frac{(1-f) \beta''\left(\frac{n_A}{f}\right)}{f \beta ''\left(\frac{n-n_A}{1-f}\right)}}  \left(\frac{1}{2} \ln
   \left(\frac{(1-f) \beta''(n)}{\beta''\left(\frac{n-n_A}{1-f}\right)}\right)-(1-f) V \beta
   \left(\frac{n-n_A}{1-f}\right)+V\beta (n)\right), & n_A<n_\text{crit} \\
   \frac{1}{2} e^{-VY}\sqrt{\frac{f \beta''\left(\frac{n-n_A}{1-f}\right)}{(1-f) \beta ''\left(\frac{n_A}{f}\right)}} \left(\frac{1}{2} \ln
   \left(\frac{f \beta''(n)}{\beta''\left(\frac{n_A}{f}\right)}\right)-f V \beta
   \left(\frac{n_A}{f}\right)+V\beta (n)\right), & n_A>n_\text{crit}
    \end{cases} \nonumber\\
    &=\begin{cases}
        \frac{1}{2}e^{VY}\left(\frac{1}{2}(V \beta (n)-\ln (2))+ \left(V \beta'(n)+\frac{\beta ^{(3)}(n) (V \beta (n)+1-\ln (2))}{\beta''(n)}\right)\left(n_A-fn\right)\right), & n_A<n_\text{crit} \\
        \frac{1}{2}e^{-VY}\left(\frac{1}{2}(V \beta (n)-\ln (2))- \left(V \beta'(n)+\frac{\beta ^{(3)}(n) (V \beta (n)+1-\ln (2))}{\beta''(n)}\right)\left(n_A-fn\right)\right), & n_A>n_\text{crit} \\
    \end{cases}\,.
\end{align}
We expand the above expression in terms of $n_A$ around $f n$, but will set $f=\frac{1}{2}$ to first order.

\begin{align}
    \sum_{N_A}\varrho \varphi_1\varphi_2&=\sum_{N_A=0}^{\overline{N}_A}\frac{1}{V}\varrho \varphi_1\varphi_2+\sum_{N_A=\overline{N}_A+1}^{N}\frac{1}{V}\varrho\varphi_1\varphi_2
\end{align}
We work on the first term:
\begin{align}
    \sum_{N_A=0}^{f n V}\frac{1}{V}\varrho(\tfrac{N_A}{V})\varphi_1(\tfrac{N_A}{V})\varphi_2(\tfrac{N_A}{V})&=\sum_{N_A=0}^{f n V}\sqrt{\frac{|\beta''(n)|}{2\pi V}}e^{-2V\left(|\beta''(n)|(\frac{N_A}{V}-fn)^2-\beta'(n)(\frac{N_A}{V}-fn)\right)}(C+D(\frac{N_A}{V}-fn)) \nonumber \\ 
    &=\sum_{N_A=0}^{f n V}\sqrt{\frac{|\beta''(n)|}{2\pi V}}e^{-2V\left(\frac{1}{V^2}|\beta''(n)|(N_A-fnV)^2-\frac{1}{V}\beta'(n)(N_A-fnV)\right)}(C+\frac{D}{V}(N_A-fnV)) \nonumber \\ 
    &=\sum_{k=-f n V}^{0}\sqrt{\frac{|\beta''(n)|}{2\pi V}}e^{-2\left(\frac{1}{V}|\beta''(n)|k^2-\beta'(n)k\right)}(C+\frac{D}{V}k)\,,
\end{align}
where the constants $C$ and $D$ can be read off from Eq.~\eqref{eq:expand_phi1phi2}. The leading order term in $V$ is given by
\begin{equation}
    \sqrt{\frac{|\beta''(n)|}{2\pi V}}C\sum_{k=-f n V}^{0}e^{2\beta'(n)k}=\sqrt{\frac{|\beta''(n)|}{2\pi V}}C\frac{e^{2\beta'(n)}}{e^{2\beta'(n)}-1}=\sqrt{\frac{|\beta''(n)|}{2\pi}}\frac{1}{2}\beta(n)\frac{e^{2\beta'(n)}}{e^{2\beta'(n)}-1}\sqrt{V}\nonumber\,.
\end{equation}
Similarly, the second term is
\begin{align}
    \sum_{N_A=fnV+1}^{nV}\frac{1}{V}\varrho(\tfrac{N_A}{V})\varphi_1(\tfrac{N_A}{V})\varphi_2(\tfrac{N_A}{V})&=\sum_{N_A=fnV+1}^{nV}\sqrt{\frac{\abs{\beta''(n)}}{2\pi V}}e^{-2V\left(|\beta''(n)|(\frac{N_A}{V}-fn)^2+\beta'(n)(\frac{N_A}{V}-fn)\right)}(C-D(\frac{N_A}{V}-fn)) \nonumber \\
    &=\sum_{N_A=fnV+1}^{nV}\sqrt{\frac{\abs{\beta''(n)}}{2\pi V}}e^{-2\left(\frac{1}{V}\abs{\beta''(n)}(N_A-fnV)^2+\beta'(n)(N_A-fnV)\right)}(C-\frac{D}{V}(N_A-fnV)) \nonumber \\
    &=\sum_{k=1}^{(1-f)n V}\sqrt{\frac{\abs{\beta''(n)}}{2\pi V}}e^{-2\left(\frac{1}{V}\abs{\beta''(n)}k^2+\beta'(n)k\right)}(C-\frac{D}{V}k)\,,
\end{align}
whose sum can be evaluated and yields
\begin{align}
    \sqrt{\frac{\abs{\beta''(n)}}{2\pi}}\frac{1}{2}\beta(n)\frac{1}{e^{2\beta'(n)}-1}\sqrt{V}\,.
\end{align}
The final result, combining the two calculations above, is 
\begin{equation}
    \sum \varrho\varphi_1\varphi_2=\sqrt{\frac{\abs{\beta''(n)}}{2\pi}}\frac{1}{2}\beta(n)\frac{e^{2\beta'(n)}+1}{e^{2\beta'(n)}-1}\sqrt{V}+O(1)\,.
\end{equation}

When $n=n^*$, the discontinuity of the series expansion of $\varphi_1\varphi_2$ coincides with the mean of the Gaussian. Just as we did previously, we need to integrate on both sides of the discontinuity by multiplying with the half-moments of Eq.~\eqref{eq:varrho_half_moments}. We calculate the leading order to be $\frac{1}{4}\beta(n)$, so
\begin{equation}
    \sum \varrho\varphi_1\varphi_2=\frac{1}{4}\beta(n)V\delta_{f,\frac{1}{2}}\delta_{n,n^*}+o(V)\,.
\end{equation}
In summary, the term $\sum\varrho\varphi_1\varphi_2$ is $o(V)$ unless $f=\frac{1}{2}$ and $n=n^*$, at which it is of order $V$.

Collecting the results of the previous three calculations yields
\begin{equation}
    \sum \varrho\varphi^2=\beta(n)^2f^2V^2-\frac{\beta(n)\abs{\beta'(n)}}{\sqrt{2\pi\abs{\beta''(n)}}}\delta_{f,\frac{1}{2}}V^\frac{3}{2}+\left(\beta(n)f(f+\ln(1-f))+\frac{f(1-f)\beta'(n)^2}{\abs{\beta''(n)}}-\frac{1}{2}\beta(n)\delta_{f,\frac{1}{2}}\delta_{n,n^*}\right)V+o(V)\,.
\end{equation}

\subsection{Evaluating the sum $\sum\varrho \chi$}
We use the expansion $\Psi'(x)=\frac{1}{x}+O(\frac{1}{x^2})$ in Eq.~\eqref{eq:varchi_definition} to see that
\begin{align}
    \chi(n_A)=\begin{cases}
        \frac{d_A}{2d_B}+O\left(\frac{1}{d_B^2}\right), & d_A< d_B \\
        \frac{1}{4}+o(1), & d_A=d_B \\ 
        \frac{d_B}{2d_A}+O\left(\frac{1}{d_A^2}\right), & d_A> d_B
    \end{cases}\,.
\end{align}
One sees therefore that $\sum\varrho\chi=o(V)$ since $\sum\varrho=1$.

\subsection{Evaluating the sum $(\sum\varrho \varphi)^2$}
This is just the square of our main result Eq.~(23) of the main text, so we have
\begin{equation}
    \left(\sum \varrho\varphi\right)^2=\beta(n)^2f^2V^2-\frac{\beta(n)\abs{\beta'(n)}}{\sqrt{2\pi\abs{\beta''(n)}}}\delta_{f,\frac{1}{2}}V^{\frac{3}{2}}+\left(\beta(n)f(f+\ln(1-f))+\frac{\beta'(n)^2}{2\pi \abs{\beta''(n)}}\delta_{f,\frac{1}{2}}-\frac{1}{2}\beta(n)\delta_{f,\frac{1}{2}}\delta_{n,n^*}\right)V+o(V)\,.
\end{equation}

\subsection{Variance}
We now have all the ingredients to evaluate Eq.~\eqref{eq:variance_sum_definition_r}. We have
\begin{align}
    \sum \varrho(\varphi^2+\chi)-(\sum\varrho \varphi)^2 =\left(f(1-f)-\frac{1}{2\pi}\delta_{f,\frac{1}{2}}\right)\frac{\beta'(n)^2}{\abs{\beta''(n)}}V+o(V)\,,
\end{align}
and so
\begin{equation}
    \left(\Delta S_A\right)^2_N=\sqrt{\frac{2\pi}{\abs{\beta''(n)}}}\left(f(1-f)-\frac{1}{2\pi}\delta_{f,\frac{1}{2}}\right)\frac{\beta'(n)^2}{\abs{\beta''(n)}}V^{\frac{3}{2}}e^{-\beta(n)V}+o(e^{-\beta(n)V})\,,
\end{equation}
which is the result we presented in Eq.~(32) in the main text.